\documentclass[3p]{elsarticle}
\pdfoutput=1
\usepackage{amsmath} 
\usepackage{amsthm} 
\usepackage{amssymb}
\usepackage{mathrsfs}
\usepackage{xcolor}

\usepackage{cellspace}

\newcommand{\sgn}{\mathrm{sgn}}

\newcommand{\ket}[1]{\left| #1 \right>} 
\newcommand{\bra}[1]{\left< #1 \right|} 
\newcommand{\braket}[2]{\left< #1 \vphantom{#2} \right|
	\left. #2 \vphantom{#1} \right>}

\newcommand{\vv}[1]{\ensuremath{\mathbf{#1}}} 
\newcommand{\gv}[1]{\ensuremath{\boldsymbol{#1}}}

\newcommand{\avg}[1]{\left< #1 \right>} 

\newtheorem{thm}{Theorem}
\newtheorem*{lem}{Lemma}
\newdefinition{rmk}{Remark}

\begin{document}

\begin{frontmatter}

\title{$\mathbb{Z}_k^{(r)}$-Algebras, FQH Ground States, and Invariants of Binary Forms}

\author{Hamed Pakatchi}
\ead{pakatchi@alum.mit.edu}
\affiliation{organization={Massachusetts Institute of Technology}, addressline={77 Massachusetts Avenue}, 	city={Cambridge}, postcode={02139}, state={MA}, country={USA}}

	\begin{abstract}
		A prominent class of model FQH ground states is those realized as correlation function of $\mathbb{Z}_k^{(r)}$-algebras. In this paper, we study the interplay between these algebras and their corresponding wavefunctions. In the hopes of realizing these wavefunctions as a unique densest zero energy state, we propose a generalization for the projection Hamiltonians. Finally, using techniques from invariants of binary forms, an ansatz for computation of correlations $\avg{\psi(z_1)\cdots\psi(z_{2k})}\prod_{i<j}(z_i-z_j)^{2r/k}$ is devised. We provide some evidence that, at least when $r=2$, our proposed Hamiltonian realizes $\mathbb{Z}_k^{(2)}$-wavefunctions as a \emph{unique} ground state.
	\end{abstract}

\end{frontmatter}

\section{Introduction}
Inspired by Laughlin's seminal work \cite{Laughlin_Original}, analyzing trial wavefunctions has been one of the main tools to research fractional quantum Hall (FQH) systems. A promising method to construct and study model FQH wavefunctions, especially non-Abelian ones \cite{wen1991non, Moore-Read}, is to realize them as conformal blocks in particular rational conformal field theories (RCFT). Moore-Read pioneered this approach \cite{Moore-Read} via the construction of the Pfaffian state (aka Moore-Read state) based on the Ising CFT. The Pfaffian state and its generalizations, namely the $\mathbb{Z}_k$ Read-Rezayi (RR) states \cite{Read-Rezayi}, now lay the foundation of our understanding of non-Abelian states.

Concretely, in the CFT point of view toward \emph{bosonic} quantum Hall states, one first recognizes the `electron' as a \emph{bosonic} chiral vertex operator $V_e(z)$ in an appropriate CFT. For our purposes, the underlying CFT is always the direct product of a parafermionic theory and a free boson theory. In an FQH system with filling fraction $\nu$, a sensible choice would be $V_e(z)=\psi(z):\exp[i\phi(z)/\sqrt{\nu}]:$ \cite{Moore-Read}. Here, $\psi$ is the parafermion that generates the cyclic symmetry (see below), $\phi$ is a free chiral boson, and $:\: :$ symbolizes normal ordering. The trial wavefunction for a quantum Hall droplet on the plane or the sphere is then the following correlation function:
\begin{equation}
	\Psi(z_1, \cdots, z_{N})\propto
	\avg{\mathcal{O}_{\mathrm{bg}}(\infty)V_e(z_1)\cdots V_e(z_{N})} = \avg{\psi(z_1)\psi(z_2)\cdots \psi(z_{N})}\prod_{i<j}(z_i-z_j)^{1/\nu}
\end{equation}
The operator $\mathcal{O}_{\mathrm{bg}}$ represents the neutralizing background charge that we put at infinity. With this convention about the neutralizing background, the wavefunction $\Psi$ becomes a translational invariant homogeneous symmetric polynomial. Ideally, one has access to a (pseudopotential) Hamiltonian \cite{Haldane_Sphere_Pseudo, Simon_Pseudopotentials}, describing a particular FQH phase of matter on a disk/sphere, and the manufactured wavefunction $\Psi$ would be its \emph{unique} densest zero-energy eigenstate (i.e., the ground state). However, aside from a handful of examples, as influential as those rare cases may be, designing models satisfying the `uniqueness' condition has proven difficult in the literature. Therefore, part of the current paper's mission would be to develop modified Hamiltonians with the uniqueness issue in mind.

In practice, it is sometimes more feasible to procure a ``nice'' polynomial with specific local properties and retroactively engineer a (pseudopotential) Hamiltonian fitting those local properties. As a basic example, consider a wavefunction (i.e., symmetric polynomial) $\Psi(z_1, \cdots, z_N)$ that vanishes of order $2r$ when $k+1$ particles are fused but does not vanish when fewer particles are fused. This local property automatically makes $\Psi$ the zero-energy state to the projection Hamiltonian $\mathscr{H}_{k+1}^{2r}$ that projects out any state where any cluster of $k + 1$ particles has relative angular momentum $<2r$ (see Ref. \cite{Simon_projection} for more details on projection Hamiltonians). We want to strengthen this vanishing property. To do so, let $\Psi_n$ denote the ground state wavefunction of a particular FQH phase when prepared with $N=nk$ electrons. A more stringent variant of the vanishing property would then be:
\begin{equation}
	\Psi_{n}(w^{\times k}, z_{k+1}, \cdots, z_{nk})=
	\prod_{i=k+1}^{nk}(w-z_i)^{2r}\Psi_{n-1}(z_{k+1}, \cdots, z_{nk})
\end{equation}
Following the authors in Ref. \cite{estienne2010clustering}, we say the wavefunction is a $(k,2r)$ \emph{clustering polynomial} if it is a conformally covariant polynomial satisfying this local property. Superficially, clustering states are appealing because if one takes a $k$-cluster of electrons to infinity, one recovers the `same' wavefunction only with fewer particles. On a classification level, as pointed out in Refs. \cite{Pattern_of_Zeros, wen2008topological, barkeshli2009structure}, the CFT approach and clustering polynomials are closely tied.

To further comment on clustering properties and introduce some of the key players in the current paper, we need to glance at the underlying parafermionic CFTs. The relevant parafermionic CFTs have a cyclic fusion algebra, say $\mathbb{Z}_k$. They consist of $k$ primary fields $\psi_0\equiv 1,\psi_1\equiv \psi, \psi_2, \cdots, \psi_{k-1}$ and the simple fusion rules $\psi_a\star \psi_b=\psi_{a+b\pmod {k}}$. The parafermion $\psi\equiv \psi_1$ is the generator of the cyclic symmetry under the fusion rules. The cyclic nature of the fusion algebra is responsible for (generally weaker versions of) the clustering property of the corresponding FQH wavefunction. Translating the fusion rules into an operator product algebra (OPA), we have (Notation: $\psi_a^+\equiv \psi_{k-a}$):
\begin{subequations}
	\begin{align}
		&\psi_a(z)\psi_b(w)(z-w)^{h_{a}+h_b-h_{a+b}}=C_{a,b}\psi_{a+b}(w)+\cdots &(0< a+b<k)\\
		&\psi^+_{a}(z)\psi^+_{b}(w)(z-w)^{h_a+h_b-h_{a+b}}=C_{a,b}\psi^+_{a+b}(w)+\cdots &(0< a+b<k)\\
		&\psi_a(z)\psi_a^+(w)(z-w)^{2h_a}= 1+ 
		{\textstyle\frac{2h_a}{c}}(z-w)^2L(w)+\cdots
	\end{align}
\end{subequations}
Here, $C_{a,b}$ are the structure constants, $c$ is the central charge and $L$ is the energy-momentum tensor. Moreover, the ellipsis only involves fields with a higher scaling dimension than the leading field by an integer. In truth, with $\mathcal{W}$ being the underlying chiral algebra, all subsequent fields are $\mathcal{W}$-descendants of the leading term. There exist many associate solutions to the above OPA. In this paper, we are only interested in solutions with scaling dimensions $h(\psi_a)\equiv h_a=a(k-a)/2\nu$, where $\nu=k/2r$ (with integer $k,r$); the symbol $\nu$ can be identified as the filling fraction of the corresponding FQH wavefunction. Such a theory will be called a \emph{$\mathbb{Z}_k^{(r)}$ (current) algebra}. The special solution with $r=1$ (i.e. $\mathbb{Z}_k^{(1)}$-algebra) is the classic construction of Zamolodchikov-Fateev \cite{Zamolodchikov_Fateev_Parafermion} and $\mathbb{Z}_k^{(1)}$ wavefunctions are the Read-Rezayi states \cite{Read-Rezayi}. Thus, $\mathbb{Z}_k^{(r)}$ states are the generalization of Read-Rezayi states.

Note that the electron vertex operator $V_e(z)=\psi(z):\exp[i\sqrt{2r/k}\:\phi(z)]:$, with $\psi$ being the generator of a $\mathbb{Z}_k^{(r)}$-algebra, has the statistics of a boson $V_e(z)V_e(w)=V_e(w)V_e(z)$, and an integer scaling dimension equal to $r$. To respect the spin-statistics theorem, we have disallowed half-integer values for $r$. However, this choice would exclude states like Gaffnian \cite{Gaffnian} from the models (i.e., $k=2, r=3/2$), and not all authors would opt for such an omission. For reference, what we named a $\mathbb{Z}_k^{(r)}$-algebra would be called a $\mathbb{Z}_k^{(2r)}$ theory in Ref. \cite{estienne2010clustering}, and they do include cases like Gaffnian as well. Regardless of allowing half-integers, as is pointed out in Ref. \cite{estienne2010clustering}, using the operator product expansions (of the form $\psi_a(z)\psi(w)$) repeatedly, the wavefunction
\begin{equation}
	\Psi_n (z_1, \cdots, z_{nk}) \equiv \frac{1}{(\prod_{a=0}^{k-1}C_{1,a})^n}\avg{\psi(z_1)\cdots \psi(z_{nk})}\prod_{i<j}(z_i-z_j)^{2r/k}
\end{equation}
can easily be seen as a $(k,2r)$ clustering polynomial. Here, the subscript $n$ refers to the size of the system and reminds us that we have $N=nk$ particles. The current paper will unravel a secondary local property of these wavefunctions called \emph{separability}. The new property is concerned with the explicit form of the projection of the wavefunction when a cluster of $k+1$ particles has a relative angular momentum of exactly $2r$. As discussed in the summary section, we use separability to design modified Hamiltonians that realize these wavefunctions as likely the unique, densest zero-energy state.

In the literature, initially, the main motivation to investigate $\mathbb{Z}_k^{(r)}$-algebras was two conjectures about (negative rational)  $(k,\rho)$ admissible Jack polynomials \cite{Feigin}, and their role as trial wavefunctions \cite{Bernevig_Haldane_Model, Bernevig_Haldane_Cluster}. The Jacks are an interesting avenue for generalizing  Read-Rezayi states to non-Abelian states with more general filling fractions. In their debut as trial wavefunctions \cite{Bernevig_Haldane_Model}, it was conjectured that these Jacks are clustering polynomials. As for the second conjecture, before coming to prominence as trial FQH states, these specialized Jack polynomials were first studied in Ref. \cite{Feigin} as a means to describe the space of symmetric polynomials that vanish as $k+1$ variables coalesce. In the same paper, and later on in Ref. \cite{feigin2003symmetric}, it was conjectured that these Jacks are certain correlation functions in the $\mathrm{WA}_{k-1}(k+1, k+\rho)$ minimal models. However, $\mathrm{WA}_{k-1}(k+1, k+\rho)$ is a (non-unitary) special case of $\mathbb{Z}_k^{\rho/2}$-algebras. Thus, the former conjecture is a consequence of the latter one. In Ref. \cite{estienne2009relating}, by studying the degenerate representations of $\mathrm{WA}_{k-1}(k+1, k+\rho)$, both conjectures were proved. We should also mention that the conjecture regarding the clustering property of Jacks was also proved independently via the representation theory of Cherednik algebras in Ref. \cite{zamaere}.

In general, FQH polynomials often have hidden structures, historically leading to new physical insight into the quantum Hall effect. For example, we have already mentioned how researchers use trial wavefunctions to reverse engineer (pseudopotential) Hamiltonians. In principle, one may also use such wavefunctions to reconstruct the underlying CFT (with various degrees of success). Assuming a given wavefunction is indeed a conformal block, one can obtain the central charge from the wavefunction alone \cite{bernevig2009central,estienne2010clustering, wen1994chiral}. In the current paper, we obtain further CFT data from these wavefunctions. To what degree and how these trial wavefunctions encode the current algebra inside them remains to be seen. However, it is safe to say that: as our technical ability to calculate these particular conformal blocks advances, our knowledge about these FQH states and their hidden structures will also improve. At the moment, aside from a handful of exceptions, the best tools currently available are only applicable situationally and even then usually only help us to compute the $2k$ point correlation function (the so-called principal wavefunction)
\begin{equation}
	\Psi_2 (z_1, \cdots, z_{2k}) \equiv \frac{1}{(\prod_{a=0}^{k-1}C_{1,a})^2}\avg{\psi(z_1)\cdots \psi(z_{2k})}\prod_{i<j}(z_i-z_j)^{2r/k}
\end{equation}
Even this limited version of the wavefunction requires special machinery to compute. To our knowledge, the only $\mathbb{Z}_k^{(r)}$-algebras with fully computed wavefunctions of arbitrary size are the following: Read-Rezayi states \cite{Read-Rezayi, cappelli} ($k$ arbitrary, $r=1$), Haffnian ($k=r=2$) \cite{Haffnian}, and correlation functions of $\mathcal{N}=1$ superconformal field theory ($k=2, r=3$) \cite{simon2008correlators}. As for the computations of $\Psi_2$ alone, in Ref. \cite{estienne2010clustering}, the principal wavefunction of $\mathbb{Z}_2^{(5)}$, $\mathbb{Z}_k^{(2)}$ with $k\leq 5$ are calculated in terms of Jack polynomials. The current paper provides additional machinery to compute principal wavefunctions by utilizing the so-called \emph{theory of invariants of binary forms}. In particular, we compute the principal $\mathbb{Z}_k^{(2)}$ wavefunction for all values of $k$ (section \ref{sec_quartic}) and the principal $\mathbb{Z}_2^{(r)}$ wavefunctions for all values of $r$ (\ref{appendix_direct}).

So far, we have not mentioned the topological nature of the quantum Hall phases of matter. A signature of quantum Hall systems is an energy gap between the ground state and the bulk excitations. Generally speaking, to ensure such a gap, the CFT describing the FQH system is expected to be \emph{unitary} and \emph{rational} \cite{read2009conformal}. As for concrete evidence behind the unitarity expectation, while outside the scope of models considered here, the Haldane-Rezayi (HR) \cite{haldane1988spin} and Gaffnian \cite{simon2007construction} are the two extensively studied examples: The CFT behind the Gaffnian is the $\mathcal{M}(5,3)$ minimal model, which has central charge $c=-\frac{3}{5}$, hence is non-unitary. As for the HR state, Ref. \cite{gurarie1997haldane} shows that it, too, is related to a non-unitary CFT. It has been argued that the HR state is the wavefunction at a phase transition \cite{read2000paired}, and thus is not the ground state of a gapped system (also see Ref. \cite{hermanns2011irrational}). In the case of the Gaffnian, the gaplessness was checked numerically (e.g., \cite{jolicoeur2014gaplessness}) and confirmed in Ref. \cite{estienne2015correlation} using Matrix Product State (MPS) formalism. In addition to non-unitarity concerns, there is also the question of \emph{rationality} of the obtained theories. As a simple non-example, consider the $\mathbb{Z}_2^{(2)}$-algebra. The $\mathbb{Z}_2^{(2)}$ wavefunction, dubbed the Haffnian, has been studied in detail by Green \cite{Haffnian} (also see Ref. \cite{hermanns2011irrational}). In this case, $\psi(z)=i\partial \phi(z)$, with $\phi(z)$ a chiral free boson. It is well-known that for each $\alpha\in \mathbb{R}$ the vertex operator $V_\alpha(z)=:\exp[i\alpha \phi(z)]:$ is a $\psi$-primary field. Hence, the spectrum is infinite and this CFT describes a gapless phase. As for the general current algebras, given ``large enough'' $k,r$, the central charge in a generic $\mathbb{Z}_k^{(r)}$-algebra is a free parameter and can assume irrational values. Given the brief discussion about the gap and rationality, we end this section by defining the scope and limitations of the current paper. It is well beyond the reach of the current paper to find what values of central charge (and possibly other free CFT data) would make the CFTs rational. In addition, the CFTs studied here include many non-unitary theories, requiring caution in the context of the FQH systems. Moreover, we lack the tools to adequately analyze the existence of a gap (or lack thereof) for the models discussed here.

\section{Motivation and Summary of Results}
The main ambition of this paper is to design a model Hamiltonian to realize the $\mathbb{Z}_k^{(r)}$ wavefunction (abbreviated WF) as \emph{the unique} ground state. More descriptively, we are searching for a multi-particle pseudopotential Hamiltonian \cite{Simon_Pseudopotentials} for which a given $\mathbb{Z}_k^{(r)}$-WF is the \emph{unique} densest zero-energy state. This demand is necessary since quantum Hall Hamiltonians on a disk/sphere must have a unique ground state. Our approach can be summarized as follows:
\begin{enumerate}
	\item We \emph{claim} that the entirety of a $\mathbb{Z}_k^{(r)}$-algebra is encoded inside a \emph{single} (explicit) polynomial $\chi$.
	\item We construct a pseudopotential Hamiltonian $H_\chi$ (in a natural way) parametrized by this proxy $\chi$.
\end{enumerate}
The idea is that if both the current algebra and the model Hamiltonian describe the same quantum Hall phase of matter, then there should be a one-to-one correspondence between them. Accordingly, with $\chi$ as an intermediate, the respective model Hamiltonian would reflect the entirety of the current algebra in this methodology.

The two key players in the current paper are the \emph{minimal polynomial} $\chi$ (the presumed proxy for a $\mathbb{Z}_k^{(r)}$-algebra) and the \emph{principal} $\mathbb{Z}_k^{(r)}$-wavefunction. The definitions of these polynomials are as follows (Notation: $w^{\times a}$ means $w$ repeated $a$ times):
\begin{align}
	\Psi_2(z_1, \cdots, z_{2k})&=
	\frac{1}{g_k^2} \avg{\psi(z_1)\cdots \psi(z_{2k})}\prod_{i<j}(z_i-z_j)^{\frac{1}{\nu}}\\
	\chi(z_1, \cdots, z_{k+1})
	&=\lim_{w\to \infty} \frac{\Psi_2(w^{\times k-1}, z_1, \cdots, z_{k+1})}{w^{2r(k-1)}}
	= \frac{1}{g_k}
	\avg{\psi^+(\infty)\psi(z_1)\cdots\psi(z_{k+1})}\prod_{i<j}(z_i-z_j)^{\frac{1}{\nu}}
\end{align}
where, $g_0=1$, and for $1\leq a\leq k$ we have  $g_a= \prod_{j=0}^{a-1}C_{1,j}$. The minimal polynomial $\chi$ is a homogeneous, translational invariant, symmetric polynomial in $k+1$ variables and degree $2r$. Such a polynomial is called a $(k+1,2r)$ \emph{semi-invariant} and their linear space is denoted by $\mathcal{T}_{k+1}^{2r}$. Additionally, due to the particular normalization we choose for the wavefunctions, the minimal polynomials \emph{always} have a normalization $\chi(1,0^{\times k})=1$.

To present the Hamiltonian $H_\chi$, let us review a few concepts. We use the notation $\mathscr{H}_{k+1}^{2r-1}$ for the projection operator that projects out any state where any cluster of $k + 1$ particles has relative angular momentum $<2r$. Traditionally, this is the most prominent model Hamiltonian used in the literature. For $k=1$ (and $r$ arbitrary), this leads to Haldane's (gapped) Hamiltonian \cite{Haldane_Sphere_Pseudo} for Laughlin $2r$-state. For $r=1$ (and $k$ arbitrary), one recovers the Read-Rezayi (gapped) Hamiltonian \cite{Read-Rezayi}. Unfortunately, the fact that $\dim \mathcal{T}_{k+1}^{2r}>1$ for almost all $(k,r)$ pairs is a big obstacle for $\mathscr{H}_{k+1}^{2r-1}$ admitting a unique densest zero-energy state. We would like to modify $\mathscr{H}_{k+1}^{2r-1}$ in a natural way so that $\dim \mathcal{T}_{k+1}^{2r}>1$ is no longer relevant to uniqueness of ground states. The projection Hamiltonian $\mathscr{H}_{k+1}^{2r-1}$ is blind to the fine details of the projection to the sector where $(k+1)$-cluster have \emph{exactly} relative angular momentum $2r$. We assert that those fine details matter and look for a modification that includes them. The guiding observation is the following local property: denoting by $\mathcal{P}_{k+1}^{2r}(z_1, \cdots, z_{k+1})$ the projection to the sector where variables $z_1, \cdots, z_{k+1}$ have relative angular momentum equal to $2r$, we show in this paper that  $\mathbb{Z}_k^{(r)}$-WFs satisfy (Notation: $\widehat{z}=(z_1+\cdots+z_{k+1})/(k+1)$ is the center-of-mass)
\begin{equation}
	\mathcal{P}_{k+1}^{2r}(z_1, \cdots, z_{k+1})\Psi_n =\chi(z_1, \cdots, z_{k+1})\prod_{i>k+1}\left(\widehat{z}-z_i\right)^{2r}\Psi_{n-1}\left(\widehat{z}, z_{k+2}, \cdots, z_{nk}\right)
\end{equation}
The function $\chi$ is the minimal polynomial and is independent of the size $n$. This local property is called \emph{separability}. Utilizing separability, we propose the following pseudopotential Hamiltonian
\begin{equation}
	H_\chi = \mathscr{H}_{k+1}^{2r-1}+ V
	\sum_{i_1<i_2\cdots<i_{k+1}}
	\mathcal{P}_{\chi^\perp}\mathcal{P}_{k+1}^{2r}(z_{i_1}, \cdots, z_{i_{k+1}}) , \qquad (V>0)
\end{equation}
Treating $\chi$ as a vector in $\mathcal{T}_{k+1}^{2r}$, the operator $\mathcal{P}_{\chi^\perp}$ projects to the hyperplane in $\mathcal{T}_{k+1}^{2r}$ with normal vector $\chi$. Due to construction, among $\mathbb{Z}_k^{(r)}$-WFs, only those with minimal polynomial $\chi$ are the densest zero-energy states of $H_\chi$. While this modification resolves the $\dim \mathcal{T}_{k+1}^{2r}>1$ issue, we still need to study if $H_\chi$ has a unique ground state. We want to clearly state that we have no proof that $H_\chi$ possesses a unique densest zero-energy state. Instead, we shall provide indirect evidence supporting the idea that $H_\chi$ Hamiltonians possess a unique ground state.

We base our argument regarding the uniqueness of the ground state on a triplet of postulates. The following diagram summarizes these postulates:
\[
\chi\xrightarrow{\textsc{Fix}} \Psi_2\xrightarrow{\textsc{Fix}} \mathbb{Z}_k^{(r)}\text{algebra}\xrightarrow{\textsc{Fix}} \Psi_n
\]
We read this diagram as follows: (1) knowing $\chi$ (together with the fact that $\Psi_1=1$ and its consequences) is enough to completely determine $\Psi_2$; (2) Determining $\Psi_2$ is enough to fix all of the free parameters of the $\mathbb{Z}_k^{(r)}$-algebra that gives birth to $\Psi_2$; (3) No two \emph{distinct} $\mathbb{Z}_k^{(r)}$-algebras can give birth to \emph{identical} wavefunctions. Much of this paper is devoted to the justification (at least for $r=1,2$) of $\Psi_2\to \mathbb{Z}_k^{(r)}\text {algebra}$ and $\chi\to \Psi_2$. As $\mathbb{Z}_k^{(r)}\text {algebra}\to \Psi_n$ (for all $n$) is the least controversial yet hardest to check, we have postponed its study until future work.

Let us illustrate this entire program with a well-known, familiar, and relatively simple example: correlation functions in $\mathcal{N}=1$ superconformal field theory (SCFT). Recall that the following operator product expansions (OPEs) describe the $\mathcal{N}=1$ SCFT, aka $\mathcal{N}=1$ super Virasoro algebra:
\begin{align*}
	L(z)L(w)&=\frac{\frac{1}{2}c}{(z-w)^4}+\frac{2L(w)}{(z-w)^2}+
	\frac{\partial L(w)}{z-w}+\mathrm{regular}\\
	L(z)\psi(w)&=\frac{\frac{3}{2}\psi(w)}{(z-w)^2}+
	\frac{\partial \psi(w)}{z-w}+\mathrm{regular}\\
	\psi(z)\psi(w) &= \frac{1}{(z-w)^{3}} + \frac{3}{c}\frac{L(w)}{z-w}+\mathrm{regular}
\end{align*}
Here, $c$ is the central charge and can take any real value; the first OPE declares $L$ as the energy-momentum tensor, the second states that $\psi$ has scaling dimension $h=3/2$, and the last describes how $\psi$ fuses with itself. A quick check shows that this theory matches the definition of $\mathbb{Z}_2^{(3)}$-algebras exactly. While we do not often have the luxury of being able to compute wavefunctions of arbitrary size, fortunately, Simon \cite{simon2008correlators} has calculated the $\mathbb{Z}_2^{(3)}$-WFs $\Psi_A$ for an arbitrary number of particles $N=2n$ (Notation: $z_{ij}=z_i-z_j$):
\begin{align*}
	f_A(z_1, z_2, w_1, w_2)&\equiv A(z_1-w_1)^3(z_1-w_2)^3(z_2-w_1)^3(z_2-w_2)^3+
	(z_1-w_1)^4(z_1-w_2)^2(z_2-w_1)^2(z_2-w_2)^4\\
	\Psi_{A} (z_1, \cdots, z_{2n}) &\propto \mathscr{S} \prod_{1\leq r<s\leq n}^n f_{A} (z_{2r-1}, z_{2r}, z_{2s-1}, z_{2s})
\end{align*}
where $\mathscr{S}$ stands for symmetrization and $A=c/3-1$ is a free parameter. Before proceeding with the analysis, we should mention that, given the confines of this example, our methodology partially aligns with the arguments presented in Ref. \cite{jackson2013entanglement}, section II.A. To expand on the content of Ref. \cite{jackson2013entanglement}, let us compute the minimal polynomial:
\[
\chi_A(z_1, z_2, z_3) = \frac{1}{2(A+1)}\mathscr{S}[A(z_1-z_3)^3 (z_2-z_3)^3+(z_1-z_3)^4 (z_2-z_3)^2]
\]
Note that $\chi_1=\mathscr{S}[(z_1-z_3)^3 (z_2-z_3)^3]$ and $\chi_2=\mathscr{S}[(z_1-z_3)^4 (z_2-z_3)^2]$ constitute a basis for $\mathcal{T}_3^6$. Therefore, the knowledge of the value $A$ is entirely the same as knowing the minimal polynomial. In other words, by fixing $\chi_A$, we determine $A$, which fixes the CFT and all of the $\Psi_n$ wavefunctions. The Hamiltonian $H_{\chi_A}$, which Ref. \cite{jackson2013entanglement} also hints at, is then really a one-parameter family of Hamiltonians (depending on $A$), which picks $\Psi_A$ as the densest zero-energy state.

If $\chi\to \Psi_2\to \mathbb{Z}_k^{(r)}\text {algebra}\to\Psi_n$ is indeed true, then $\chi$ acts as a proxy for the $\mathbb{Z}_k^{(r)}$-algebra in the Hamiltonian $H_\chi$. In other words, no two current algebras will have the same minimal polynomial. Thus, \emph{among wavefunctions descending from a $\mathbb{Z}_k^{(r)}$-algebra}, the model Hamiltonian $H_\chi$ has exactly one zero-energy state. Consequently, unless $H_\chi$ has a densest zero-energy state that is not a $\mathbb{Z}_k^{(r)}$-WF for certain choice of free parameters, the Hamiltonian $H_\chi$ has a unique ground state. Whether such a non-CFT ground state exists is beyond the scope of this paper. Regarding the gap, we make another \emph{guess}: ``If a $\mathbb{Z}_k^{(r)}$-algebra is unitary and has a finite spectrum, and $\chi$ is a proxy for such a current algebra, then $H_\chi$ is gapped.''

\subsection{$\Psi_1=1$ and Evidence for $(\Psi_1,)\Psi_2\to \mathbb{Z}_k^{(r)}\text {algebra}$}
Upon fixing a $\mathbb{Z}_k^{(r)}$-algebra, one can construct its wavefunctions in $nk$ variables ($n\geq 1$ arbitrary integer). We denote this wavefunction by $\Psi_n$ and often work with the sequence $\gv{\Psi}=(\Psi_1, \Psi_2, \cdots, \Psi_n, \cdots)$.
Quite generally, if no two distinct $\mathbb{Z}_k^{(r)}$-algebras can lead to an identical wavefunction sequence, the first ``few'' elements in this sequence should contain all of the classifying information about the current algebra. To expand on this comment, let $\alpha$ be some free parameter of the CFT. If changing $\alpha$ does not alter any of the $\Psi_n$, we violate the correspondence between current algebras and wavefunctions. Thus, there exist a smallest $n_\alpha$ such that $\Psi_{n_\alpha}$ is a continuous non-constant function of $\alpha$. Now since the number of free parameters is finite, there exists some $l$ such that determining $(\Psi_1, \Psi_2, \cdots, \Psi_l)$ would completely fix all of the parameters of the $\mathbb{Z}_k^{(r)}$-algebra. Our postulate is now stating that the first \emph{two}, i.e. $\Psi_1, \Psi_2$, already fix the current algebra.

As we will prove in subsection \ref{sec_structure_polynomial}, $\Psi_1$, i.e. the first wavefunction, is \emph{trivial}:
\begin{equation}
	1=\frac{1}{g_k}\avg{\psi(z_1)\cdots\psi(z_k)}\prod_{i<j}(z_i-z_j)^{\frac{1}{\nu}}=
	\Psi_1(z_1, \cdots, z_k)
\end{equation}
Despite $\Psi_1$ being trivial as a polynomial, as a correlation function, the above statement is highly non-trivial. As we will show, the relation $\Psi_1=1$ leads to a constraint on the structure constants:
\begin{equation}
	C_{a,b} = \frac{g_{a+b}}{g_ag_b} = \prod_{j=0}^{a-1} \frac{C_{1,j+b}}{C_{1,j}}
\end{equation}
However, beyond this constraint, the remaining CFT Data is carried by the so-called \emph{principal wavefunction} $\Psi_2$ (or its specialization $\chi$). The justification of $\chi \to \mathbb{Z}_k^{(r)}\text {algebra}$ (this is slightly more convenient than $\Psi_2\to \mathbb{Z}_k^{(r)}\text {algebra}$) is thus reduced to studying how the CFT data -- e.g. structure constants and central charge, etc. -- can be explicitly read off from the semi-invariant $\chi$. We prove that since $\chi\in \mathcal{T}_{k+1}^{2r}$ (i.e. due to the structure of $\chi$, not its specifics), together with $\chi(1^{\times a+1},0^{\times k-a})=C_{1,a}^2$ and $C_{a,b}=g_{a+b}/g_ag_b$, the structure constants $C_{a,b}$ with $a+b\leq k$ are as follows ($\Gamma$ is the gamma function)
\begin{equation}
	C_{a,b}^2 = \frac{(a+b)!(k-a)!(k-b)!}{k!a!b!(k-a-b)!}
	\prod_{p=1}^{r-1}\frac{\Gamma(t_p+a+b)\Gamma(t_p+k-a)\Gamma(t_p+k-b)\Gamma(t_p)}{\Gamma(t_p+k-a-b)\Gamma(t_p+a)\Gamma(t_p+b)\Gamma(t_p+k)}
	\label{eq:intro_Cab}
\end{equation}
Here, $t_1, t_2, \cdots, t_{r-1}$ are \emph{free} parameters of the CFT (when $k+1>2r$; when $k+1\leq 2r$, some of $t_i$ are redundant). Complementing this observation with charge conjugation symmetry $C_{k-a, k-b}=C_{a,b}$, all structure constants are acquired. Moving on to the central charge, using the existing techniques in the literature \cite{bernevig2009central,estienne2010clustering, wen1994chiral}, we can read off $c$ from a different specialization of $\chi$:
\begin{equation}
	\xi(x):=\chi(1,x, 0^{\times k-1}) = 1 - \frac{x}{\nu} + \left(\frac{\frac{1}{\nu}(\frac{1}{\nu}-1)}{2}+\frac{2h^2}{c}\right)x^2+O(x^3)
\end{equation}
i.e. $c$ can be found from the coefficient of $x^2$ in the polynomial $\xi(x)$.

As a quick sanity check, consider the $\mathbb{Z}_k^{(1)}$-algebras. 
In this case, $\dim\mathcal{T}_{k+1}^{2}=1$, and the constraint $\chi(1,0^{\times k})=1$ completely fixes the minimal polynomial:
\[
\chi(z_1, \cdots, z_{k+1})=\frac{1}{k}\sum_{i<j}(z_i-z_j)^2
\]
From this, one can find $c=2(k-1)/(k+2)$; also $C_{1,a}^2=\chi(1^{\times a+1}, 0^{\times k-a})=(a+1)(k-a)/k$. Complementing this with $C_{a,b}=g_{a+b}/g_ag_b$, the $r=1$ case of Eq. \eqref{eq:intro_Cab} is obtained. Unsurprisingly, these results match what Zamolodchikov and Fateev have reported upon discovering these parafermionic theories \cite{Zamolodchikov_Fateev_Parafermion}.

\subsection{Evidence for $(\Psi_1,)\chi\to \Psi_2$ for $r=1,2$}
In this introductory subsection, we will limit ourselves to two concrete examples: $\mathbb{Z}_2^{(2)}$ and
$\mathbb{Z}_3^{(2)}$ algebras. The current algebra $\mathbb{Z}_2^{(2)}$ describes a free boson (i.e. $\psi(z)=i\partial \phi(z)$), has $c=1$ and zero degrees of freedom. In contrast, $\mathbb{Z}_3^{(2)}$-algebras have one free parameter $t$. The central charge $c$ and the structure constant $C_{1,1}$ of the theory are given by
\[
c = 2\left(1-\frac{3}{(2t+3)(2t+1)}\right),
\qquad C^2_{1,1} = \frac{4}{3}\frac{(t+1)^2}{t(t+2)}
\]
The $\mathbb{Z}_3^{(2)}$-algebra, and its representations, are studied in depth in Ref. \cite{zamolodchikov1987representations}. Moreover, the principal wavefunction $\Psi_2$ are computed in Ref. \cite{estienne2010clustering} using Jack polynomials. Using these two examples, we will demonstrate an alternative approach to computations of this type. The principal $\mathbb{Z}_3^{(2)}$-WF, in particular, can capture some of the nuances one often encounters carrying out these computations. At the same time, it is still simple enough to discuss in an introductory and minimalistic manner. Via these examples, we illustrate how to check $\chi\to \Psi_2$ by explicitly computing both $\chi$ and $\Psi_2$, \emph{but without directly using the OPEs}.

Finding the minimal polynomial is relatively simple. When $2r=4$, we have $\dim \mathcal{T}_{2+1}^{4}=1$ and $\dim \mathcal{T}_{3+1}^{4}=2$. We are looking for $\chi\in \mathcal{T}_{k+1}^4$ satisfying $\chi(1,0^{\times k})=1$. For $k=2$, there is only one choice:
\[
\chi_2(z_1, z_2, z_3)=\frac{1}{2}\sum_{i<j}(z_i-z_j)^4
\]
In contrast, when $k=3$, we have a continuous family (parametrized by $\beta$) of semi-invariants $\chi_3$:
\[
\chi_3(z_1, z_2, z_3, z_4)=\frac{1}{3}\sum_{i<j}(z_i-z_j)^{4}+
\frac{2\beta}{3}
\big[(z_1-z_2)^2(z_3-z_4)^2
+(z_1-z_3)^2(z_2-z_4)^2
+(z_1-z_4)^2(z_2-z_3)^2
\big]
\]
Using the fact that $\chi(1,1,0,0)=C_{1,1}^2$, we then find that $\beta= 1/t(t+2)$. However, note that $\beta$ works just as well as $t$ for parameterizing the CFT. By calculating $\chi_2(1,x,0)$ and $\chi_3(1,x,0,0)$ one can rediscover the central charges that we listed above. To summarize, we wrote down the most general semi-invariant $\chi \in \mathcal{T}_{k+1}^{2r}$, applied the normalization condition $1=\chi(1,0^{\times k})$, and then identified the remaining free parameters with the CFT data through certain specializations.

The computation of $\Psi_2$ follows a similar methodology, albeit the procedure is much more involved. We begin by creating an ansatz for the principal $\mathbb{Z}_k^{(r)}$-WF $\Psi_2$. The starting point is the fact that $\Psi_2$ is conformally covariant: Let $f(z)=(az+b)/(cz+d)$ be an arbitrary M\"{o}bius map with $ad-bc\neq 0$. Then 
\[
\Psi_2(f(z_1), \cdots, f(z_{2k}))=\prod_{i=1}^{2k}(df/dz|_{z_i})^{2r/2}
\Psi_2(z_1, \cdots, z_{2k})
\]
One says that $\Psi_2$ is a uniform state on a sphere, with $2k$ particles and $2r$ flux quanta (aka $(2k,2r)$ uniform state). Now, there is an isomorphism between uniform states on the sphere and so-called \emph{binary invariants} (see \cite{olverinv} for an introduction to binary invariants). We discuss this isomorphism in section \ref{sec_binary}. It is worth mentioning that the binary invariant that corresponds to a uniform state (i.e., a wavefunction) is essentially its second quantized formulation in terms of lowest Landau level orbits. 
The critical point is that uniform states over a sphere with $2r$ flux quanta ($r=1,2,3,4$) are already classified. These mathematical endeavors began in the late 19th century and spanned several generations. While these classification efforts are made in the literature on binary invariants, hopefully, this paper can bridge the language barrier for the quantum Hall research community. Reporting the end-result, the most general $(2k,4)$ uniform state, with $k=2,3$ are
\[
\Psi^{[k=2]}_2(z_1, \cdots, z_4)=\frac{1}{2!^2 2!^2}
\mathscr{S}\left[\prod_{p=0,1}(z_{2p+1}-z_{2p+2})^4\right]
\]
and 
\[
\begin{aligned}
	\Psi^{[k=3]}_2(z_1, \cdots, z_6)&=\frac{1}{2!^3 3!^2}
	\mathscr{S}\left[\prod_{p=0,1,2}(z_{2p+1}-z_{2p+2})^4\right]
	\\
	&+\frac{\alpha}{3!3!^2}
	\mathscr{S}\left[
	\prod_{p=0,1}
	(z_{3p+1}-z_{3p+2})^2
	(z_{3p+1}-z_{3p+3})^2
	(z_{3p+2}-z_{3p+3})^2
	\right]
\end{aligned}
\]
Here, $\alpha$ is a free parameter we need to fix. It turns out that $\alpha=\beta=1/t(t+2)$. While this clean $\alpha=\beta$ relation is an artifact of the simplicity of $3\leq k<6$ (the general $k$ is more involved), the overall theme remains the same: For $k>2$, the $\mathbb{Z}_k^{(2)}$-algebras and subspace of $\mathcal{T}_{k+1}^{4}$ with $\chi(1,0^{\times k})=1$ are both one-dimensional. We can choose either $t$ or $\beta$ as the parameter. The most general $(2k,4)$ uniform state is explicitly known up to coefficients like $\alpha$. 
Thus, one must find the $\alpha$ coefficients via various specializations and relate them to the CFT data. We omit the details at this introductory stage. The sections \ref{sec_binary} and \ref{sec_about} will greatly expand on the above paragraph and generalize it appropriately.

In subsection \ref{sec_quartic}, we calculate the principal $\mathbb{Z}_k^{(2)}$ wavefunctions for an arbitrary value of $k$ and provide a relatively clean presentation for it. Similar to the abovementioned examples, the calculation relies on an ansatz for principal wavefunction that uses binary invariants. In addition, we have extensively used the graph-theoretic point of view toward FQH states that we developed elsewhere \cite{pakatchi2018quantum}. In subsections \ref{sec_sextic} and \ref{sec_octavic}, we have carried out the procedure for $r=3,4$ as well. However, we are not as successful in those cases as in $r=2$. Some of the coefficients in the ansatz remain undetermined in those cases. Determining these unknown coefficients will require further technical innovations and possibly a better understanding of $\mathbb{Z}_k^{(3)}, \mathbb{Z}_k^{(4)}$-algebras. We also compute the principal $\mathbb{Z}_2^{(r)}$ wavefunctions (arbitrary $r$) in \ref{appendix_direct}, although the ansatz and method used is different from that in the main body.

\section{Uniform States on a Sphere}
\label{sec:uniform}
We begin by reviewing the notion of uniform states (over a sphere) and their relevance to a quantum Hall system. Consider a gas of
$N$ spin-polarized bosonic/fermionic electrons over a sphere with radius $R=1/2$. Due to a magnetic monopole at the center, $N_\phi$ flux quanta is passing through the sphere. This setup is sometimes called the Haldane sphere \cite{Haldane_Sphere_Pseudo}. We are interested in describing the states in the lowest Landau level (LLL). Using stereographic projection (see Fig. \ref{streo}), the location of a particle on the sphere is parametrized by a complex coordinate $z$. We take the inner product of the Hilbert space to be
\begin{equation}
	\braket{\phi_2}{\phi_1} = \int \frac{dx dy}{(1+z\bar{z})^{2+N_\phi}}\phi_1(z)\overline{\phi_2(z)}
\end{equation}
The LLL is $(N_\phi+1)$-fold degenerate with an orthonormal basis $\ket{m}$, $0\leq m\leq N_\phi$, where $\braket{z}{m}\propto z^m$.
\begin{figure}
	\centering
	\includegraphics[scale=.5]{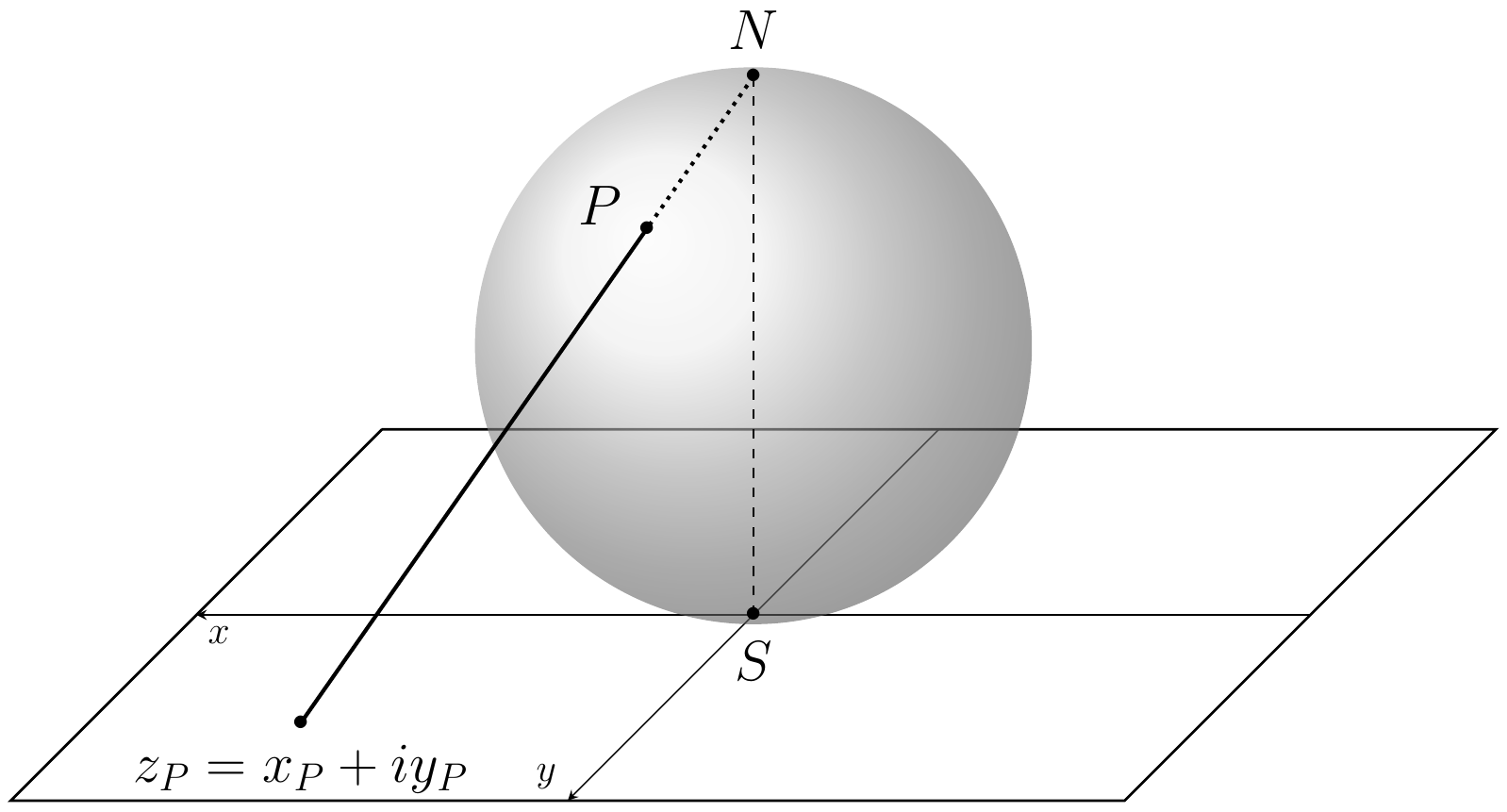}
	\caption{Stereographic Projection: Let $S,N$ be the south and north pole. Make the complex plane tangential to $S$, oriented so that $\hat{x}\times \hat{y}$ points inward. The intersection of the ray $NP$ with the complex plane is the \emph{complex coordinate} $z_P$ (near zero) of the point $P$.}
	\label{streo}
\end{figure}
Extending the inner product to the $N$-body Hilbert space, wavefunctions in the LLL are polynomials $P(z_1, \cdots, z_N)$. In this article, we focus on bosonic systems (Any fermionic wavefunction is the product of a bosonic wavefunction and a Laughlin-Jastrow factor $\prod_{i<j}(z_i-z_j)$). A bosonic wavefunction $P(z_1, \cdots, z_N)$ in the LLL is a symmetric polynomial where the degree of no variable $z_i$  can exceed $N_\phi$. We assume that the quantum Hall wavefunctions are all confined to the LLL.

One of the characteristics of a fractional quantum Hall ground state (i.e., in the absence of local bulk excitations) is that the one-particle density is uniform over the sphere. The uniformity is equivalent to the wavefunction $P$ being $SO(3)$ invariant; i.e., $L^2P=0$ in the induced representation of $SO(3)$ in the presence of $N_\phi$ flux quanta. More concretely, 
\begin{subequations}
	\begin{align}
		L_-P&:=
		\sum_i \partial_iP
		=0\\
		L_zP&:= \sum_i \left(-\frac{N_\phi}{2}+z_i\partial_i\right)P = 0\\
		L_+P &:= \left(\sum_{i}N_\phi z_i-z_i^2\partial_i\right)P = 0
	\end{align}
\end{subequations}
These equations are a variant of the Ward identities in conformal field theory and can be reformulated in terms conformal transformations: For any arbitrary M\"{o}bius map $f(x):=(ax+b)/(cx+d)$ (with $ad-bc\neq 0$) the wavefunction $P$ transform via
\begin{equation}
	P(f(z_1), \cdots, f(z_N))=\prod_{i=1}^{N}\left(\frac{df}{dx}\right)_{x=z_i}^{\frac{1}{2}N_\phi}P(z_1, \cdots, z_N)
\end{equation}
We say $P$ is an \emph{$(N,N_\phi)$ uniform state}. Furthermore, at the onset of the quantum Hall effect, $N, N_\phi$ are related through the constraint
$N_\phi = N\nu^{-1}-\mathcal{S}$; here $\nu$ is the filling fraction (a rational number) and $\mathcal{S}$ is the shift (an integer).

\section{$\mathbb{Z}_k$ Parafermionic Conformal Field Theories}
\label{sec:para}
This paper is concerned with model FQH ground states that descend from $\mathbb{Z}^{(r)}_k$ parafermionic current algebras. We will define the relevant conformal field theories (CFT) in this section.

\subsection{Zamolodchikov-Fateev Parafermions}
A $\mathbb{Z}_k$ parafermionic CFT, with some chiral algebra $\mathcal{W}$ (see \ref{appendix_chiral} for a review), consists of $k$ holomorphic  $\mathcal{W}$-primary fields: $1, \psi(z), \psi_2(z), \cdots, \psi_{k-1}(z)$. These fields make a representation of the cyclic group $\mathbb{Z}_k$ via the fusion rules $\psi_a\star \psi_b=\psi_{a+b}$ (addition done in $\mathbb{Z}_k$, $\psi_0\equiv 1$, and $\psi_1\equiv \psi$). The transformation $\psi_a\mapsto \psi_a^+\equiv \psi_{k-a}$ is the \emph{charge conjugation}. We denote the scaling dimension of $\psi_a$ by $h_a$. The relative scaling dimension is given by
\begin{equation}
	\gamma(a,b)=h_a+h_b-h_{a+b}, \qquad \gamma\equiv \gamma(1,1)
\end{equation}
The field space of the chiral algebra $\mathcal{W}$ can be graded as $\mathcal{W}=\bigoplus_{\Delta\in \mathbb{N}}\mathcal{W}_\Delta$, where $\mathcal{W}_\Delta$ is the subspace of chiral field with spin/scaling dimension $\Delta$. We assume that $\dim \mathcal{W}_1=0$ and $\mathcal{W}_2=\mathrm{span}(L)$, where $L$ is the energy-momentum tensor. The operator product expansions (OPEs) between the parafermions $\psi_a$'s are given by
\begin{subequations}
	\label{OPEs_eq}
	\begin{align}
		\psi_a(z)\psi_b(w)&=\frac{C_{a,b}}{(z-w)^{\gamma(a,b)}}\psi_{a+b}(w)+\cdots, \qquad (a+b\neq 0)\\
		\psi_a(z)\psi_a^+(w)&=\frac{1}{(z-w)^{2h_a}}\left(1+\frac{2h_a}{c}L(w)+\cdots\right)
	\end{align}
\end{subequations}
where $C_{a,b}$'s are the structure constants, and $c$ is the central charge. The fields in the ellipsis have scaling dimensions that differ from the leading field by a non-zero integer. We require the CFT to be symmetric under charge conjugation; i.e., if we replace all of $\psi_a$'s with their charge conjugation $\psi_a^+$,  then no correlation function changes. An associative solution to this operator product algebra (OPA) is called a \emph{$\mathbb{Z}_k$ parafermionic algebra}. The fields $\psi_a$'s are the parafermions introduced by Zamolodchikov-Fateev in Ref. \cite{Zamolodchikov_Fateev_Parafermion}. For these parafermions the \emph{(semi-)locality condition} \cite{noyvert, bakalov} is given by
\begin{equation}
	\label{locality_eq}
	\psi_a(z)\psi_b(w)(z-w)^{\gamma(a,b)}=\mu_{a,b}\psi_b(w)\psi_a(z)(w-z)^{\gamma(a,b)}
\end{equation}
where $0\neq \mu_{a,b}\in \mathbb{C}$ is called the \emph{commutation factor}. For the reader's convenience, we have rederived \cite{Pattern_of_Zeros} the most general solution for scaling dimensions consistent with associativity in  \ref{appendix_pattern}. 

\subsection{$\mathbb{Z}_k^{(r)}$-Algebras \& $\mathbb{Z}_k^{(r)}$-Wavefunctions}
In this paper, we will utilize and study only a sub-class of $\mathbb{Z}_k$ parafermionic algebras. The associative solutions of interest are those with the following scaling dimensions:
\begin{equation}
	\label{scaling_Zkr_eq}
	h_a = \frac{a(k-a)}{2}\gamma, \qquad \gamma = \frac{2r}{k}, \qquad (r\in \mathbb{Z}_{+})
\end{equation}
We call any such associative solution a \emph{$\mathbb{Z}_k^{(r)}$ (current) algebra}. For these current algebras, the commutation factors are trivial: For any pair $a,b$, we have
\begin{equation}
	\psi_a(z)\psi_b(w)(z-w)^{ab\gamma}=\psi_b(w)\psi_a(z)(w-z)^{ab\gamma}
\end{equation}
Alternatively, since $D(a,b)=ab\gamma - \gamma(a,b)=S_{a+b}-S_a-S_b$ is an even integer for $\mathbb{Z}_k^{(r)}$-algebras (See \ref{appendix_pattern} and Eq. \eqref{eq_pattern_Zkr}), we may also write
\begin{equation}
	\label{eq_semilocality}
	\psi_a(z)\psi_b(w)(z-w)^{\gamma(a,b)}=\psi_b(w)\psi_a(z)(w-z)^{\gamma(a,b)}
\end{equation}
Moving on to the correlation functions, corresponding to a $\mathbb{Z}_k^{(r)}$-algebra we construct an infinite sequence of wavefunctions $\gv{\Psi}=(\Psi_1, \Psi_2, \cdots, \Psi_n, \cdots)$ as follows:
\begin{equation}
	\label{krwavefunciton_eq}
	\Psi_n(z_1, \cdots, z_{nk})=\frac{1}{\prod_{j=0}^{k-1}C_{1,j}^n}\avg{\psi(z_1)\psi(z_2)\cdots\psi(z_{nk})}\prod_{i<j}(z_i-z_j)^{\gamma}
\end{equation}
Due to the semi-locality condition, the functions $\Psi_n$ are symmetric polynomials. The holomorphic Ward identities further show that $\Psi_n$ is an $(N,N_\phi)$ uniform states with $N=nk$ and $N_\phi=(n-1)r$
(i.e. uniform state over sphere with $N$ particles and $N_\phi$ flux quanta). The filling fraction $\nu$ and the shift $\mathcal{S}$ are respectively
$\nu=\gamma^{-1}=k/2r$ and $\mathcal{S}=2r$. Finally, note that \cite{estienne2010clustering}
\begin{equation}
	\begin{aligned}
		\Psi_{n+1}(w^{\times k}, z_1, \cdots, z_{nk})&\equiv
		\lim_{w_1\to w}\lim_{w_2\to w}\cdots 
		\lim_{w_k\to w}\Psi^{(k,r)}_{n+1}(w_1, w_2, \cdots, w_k, z_1, \cdots, z_{nk})\\
		& = \prod_{i=1}^{nk}(z_i-w)^{2r}\Psi^{(k,r)}_n(z_1, \cdots, z_{nk})
	\end{aligned}
	\label{Zkr_clustering}
\end{equation}
One says the sequence $\gv{\Psi}=(\Psi_n)_{n\geq 1}$ is \emph{$(k,2r)$-clustering} (or $k$-clustering of degree $2r$).

Throughout the paper, the wavefunction  $\Psi_n$ defined in Eq. \eqref{krwavefunciton_eq} is called a \emph{$\mathbb{Z}_k^{(r)}$-wavefunction (WF)} of size $n$.
The $\mathbb{Z}_k^{(r)}$-WFs of size $n=1,2$ are called the \emph{trivial} and \emph{principal} $\mathbb{Z}_k^{(r)}$-wavefunctions respectively. Moreover, 
a reference to symbol $\Psi$ (without any $n$ index) will automatically mean $n=2$; i.e. $\Psi$ is the principal $\mathbb{Z}_k^{(r)}$-wavefunction ($2k$ variables).

\section{$\mathbb{Z}_k^{(r)}$-Wavefunctions as Model FQH Ground States}
The pseudopotential formalism is easier to work with in the plane/disk geometry. Thus, let us begin with some comments on FQH states in plane geometry. The states with $N$ `bosonic electrons' in the LLL are described by
\begin{equation}
	\Phi(z_1, \cdots, z_N) = P(z_1, \cdots, z_N)\exp\left[-\sum_i |z_i|^2/4\ell_B^2\right]
\end{equation}
where $P$ is a symmetric polynomial and $B(2\pi \ell_B^2)=h/e=$ magnetic flux quantum. Similar to the sphere, we can absorb the geometric measure (the Gaussian) into the definition of the inner product and consider $P$ itself as the wavefunction. If $P$ is an $(N, N_\phi)$ uniform state over a sphere, and if we think of the norm of $\Phi$ (in plane geometry) as a droplet, the wetting area of $|\Phi|^2$ would be roughly $N_\phi (2\pi \ell_B^2)$. In other words, the lower $N_\phi$ is, the higher the density.

In plane geometry, projected to the lowest Landau level, the effective Hamiltonians of fractional quantum Hall systems are pure multi-electron interaction. Using $b$ as the number of electrons involved in the interaction, in plane geometry, the effective interaction $V(z_1, \cdots, z_b)$ is both translational and rotational invariant. To get from these interactions to pseudopotentials, let us introduce some terminology. Some of the core players in this paper are the so-called \emph{semi-invariants}. A $(b,m)$ semi-invariant is a polynomial with $b$ variables, that is symmetric, translational invariant, and homogeneous of degree $m$. We denote the linear space of $(b,m)$ semi-invariants  by $\mathcal{T}_{b}^m$. If $\ket{m,q}$ is an orthonormal basis for $\mathcal{T}_{b}^m$ ($1\leq q\leq \dim \mathcal{T}_b^m$), the $b$-body interaction Hamiltonian is generally of the form  \cite{Simon_projection}:
\begin{equation}
	H = \sum_{i_1<\cdots<i_b}\:
	\sum_{m\geq 0} \:
	\sum_{q,q'=1}^{\dim \mathcal{T}_b^{m}}
	\ket{m,q;z_{i_1},\cdots, z_{i_b}}
	V^{[m]}_{q, q'}
	\bra{m, q';z_{i_1},\cdots, z_{i_b}}
	\label{hamiltonian_sp_eq}
\end{equation}
The matrices $V^{[m]}$ are called the \emph{pseudopotentials} and should be positive semi-definite (so that the ground state has zero energy).
The \emph{interaction number} $b$, together with the pseudopotentials $V^{[m]}$ are the data of the Hamiltonian. A side-effect of the plane geometry is that it is infinite, and thus $N_\phi$ (the max degree of a variable) for a zero-energy state can be arbitrarily large. To circumvent this issue and also simulate the more realistic geometry that is of a disk, it is customary to call the \emph{densest} (smallest $N_\phi$) zero-energy state of $H$ as \emph{the} ground state. We want to emphasize that an actual FQH Hamiltonian has to have a \emph{unique} ground state in disk/sphere topology.

At this stage, the best option to study FQH states from the Hamiltonian point of view is to design \emph{models}. Thus, one must make a particular choice for the interaction number $b$ and the pseudopotentials $V^{[m]}$. A realistic quantum Hall Hamiltonian is gapped (incompressibility), has a unique ground state, and ground state would be uniform if put on a sphere. Thus, if a model pseudopotential is reasonable in the above sense, its densest zero-energy state is called a \emph{model FQH ground state}. The simplest of model Hamiltonians is the \emph{Haldane Hamiltonian} \cite{Haldane_Sphere_Pseudo}:
\begin{equation}
	\mathscr{H}_{2}^{2r-1}=
	\sum_{m=0}^{2r-1}V_m
	\sum_{i< j}
	\mathcal{P}_{2}^m(z_{i}, z_j), \qquad (V_m>0)
\end{equation}
Here and onwards, for a general $b$, $\mathcal{P}_b^m(z_I) \equiv \sum_{q}\ket{m,q;z_I}\bra{m,q;z_I}$ is the projection operator to the sector where a $b$-cluster of electrons with coordinates $z_I\equiv z_{i_1}, \cdots, z_{i_{b}}$ has a relative angular momentum $m$. The model ground state of $\mathscr{H}_{2}^{2r-1}$ is the $2r$-Laughlin state. The natural generalization \cite{Gaffnian, Simon_Pseudopotentials} of Haldane's construction would be
\begin{equation}
	\mathscr{H}_{k+1}^{r-1}=
	\sum_{m=0}^{r-1}V_m
	\sum_{i_1< \cdots<i_{k+1}}
	\mathcal{P}_{k+1}^m(z_I), \qquad (V_m>0)
\end{equation}
We call these the \emph{projection Hamiltonians}. For specific values of $k,r$, the model Hamiltonian $\mathscr{H}_{k+1}^{r-1}$ works exceptionally well and has a unique ground state. Other than Laughlin, three other examples are: Read-Rezayi ($\nu=k/r=k/2$) \cite{Read-Rezayi}, Gaffnian ($\nu=k/r=2/3$) \cite{Gaffnian}, and Haffnian ($\nu=k/r=2/4$) \cite{Haffnian, Simon_projection}. However, for the majority of $(k,r)$ values, we have $\dim \mathcal{T}_{k+1}^{r-1}>1$ and thus, these Hamiltonians often do not have a unique densest zero-energy state. This section aims to construct a family of model Hamiltonians that have a better chance of realizing $\mathbb{Z}_k^{(r)}$-wavefunctions as a \emph{unique} ground state (the gap should be studied separately). The construction is based on a novel local property called \emph{separability}.

\subsection{Separability}
\label{sec_sep}
Let us first establish the required terminology. Throughout the discussion, the quantities $k,r>0$ are two fixed integers. The notation $a=pk+q$ is reserved, where $p\geq 0$ and $0<q\leq k$. The functions $\mathcal{N}(n)=kn$ and $\mathcal{N}_\phi(n)=2r(n-1)$ signify the number of particles and number of flux quanta at ``size $n$''. The multi-index $I$ stands for the set $\{i_1, \cdots, i_{a}\}$ (assuming $i_1<\cdots<i_a$). The complementary multi-index $\overline{I}$ is such that $I\cap \overline{I}=\emptyset$ and $I\cup \overline{I}=\{1,2,\cdots, \mathcal{N}(n)\}$. Given a multi-index $J={j_1, \cdots, j_l}$, we define $z_J={z_{j_1}, \cdots, z_{j_l}}$ and $\widehat{z}_J=(z_{j_1}+\cdots+z_{j_l})/l$ (i.e. the center-of-mass). The pattern of zeros $S=(S_1, S_2, \cdots, S_a, \cdots)$ \cite{Pattern_of_Zeros} is understood as the infinite integral sequence:
\begin{equation}
	\label{eq_pattern_Zkr}
	S_a = 2r\left[\frac{p(p-1)}{2}k + pq\right]
\end{equation}
Finally, $\gv{\Psi}$ stands for a sequence of polynomials $\gv{\Psi}=(\Psi_n)_{n\geq 1}$. In the current discussion, The polynomial do not have to be $\mathbb{Z}_k^{(r)}$-wavefunctions. However, the sequence $\gv{\Psi}$ needs to satisfy the following criteria:
\begin{enumerate}
	\item Uniformity: $\Psi_n$ is an $(\mathcal{N}(n),\mathcal{N}_\phi(n))$ uniform state.
	\item $(k,2r)$ Clustering: We have $\Psi_1=1$ and
	$
	\Psi_{n+1}(w^{\times k}, z_1, \cdots, z_{nk})=\prod_{i=1}^{nk}(w-z_i)^{2r}\Psi_n(z_1, \cdots, z_{nk})
	$.
	\item Squeezablity: For all $a$ and $m<S_a$ we have $\mathcal{P}_{a}^{m}\Psi_n=0$. Furthermore, $\mathcal{P}_{a}^{S_a}\Psi_n\neq 0$.
\end{enumerate}
Ignoring the gap issue, we may think of $\Psi_n$ as the model ground state of a hypothetical model Hamiltonian in $\mathcal{N}(n)$ particles. The third condition states that the minimal relative angular momentum that any cluster of $a$ electrons can carry in $\gv{\Psi}$ equals $S_a$. 

On top of the previous properties, we add the following:
We say $\gv{\Psi}$ is \emph{$a$-separable} if there exists a polynomial $\chi^{(a)}\in \mathcal{T}_a^{S_a}$, independent of $n$ (for all $n>p$), such that 
\begin{equation}
	\mathcal{P}_a^{S_a}(z_I)\Psi_n =\chi^{(a)}(z_1, \cdots, z_{a})\prod_{i\in \overline{I}}(z_i-\widehat{z}_I)^{2pr}\Psi_{n-p}(\widehat{z}_I^{\times q}, z_{\overline{I}})
\end{equation}
If $\chi^{(a)}$ exists, we call it the \emph{$a$-factor} of $\Psi_n$. The $(k+1)$-factor is called the \emph{minimal polynomial} of $\gv{\Psi}$ and is simply denoted as $\chi$. We say $\gv{\Psi}$ is \emph{fully separable} if it is $a$-separable for all $a$. As we will show in \S\ref{section_parafermion_OPE}, all $\mathbb{Z}_k^{(r)}$-wavefunctions are fully separable.

The $a$-factor is an element of $\mathcal{T}_a^{S_a}$; i.e. a semi-invariant. However, not every semi-invariant can be a factor. Note that if $\gv{\Psi}$ is $a$-separable, then we can find its factor $\chi^{(a)}$ from $\Psi_{p+1}$ via
\begin{equation}
	\label{eq_chi_f}
	\chi^{(a)}(z_1, \cdots, z_a)=
	\lim_{w\to \infty} w^{-2\mathcal{J}_a} \Psi_{p+1}(w^{\times k-q}, z_1, \cdots, z_a)
\end{equation}
where $2\mathcal{J}_a=a\mathcal{N}_\phi(p+1)-2S_a$. Hence, due to the clustering property, the $a$-factors need also to satisfy
\begin{equation}
	\chi^{(a)}(w_1^{\times k}, w_2^{\times k}, \cdots, w_{p}^{\times k}, 0^{\times q})=\prod_{1\leq i<j\leq p}(w_i-w_j)^{2kr}\prod_{i=1}^p w_i^{2qr}
	\label{eq_factor_constraint}
\end{equation}
We denote by $\mathcal{T}^\star_a$ the subspace of $\mathcal{T}_a^{S_a}$ satisfying \eqref{eq_factor_constraint}. Now, corresponding to any $\omega\in \mathcal{T}_a^\star$ we define a model Hamiltonian
\begin{equation}
	\label{bbody_H_eq}
	H_{\omega}=\mathscr{H}_{a}^{S_a-1}+
	V'\sum_{I}\left(1-\frac{\ket{\omega;z_I}\bra{\omega;z_I}}{\braket{\omega}{\omega}}\right)\mathcal{P}_a^{S_a}, \qquad (V'>0)
\end{equation}
If $\gv{\Psi}$ is a uniform $(k,2r)$-clustering squeezable $a$-separable sequence with $a$-factor $\omega$, then it is a densest zero-energy state of the Hamiltonian $H_\omega$ \cite{pakatchi2018quantum}. Furthermore, if $\gv{\Psi}$ is not $a$-separable, or if it is $a$-separable but with an $a$-factor different than $\omega$, then it would not be an eigenstate of $H_\omega$. As it is clear from the definition, the Hamiltonian $H_\omega$ is a modification of the projection Hamiltonian $\mathscr{H}_a^{S_a-1}$. The two cases where $a=k+1$ and $a=2k$ are of particular interest to us.  Suppose $\gv{\Psi}$ is both $(k+1)$ and $2k$-separable, with minimal polynomial $\chi$. The $2k$-factor of $\gv{\Psi}$ is nothing but its principal wavefunction $\Psi_2$. Referring to \eqref{eq_chi_f}, it is clear that $\chi$ is fixed by $\Psi_2$. In other words, while $\gv{\Psi}$ is a ground state for both $H_\chi$ and $H_{\Psi_2}$, naively, it is reasonable to think $H_{\Psi_2}$ is more restrictive than $H_\chi$.

The above observation leads to the following question:  
Assuming we are only interested in fully separable $(k,2r)$-clustering squeezable FQH ground states, should we abandon the long-standing belief that FQH Hamiltonians of $k$-clustering states are $(k+1)$-body interactions? After all, the $2k$-body Hamiltonian $H_{\Psi_2}$ looks more promising to get a \emph{unique} ground state. Alternatively, perhaps, despite the naive expectations, in actuality, the two Hamiltonians are equivalent! We cannot give definitive answers to these questions in this paper. However, by restricting our attention to $\mathbb{Z}_k^{(r)}$-WFs (i.e., the prime example of fully separable states), we would investigate whether or not the minimal polynomial $\chi$ can fix the principal wavefunction $\Psi_2$. An affirmative answer is evidence of the equivalence of $H_\chi$ and $H_{\Psi_2}$. This is the subject of sections \ref{sec_binary} and \ref{sec_about}. The evidence gets stronger if $\chi$ fixes the entire $\mathbb{Z}_k^{(r)}$-algebra. The latter requires obtaining the structure constants, central charge, and chiral weights as data encoded inside the minimal polynomial. We pursue this study in section \ref{sec_analysis}.

\subsection{Multi-Parafermion OPEs \& Full Separability of $\mathbb{Z}_k^{(r)}$-Wavefunctions}
\label{section_parafermion_OPE}
To complete the discussion, in this section, we show that $\mathbb{Z}_k^{(r)}$-WFs are indeed fully separable. Define $g_0=1$ and $g_a=\prod_{j=0}^{a-1}C_{1,j}$ for all $1\leq a\leq k$. The key quantity that we need to compute is the $a$-parafermion OPE/field $\Phi_a(z_1, \cdots, z_a)$ defined as
\begin{equation}
	\Phi_a(z_1, \cdots, z_a):= \frac{1}{g_k^p g_q}\psi(z_1)\psi(z_2)\cdots\psi(z_a)\prod_{1\leq i<j\leq a}(z_i-z_j)^{\gamma}
\end{equation}
Ordinarily, one computes such an OPE by repeatedly utilizing the two-point OPEs $\psi_a(z)\psi_b(w)$ (after the associativity of a solution is confirmed). This method is very inefficient and computationally very expensive. Here, we will provide an alternative. Due to the simple fusion rules $\psi_a\star \psi_b=\psi_{a+b}$, it is immediate that $\Phi_a$ lies in the $\mathcal{W}$-family of $\psi_{a}$. In fact, to the lowest order, we expect
\begin{equation}
	\Phi_a(z_1, \cdots, z_a)= \chi^{(a)}(z_1, \cdots, z_a)\psi_a\left(\widehat{z}_a\right)+\cdots
	\label{multipara_eq}
\end{equation}
where $\widehat{z}_a=(z_1+\cdots+z_a)/a$ is the center of mass and $\chi^{(a)}(z_1, \cdots, z_a)$ is (at this stage) just some function. We claim that:
\begin{thm}
	\label{thm_multipara}	
	The function
	$\chi^{(a)}(z_1, \cdots, z_a)$ belongs to $\mathcal{T}_a^{S_a}$.
\end{thm}
\noindent
Before proving the theorem, note that in Eq. \eqref{multipara_eq}, the fields appearing in ellipsis are $\mathcal{W}$-descendants of $\psi_a$. If $\phi$ is one such descendant with scaling dimension $h_a+n$, then it will be accompanied by a translation invariant homogeneous symmetric polynomial of degree $\deg \chi^{(a)}+n$ in the OPE. We will only need the leading term of this OPE in the main body of the paper. We present an ansatz for higher order terms in \ref{appendix_OPE}.
\newproof{pot}{Proof of Theorem \ref{thm_multipara}}
\begin{pot}
	Note that
	\begin{equation}
		\label{chib_eq}
		\lim_{w\to \infty} w^{-2pr(k-q)}\Psi_{p+1}(w^{\times k-q}, z_1, \cdots, z_{a})=
		\avg{\psi_q^+(\infty)\Phi_a(z_1, \cdots,z_a)}
		=\chi^{(a)}(z_1, \cdots, z_a)
	\end{equation}
	where we have used the $a$-point OPE ansatz. [Throughout, given a field $\phi$ of scaling dimension $\Delta$, the notation $\phi(\infty)$ will mean $\lim_{w\to \infty}w^{2\Delta}\phi(w)$].
	The left-most quantity in the equality is a homogeneous translational invariant symmetric polynomial in $z_1, \cdots, z_a$. We refer the reader to \ref{appendix_pattern} for a proof of  $\deg \chi^{(a)}=S_a$ .
\end{pot}
Multi-parafermion OPEs allow us to compute the wavefunction's leading ``relative behavior'' $\Psi_n$. Note that (with $\vv{w}=w_1, \cdots, w_a$, $\vv{z}=z_1, \cdots, z_{nk-a}$ and $\widehat{w}=(w_1+\cdots+w_a)/a$):
\[
\begin{aligned}
	\Psi_n(\vv{w}, \vv{z})&=
	\frac{1}{g^{n-(p+1)}_kg_{k-q}}
	\avg{\Phi_a(\vv{w})\psi(z_1)\cdots \psi(z_{nk-a})}
	\prod_{i=1}^a\prod_{j=1}^{nk-a}(w_i-z_j)^{\gamma}
	\prod_{i<j}(z_i-z_j)^{\gamma}\\
	&=
	\frac{\chi^{(a)}(\vv{w})}{g_{n-(p+1)}^kg_{k-q}}
	\avg{\psi_q(\widehat{w})\psi(z_1)\cdots \psi(z_{nk-a})}
	\prod_{j=1}^{nk-a}(\widehat{w}-z_j)^{a\gamma}
	\prod_{i<j}(z_i-z_j)^{\gamma}+\cdots\\
	&=\chi^{(a)}(\vv{w})
	\prod_{i=1}^{nk-a}(z_i-\widehat{w})^{2pr}
	\Psi_{n-p}(\widehat{w}^{\times q}, \vv{z})+\cdots
\end{aligned}
\]
Here, the ellipsis consists of terms with relative angular momentum higher than $S_a$ (in $\vv{w}$-variables). In other words, using the notation of \S\ref{sec_sep}, $\mathcal{P}_a^{m}(z_I)\Psi_n=0$ for all $m<S_a$, and
\begin{equation}
	\label{local_eq}
	\mathcal{P}_a^{S_a}(z_I)\Psi_n=
	\chi^{(a)}(z_I)\prod_{i\in \overline{I}}(z_i-\widehat{z}_I)^{2pr}\Psi_{n-p}(\widehat{z}_I^{\times q}, z_{\bar{I}})
\end{equation}
This concludes the proof that $\mathbb{Z}_k^{(r)}$-wavefunctions are fully separable and squeezable.

\section{Polynomialities}
\label{sec_analysis}
As discussed in the motivation section, our main hypothesis is that $\chi\xrightarrow{\mathrm{fix}} \Psi_2 \xrightarrow{\mathrm{fix}} \mathbb{Z}_k^{(r)}\text{algebra}\xrightarrow{\mathrm{fix}} \gv{\Psi}$. The validity of the hypothesis, in turn, makes the $H_\chi$ Hamiltonians viable and natural as models. In this section, we will explore if and how the principal (or its specialization $\chi$) can fix various CFT data; e.g. structure constants $C_{a,b}$, central charge $c$, chiral weights $h_a^W$, etc.

\subsection{A General Observation}
As usual, let us begin by introducing a bit of terminology. Fixing some $0\leq b<k$, suppose $\phi(z)$ be a quasi-primary $\mathcal{W}$-descendant of $\psi_{b}^+(z)$. We use the notation $\beta(\phi)=b$ and write the scaling dimension of $\phi$ in the form $h_\phi=h_b+l(\phi)$. We call $l(\phi)$ the \emph{level} of $\phi$. In the study of $\mathbb{Z}_k^{(r)}$-wavefunctions, especially the principal ones, correlation functions of the following form are a common sight:
\begin{equation}
	T_\phi(z_1, \cdots, z_{k+\beta(\phi)})=\frac{1}{g_kg_{\beta(\phi)}}
	\avg{\phi(\infty)\psi(z_1)\cdots\psi(z_{k+\beta(\phi)})}\prod_{i<j}(z_i-z_j)^{\gamma}
\end{equation}
We call $T_\phi$ the principal semi-invariant corresponding to $\phi$. Various data of the $\mathbb{Z}_k^{(r)}$-algebra is encoded inside of semi-invariants of this form. We can extract these data by performing specializations. For example, a beneficial specialization would be:
\begin{equation}
	\tau_\phi(a)=T_\phi(1^{\times a+\beta(\phi)},0^{\times k-a})=
	\frac{g_{a+\beta(\phi)}g_{k-a}}{g_k g_{\beta(\phi)}}\avg{\phi(\infty)\psi_{a+\beta(\phi)}(1)\psi^+_a(0)}
\end{equation}
Our first task in this section is to prove that $\tau_\phi(a)$ is a polynomial in the (integer) cluster parameter $a$. We use the umbrella term \emph{Polynomiality} when a certain $\mathbb{Z}_k^{(r)}$ quantity parametrized by $a$ depends on it polynomially. The polynomialities coming from the special cases $\phi=\psi_b^+$ (for $b\neq 0$) and $\phi=W\in \mathcal{W}$ are studied in some detail in this section.

\begin{thm}
	\label{thm_polynomiality}
	$\tau_\phi(a)$ is a polynomial of degree $2r\beta(\phi)+l(\phi)$ with the following properties:
	\begin{enumerate}
		\item It enjoys a symmetry $\tau_\phi(k-a-\beta(\phi))=(-1)^{l(\phi)} \tau_\phi(a)$.
		\item It is divisible by $(a+\beta(\phi))(k-a)$.
		\item $\tau_\phi(0)=\delta_{l(\phi),0}$.
	\end{enumerate}
\end{thm}

\begin{proof}
	For simplicity, let $b=\beta(\phi)$ and $l=l(\phi)$.
	Beginning with the last condition, note that $\tau_\phi(0)= \avg{\phi(\infty)\psi_b(1)}$ which is only non-zero if $\phi=\psi_b^+$, i.e. $l=0$. Regarding the degree, the proof that
	$\deg T_\phi=2br+l$ can be found in \ref{appendix_OPE}.
	We prove the rest by constructing a convenient basis for $\mathcal{T}^{2br+l}_{k+b}$. Define
	\begin{equation}
		\label{basis_eq}
		P_n(z_1, \cdots, z_{k+b})=
		\sum_{i=1}^{k+1}\left(z_i - \widehat{z}\right)^n\in \mathcal{T}_{k+b}^n
	\end{equation}
	where $\widehat{z}=(z_1, \cdots, z_{k+b})/(k+b)$ is the center-of-mass. Note that $P_1=0$. Let $\pi^*(2br+l, k+b)$ be the set of partitions $\lambda=(\lambda_1, \lambda_2, \cdots)$ that satisfy  $|\lambda|=\lambda_1+\cdots+\lambda_\ell=2br+s$ and $2\leq \lambda_i\leq k+b$ for all $i$. For each $\lambda\in \pi^*(2br+l, k+b)$, define $P_\lambda =\prod_{j} P_{\lambda_j}$. Then $P_\lambda$ is a basis for $\mathcal{T}_{k+b}^{2br+l}$. Now note that
	\[
	\tau_n(a)\equiv P_n(1^{\times a+b}, 0^{\times k-a}) = \frac{(a+b)(k-a)^n + (-1)^n(k-a)(a+b)^n}{(k-b)^n}
	\]
	which is a polynomial in $a$ of degree $n$ (the coefficient of $a^{n+1}$ is zero). Since $n>1$, $(a+b)(k-a)$ divides $\tau_n(a)$. Moreover, $\tau_n(k-a-b)=(-1)^n\tau(a)$. The assertion now follows from the fact that $T_\phi$ is a linear combination of $P_\lambda$'s, and $P_\lambda$'s are themselves a product of $P_n$'s.
\end{proof}
\noindent
Note that the polynomiality of $\tau_\phi(a)$ is purely a consequence of $T_\phi$ being a semi-invariant.

\subsection{Structure Constants}
\label{sec_structure_polynomial}
As a first example of polynomialities, we will shortly consider $\phi=\psi^+$ (i.e. $\beta(\phi)=1$ and $l(\phi)=0$). We will use this polynomiality to determine the structure constants $C_{a,b}$. It would be helpful if we first proved a few general identities concerning $C_{a,b}$'s. Firstly, from charge conjugation symmetry it is immediate that $C_{a,b}=C_{k-a, k-b}$. The identity $C_{a,b}=C_{b,a}$ is a straightforward consequence of the semi-locality relation \eqref{eq_semilocality}. The next identity is obtained by studying the following 3-point correlation function:
\[
C_{a+1, k-a} = \avg{\psi^+(\infty)\psi_{a+1}(1)\psi^+_a(0)}
\]
However, utilizing the semi-locality relation \eqref{eq_semilocality}, we may alternatively write
\begin{align*}
C_{a+1, k-a}&=\avg{\psi^+(u)\psi_{a+1}(v)\psi_a^+(w)}(u-v)^{\gamma(k-1, a+1)}(u-w)^{\gamma(1, a)}(v-w)^{\gamma(a+1, k-1)}\\
&=\avg{\psi^+(u)\psi_a^+(w)\psi_{a+1}(v)}(u-v)^{\gamma(k-1, a+1)}(u-w)^{\gamma(1, a)}(w-v)^{\gamma(a+1, k-1)}\\
&=C_{k-1, k-a}
\end{align*}
In short, $C_{a+1, k-a}=C_{1,a}$. Putting these identities to use, we now show that $\tau_{\psi^+}(a)=C_{1,a}^2$. Note that $T_{\psi^+}$, the principal semi-invariant of $\psi^+$, is the minimal polynomial $\chi$. Therefore, we have (Recall that $g_0=1$ and $g_a =\prod_{j=0}^{a-1}C_{1,j}$)
\begin{equation}
	\tau_{\psi^+}(a)=\chi(1^{\times a+1}, 0^{\times k-a})=\frac{g_{a+1}g_{k-a}}{g_k}\avg{\psi^+(\infty)\psi_{a+1}(1)\psi^+_a(0)}=
	\frac{g_ag_{k-a}}{g_k} C_{1,a}C_{a+1,k-a}=\frac{g_ag_{k-a}}{g_k} C_{1,a}^2
\end{equation}
We will shortly prove that $g_{a}g_{k-a}=g_k$. Therefore, as a corollary of Theorem \ref{thm_polynomiality}, $C_{1,a}^2$ is a polynomial in $a$ of degree $2r$ which is divisible by $(a+1)(k-a)$, has the symmetry $C_{1,k-1-a}^2=C_{1,a}^2$, and $C_{1,0}^2=1$. Consequently, $C_{1,a}^2$ has roots at $-1, k$, and at some points $-t_p, t_p+k-1$ for $1\leq p\leq r-1$. In other words,
\begin{equation}
	\label{C1a2_eq}
	C_{1,a}^2 = \frac{(a+1)(k-a)}{k}\prod_{p=1}^{r-1}\frac{(t_p+a)(t_p+k-1-a)}{t_p(t_p+k-1)}
\end{equation}
The parameters $t_1, t_2, \cdots, t_{r-1}$ are part of the classification data of the $\mathbb{Z}_k^{(r)}$-algebra. When $k+1\geq 2r$, these parameters are generically distinct. When $k+1<2r$ some of these roots will coincide.

To find the rest of the structure constants and obtain even more identities between the structure constant, we will utilize the trivial wavefunction:
\begin{equation}
	\Psi_1(z_1, \cdots, z_k) = \frac{1}{g_k}\avg{\psi(z_1)\cdots \psi(z_k)}\prod_{i<j}(z_i-z_j)^{\gamma}
\end{equation}
Using the OPEs, by performing a sequence of limits, it is straightforward to show that
\begin{equation}
	\Psi_1(x^{\times a}, y^{\times k-a})=\frac{g_a g_{k-a}}{g_k}\avg{\psi_a(x)\psi_a^+(y)}(x-y)^{a(k-a)\gamma}
	=
	\frac{g_a g_{k-a}}{g_k}
	\label{eq:1st_id}
\end{equation}
and
\begin{equation}
	\begin{aligned}
		\Psi_1(x^{\times k-a-b}, y^{\times a}, z^{\times b})&=\frac{g_{k-a-b}g_a g_{b}}{g_k}\avg{\psi^+_{a+b}(x)\psi_a(y)\psi_b(z)}(x-y)^{(k-a-b)a\gamma}
		(x-z)^{(k-a-b)b\gamma}
		(y-z)^{ab\gamma} \\
		&= 
		\frac{g_{k-a-b}g_{a+b}}{g_k}
		\frac{g_a g_{b}}{g_{a+b}}C_{a,b}
	\end{aligned}
	\label{eq:2nd_id}
\end{equation}
At the same time, the trivial wavefunction $\Psi_1$ is a $(k,0)$ uniform state; hence it is a constant! Using the $a=0$ special case of Eq. \eqref{eq:1st_id} we find that $\Psi_1=1$. Consequently, the identity \eqref{eq:1st_id} turns into $g_ag_{k-a}=g_k$ for all $a$, and the identity \eqref{eq:2nd_id} morphs into
\begin{equation}
	C_{a,b} = \frac{g_{a+b}}{g_ag_b} = \prod_{j=1}^{b-1}\frac{C_{1,a+j}}{C_{1,j}}
	\label{eq:hidden}
\end{equation}
Despite $\Psi_1$ being trivial as a polynomial, it results in a highly non-trivial constraint between the structure constants. Combining Eqs. \eqref{C1a2_eq} and the constraint \eqref{eq:hidden}, we find (Notation: $\Gamma$ is gamma function)
\begin{equation}
	C_{a,b}^2 = \frac{(a+b)!(k-a)!(k-b)!}{k!a!b!(k-a-b)!}
	\prod_{p=1}^{r-1}\frac{\Gamma(t_p+a+b)\Gamma(t_m+k-a)\Gamma(t_m+k-b)\Gamma(t_m)}{\Gamma(t_m+k-a-b)\Gamma(t_m+a)\Gamma(t_m+b)\Gamma(t_m+k)}
\end{equation}
Note that this is a straightforward generalization of the structure constants of $\mathbb{Z}_k^{(1)}$ and $\mathbb{Z}_k^{(2)}$ algebras first reported in Ref. \cite{Zamolodchikov_Fateev_Parafermion}. A more symmetric formula can be provided for $C^2_{a,b,c}=\avg{\psi_a(\infty)\psi_b(1)\psi_c(0)}$ with $a+b+c=k$:
\begin{equation}
	\label{Cabc_eq}
	C_{a,b,c}^2 = \frac{(a+b)!(a+c)!(b+c)!}{a!b!c!(a+b+c)!}
	\prod_{p=1}^{r-1}\frac{(t_p)_{a+b}(t_p)_{a+c}(t_p)_{b+c}}{(t_p)_{a}(t_p)_{b}(t_p)_{c}(t_p)_{a+b+c}}
\end{equation}
where $(x)_n=x(x+1)\cdots (x+n-1)$ is the Pochhammer symbol. In what follows, we will also use the notation $(x)_n^-=x(x-1)\cdots (x-n+1)$ for the falling factorial.

\subsection{Polynomiality of Chiral Weights}
\label{sec_poly_chiral}
Let $W$ be a quasi-primary field of scaling dimension $\Delta$ in the chiral algebra $\mathcal{W}$ (See \ref{appendix_chiral} for a review). The fact that $\psi_a$ is a $\mathcal{W}$-primary field is reflected in the OPE of $W(z)\psi_a(w)$ as:
\begin{equation}
	W(z)\psi_a(w) = \frac{h_a^W\psi_a(w)}{(z-w)^\Delta}+\frac{(W_{-1}\psi_a)(w)}{(z-w)^{\Delta-1}}+\cdots
\end{equation}
The quantity $h_a^W$ is called the \emph{$W$-weight} of $\psi_a$. In the polynomiality formalism, if we choose $\phi=W$ (with $\beta(\phi)=0$ and $l(\phi)=\Delta$), the above OPE can be used to determine the specialization $\tau_W(a)$:
\begin{equation}
	\tau_W(a) =
	T_W(1^{\times a}, 0^{\times k-a})=\avg{W(\infty)\psi_a(1)\psi^+_{a}(0)}=h_a^W
\end{equation}
This shows that the $W$-weight $h_a^W$ is a polynomial of degree $\Delta$ in $a$ with the symmetry $h_{k-a}^W=(-1)^{\Delta}h_a^W$ and $h_0^W=0$. For example, when $\Delta=2$ (which means $W=L$) we rediscover $h_a=a(k-a)\gamma/2$. For a spin-$3$ field $W_3$ (which is necessarily a Virasoro primary), we have
\begin{equation}
	\label{eta_eq}
	w_a \equiv h_a^{W_3} = \frac{a(k-a)(k-2a)}{3}\eta
\end{equation}
where $\eta$ plays a similar role for $W_3$ that  $\gamma$ plays for $L$. 
\begin{rmk}
	It can be shown that the field corresponding to $\ket{\partial;W,a}\propto (2W_{-1}L_0-\Delta L_{-1}W_0)\ket{a}$ decouples from the $\mathbb{Z}_k^{(r)}$-algebra (see \ref{appendix_null}). In other words,  for $a+b< k$
	\begin{equation}
		\psi_a(z)\psi_b(w)=C_{a,b}\left(\psi_{a+b}(w)+
		\frac{a}{a+b}(z-w)\:\partial\psi_{a+b}(w)
		\right)+O(z-w)^2
	\end{equation}
	Put differently, the level one descendant $(W_{-1}\psi_a)$ is essentially $(\Delta h_a^W/2h_a)\partial \psi_a$.
\end{rmk}

\subsection{Conformal Blocks $\xi_a(x)$, Central Charge and Constraints on Chiral Weights}
\label{sec_central}
The specializations $\tau_\phi(a)$ are not the only important quantities one can compute from the semi-invariants $T_\phi$. As an example, consider the following specialization of $\chi$:
\begin{equation}
	\begin{aligned}
		\xi_a(x)\equiv&
		\chi(1,x^{\times a}, 0^{\times k-a})
		=\lim_{w\to \infty} w^{-2r(k-1)}\Psi_2(w^{\times k-1}, 1, x^{\times a},0^{\times k-a})\\
		=&\avg{\psi^+(\infty)\psi(1)\psi_a(x)\psi_a^+(0)}(1-x)^{a\gamma}x^{2h_a}
	\end{aligned}
\end{equation}
This relation recognizes $\xi_a(x)$, a specialization of the minimal polynomial, as a conformal block. Though this identity is not easy to work with, it nonetheless contains information about the CFT/wavefunction.

We start by writing down the OPE of $\psi_a\psi_a^+$. Let $\{W^i\}$ be a basis for the field space $\mathcal{W}$, and $\Delta_i$ the scaling dimension of $W^i$. The following formula is a special case of the formalism introduced in Ref. \cite{nahm} (also see \cite{blumenhagenCFT}):
\begin{equation}
	\psi_a(x)\psi_a^+(y) (x-y)^{2h_a}= \sum_{i}\widehat{h}^i_a\sum_{n\geq 0}
	\frac{(x-y)^{\Delta_i+n}}{n!}\frac{(\Delta_i)_n}{(2\Delta_i)_n}\partial^n W^i(w)
	\label{qOPE_eq}
\end{equation}
Here $\widehat{h}_a^i$ is defined via the relation
$h_a^i = \sum_{j}d_{ij}\widehat{h}^j_a\delta_{\Delta_i, \Delta_j}$ (with $\langle W^i(\infty)W^j(0)\rangle=d_{ij}\delta_{\Delta_i, \Delta_j}$). Using this OPE, the block $\xi_a(x)$ can be computed:
\begin{equation}
	\xi_a(x) = \avg{\psi^+(\infty)\psi(1)\psi_a(x)\psi_a^+(0)}(1-x)^{a\gamma}x^{2h_a}
	=
	(1-x)^{a\gamma}
	\left(
	1+
	\sum_{s\geq 2}
	\Xi_a^{\Delta}
	(-x)^{\Delta}{}_2F_1(\Delta, \Delta; 2\Delta;x)
	\right)
\end{equation}
where $\Xi^\Delta_a:=\sum_{i}\widehat{h}_a^i h^i
\delta_{\Delta_i,\Delta}$, and ${}_2F_1(a,b;c;x)$ is the Gauss hypergeometric function. The polynomiality of chiral weights results in polynomiality of $\Xi^\Delta_a$. To the first few orders, we can expand
\begin{equation}
	\xi_a(x)= 1-a\gamma x+\left\{\frac{(a\gamma)_{2}^-}{2}+\Xi_a^2\right\}x^2-
	\left\{\frac{(a\gamma)_{3}^-}{6}
	+(a\gamma-1)\Xi_a^2+\Xi_a^3
	\right\}x^3+\cdots
\end{equation}
Assuming we are given the minimal polynomial $\chi$ beforehand, we can compute the coefficient of $x^2, x^3, x^4, \cdots$ in $\xi_a(x)$ from $\chi$. Knowing these coefficients allows us to inductively fix $\Xi_a^{\Delta}$ for all $\Delta$. Given our current knowledge about the internal structure of $\mathbb{Z}_k^{(r)}$, it is unclear if the minimal polynomial is capable of fixing all of the chiral weights. Nonetheless, through the above discussion, $\chi$ enforces constraints on the chiral weights of each scaling dimension $\Delta$ via $\Xi_a^\Delta$.

The above considerations can provide us with the central charge. Since $d_{LL}=c/2$, we find that $\widehat{h}_a^L=2h_a/c$ and $\Xi_a^2 = 2hh_a/c$. This, in turn, yields:
\[
\xi_1(x)=\chi(1,x,0^{\times k-a}) = 1 - \gamma x + \left(\frac{\gamma(\gamma-1)}{2}+\frac{2h^2}{c}\right)x^2+O(x^3)
\]
Therefore, one can read off the central charge $c$ from the coefficient of $x^2$ in the specialization $\chi(1, x, 0^{\times k-a})$.

\subsection{Central Charge \& First Few Chiral Weights of $\mathbb{Z}_k^{(1)}$-Algebra}
To conclude the polynomiality section,  we will rediscover the findings of Zamolodchikov-Fateev in Ref. \cite{Zamolodchikov_Fateev_Parafermion} regarding the $\mathbb{Z}_k^{(1)}$-algebras. We have already found the structure constants:
\begin{equation}
	C_{a,b}^2 = \frac{(a+b)!(k-a)!(k-b)!}{k!a!b!(k-a-b)!}
\end{equation}
We would use the fact that the chiral algebra $\mathcal{W}$ is generated by $L, W_3, \cdots, W_k$ for $\mathbb{Z}_k^{(1)}$-algebra. This is not required to find the central charge but allows us to find the $W_3$ and $W_4$-weights. For simplicity, we use the notation $W,U$ for $W_3, W_4$, and their chiral weights will be denoted by $w_a, u_a$ respectively. We use the normalization convention $d_{WW}=c/3, d_{UU}=c/4$. Other than $1, L, W, U$, up to level 4, we also have a quasi-primary chiral field $\Lambda = (LL)-\frac{3}{10}\partial^2 L$. The weight of $\Lambda$ is denoted by $\lambda_a$. It is straightforward to find
\begin{equation}
	\lambda_a = h_a\left(h_a+\frac{1}{5}\right), \qquad d_{\Lambda\Lambda}=\frac{(5c+22)c}{10}
\end{equation}

To compute the central charge and the weights, we first need to find the minimal polynomial $\chi$.
For $\mathbb{Z}_k^{(1)}$-algebras $\chi\in \mathcal{T}_{k+1}^2$ and there is exactly one such polynomial satisfying $\chi(1, 0^{\times k})=1$; namely
\begin{equation}
	\chi(z_1, \cdots, z_{k+1})=\frac{1}{k}\sum_{1\leq i<j\leq k+1}(z_i-z_j)^{2}
\end{equation}
Computing the specialization $\xi_a(x)$, we find
\begin{equation}
	\label{xik1}
	\xi_a(x)=\chi(1, x^{\times a}, 0^{\times k-a})=
	1-\frac{2a}{k}x + \frac{a(k-a+1)}{k}x^2
\end{equation}
On the other hand, by interpreting $\xi_a(x)$ as a conformal block, we can write
\begin{equation}
	\begin{aligned}
		\xi_a(x) = (1-x)^{a\gamma}&\left\{1+ \frac{hh_ax^{2}}{d_{LL}}{}_2F_1(2,2,4,x)
		-
		\frac{ww_ax^{3}}{d_{WW}}{}_2F_1(3,3,6,x)
		+\left[\frac{\lambda\lambda_a}{d_{\Lambda\Lambda}}+\frac{uu_a}{d_{UU}}\right]x^4
		{}_2F_1(4,4,8,x)+\cdots
		\right\}
	\end{aligned}
\end{equation}
Expanding in powers of $x$ and comparing with Eq. \eqref{xik1} we find
\begin{subequations}
	\begin{align}
		c&= \frac{2(k-1)}{k+2}\\
		w_a&=\frac{a(k-a)(k-2a)}{(k-1)(k-2)}w_1\\
		w_1&=\frac{(k-1)}{3}\sqrt{\frac{2(k-2)(3k+4)}{k^3(k+2)}}\\
		u_a&=
		\frac{a(k-a)}{(k-1)(k-3)}\left[
		(a-2)(k-a-2)
		+\frac{(k-16)}{2(k-2)}(a-1)(k-a-1)
		\right]u_1\\
		u_1&=(k-1)\sqrt{\frac{(k-2)(k-3)(2k+1)(3k+4)}{2k^4(k+2)(16k+17)}}
	\end{align}
\end{subequations}
In principle, if some information is known about the chiral algebra of $\mathbb{Z}_k^{(r)}$ algebras, a similar procedure could determine some of the chiral weights in the general case. However, we will refrain from doing so for $r>1$. This is because we have incomplete information about the $\mathcal{W}$-algebra in these cases. Moreover, the formulas for these weights are significantly more complicated for larger $r$.

\section{Anstaz for Principal $\mathbb{Z}_k^{(r)}$-Wavefunctions}
\label{sec_binary}
To determine the principal $\mathbb{Z}_k^{(r)}$-wavefunction $\Psi$, the idea is first to build a basis for the space of $(2k,2r)$ uniform states; which we denote by $\mathcal{U}_{2k}^{2r}$. If $\mathfrak{f}_m$ is such a basis, then we can write:
\[
\Psi_{\mathrm{ansatz}}(z_1, \cdots, z_{2k})= \frac{1}{k!^2}\sum_{m} \alpha_m \mathfrak{f}_m(z_1, \cdots, z_{2k})
\]
We call $\alpha_m$ the \emph{ansatz constants} or $\alpha$-constants of $\Psi$ (with respect to the basis $\mathfrak{f}_m$). Our design for the basis is an application of classical invariant theory, in particular, invariants of binary forms. We will spend some time in this section, briefly but adequately, introducing the necessary mathematics. The discussion on the basis alone has no intersection with conformal field theory. The CFT essence of the $\mathbb{Z}_k^{(r)}$-WFs leaves its mark on the ansatz constants, not the basis. In other words, the general program for finding principal wavefunctions is as follows:
\begin{enumerate}
	\item Fix some $r$ and construct a set of ``generators'' for $\mathcal{U}^{2r}=\bigoplus_{k\geq 0}\mathcal{U}_{2k}^{2r}$. \\[2pt]
	Despite the word ``generator'' needing clarification, this method yields an explicit basis $\mathfrak{f}_m$ for $\mathcal{U}_{2k}^{2r}$.
	
	\item Writing $\Psi_{\mathrm{ansatz}}=k!^{-2}\sum_m \alpha_m \mathfrak{f}_m$, fix $\alpha_m$ so that $\Psi_{\mathrm{ansatz}}$ matches the principal $\mathbb{Z}_k^{(r)}$-WF $\Psi$.\\[2pt]
	This is done by performing various specializations of $\Psi_{\mathrm{ansatz}}$ and relating the $\alpha$-constants to CFT data.
\end{enumerate}
In this paper, we will present the basis for $r=1,2,3,4$. The process of determination of the ansatz constants is fully carried out for $r=1,2$ (i.e., we explicitly find the principal wavefunction). For $r=3,4$, we cannot determine all of the $\alpha$-constants. It is unclear if this is because $\chi\xrightarrow{\mathrm{fix}} \Psi$ fails in this case, or if there is extra structure (currently unknown to us) in $\mathbb{Z}_k^{(r)}$-algebras ($r>2$) that allows us to fix the rest of the $\alpha$-constants. The development of tools for finding the $\alpha$-constants, together with the explicit computation of principal wavefunctions, is the subject of section \ref{sec_about}.

\subsection{Invariants of Binary Forms}
A \emph{binary $n$-form} (or a binary $n$-ic) is a degree $n$ homogeneous polynomial in two formal variables $\mathbf{X}=(X, Y)^t$ (hence the name \emph{binary}). By convention, we write this polynomial in the following fashion:
\begin{equation}
Q_n(\mathbf{X}\mid \mathbf{a}) = \sum_{j=0}^n \binom{n}{j}a_j X^j Y^{n-j}
\end{equation}
We call $\vv{a}=(a_0,a_1, \cdots,a_n)^t$ the  \emph{coefficient vector} of the binary $n$-ic. The group $\mathrm{GL}_2(\mathbb{C})$ naturally acts on $\mathbf{X}$ via matrix multiplication (i.e. for $M\in \mathrm{GL}_2(\mathbb{C})$, we have $\mathbf{X}\mapsto \mathbf{X}'=M\mathbf{X}$). Using the binary $n$-ic, we can find a linear representation $R:\mathrm{GL}_2(\mathbb{C})\to \mathrm{GL}_{n+1}(\mathbb{C})$ which acts on the coefficients. This is done via the following implicit equality
\begin{equation}
\forall M\in \mathrm{GL}_2(\mathbb{C}):\qquad Q_n(M\mathbf{X}\mid R_M\mathbf{a})=Q_n(\mathbf{X}\mid\mathbf{a})
\end{equation}
The representation $R$ is called the \emph{induced representation} for the coefficients. A polynomial in the coefficients $I(a_0, a_1, \cdots , a_n)\in \mathbb{C}[a_0, \cdots, a_n]$ is called an \emph{invariant of the binary $n$-ic} if for all $M\in \mathbf{GL}_2$ we have
\begin{equation}
	I(\vv{a}) = (\det M)^{\mathrm{wt}(I)}I(R_M\vv{a}), \qquad \mathrm{wt}(I)\in \mathbb{N}
\end{equation}
The quantity $\mathrm{wt}(I)$ is called the \emph{weight} of $I$. Working out the special case $M=\lambda \,\mathrm{Id}_{2}$ ($\mathrm{Id}_s$ the $s\times s$ identity matrix, $0\neq \lambda\in \mathbb{C}$), we have $R_M=\lambda^{-n} \,\mathrm{Id}_{n+1}$ and therefore $I$ is necessarily homogeneous. If $m$ is the degree of $I$, then $2\mathrm{wt}(I) = nm$. The space of invariants of binary $n$-ics with degree $m$ is denoted by $\mathcal{B}_n^m$. We have $\mathcal{B}_n^0=\mathbb{C}$ and $\mathcal{B}_n^{m_1}\mathcal{B}_n^{m_2}\subset \mathcal{B}_n^{m_1+m_2}$. Consequently, the direct sum $\mathcal{B}_n=\bigoplus_{m\geq 0}\mathcal{B}_n^m$ is a \emph{graded $\mathbb{C}$-algebra} called the \emph{algebra of invariants of the binary $n$-ic}.

\paragraph{Example: Binary Quadratics} Utilizing the definitions, the binary 2-form/quadratic is $Q_2(X,Y) = a_0 X^2 + 2a_1 XY+a_2Y^2$. In this case, the induced representation of $M\in \mathrm{GL}_2$ can be found to be
\begin{equation}
M=
\begin{pmatrix}
\alpha & \beta\\
\gamma & \delta
\end{pmatrix}
\mapsto
R_M=
\frac{1}{(\alpha \delta - \beta \gamma)^2}
\begin{pmatrix}
\alpha^2 &- 2\alpha \beta & \beta^2\\
-\alpha\gamma & \alpha\delta+\beta\gamma & -\beta\delta \\
\gamma^2 & -2\gamma \delta &\delta^2
\end{pmatrix}
\end{equation}
The \emph{discriminant} $\Delta= a_0a_2-a_1^2$ is an  invariant of the binary quadratic ($\deg\Delta = \mathrm{wt}(\Delta) = 2$) and has the smallest possible degree. It can be shown (see the remark in \ref{appendix_direct}) that $\mathcal{B}_2=\mathbb{C}[\Delta]$ (as a graded algebra); in other words, the invariants of the binary quadratic are $1, \Delta, \Delta^2, \cdots$. The invariant $\Delta^k$ will be used to construct the principal $(k,1)$-wavefunction (aka $\mathbb{Z}_k$ Read-Rezayi state for $N=2k$).

\paragraph{Example: Binary Cubics} The binary cubic form is $Q_3=a_0Y^3+3a_1XY^2+3a_2XY+a_3X^3$. Similar to binary quadratics, the algebra $\mathcal{B}_3$ is generated by the cubic discriminant $\Delta_3$ (see \ref{appendix_direct}); i.e. up to scaling, the only invariants of binary cubics are $1, \Delta_3, \Delta_3^2, \cdots$. The cubic discriminant is of degree $4$ (weight $6$) and is given by
\begin{equation}
\label{cubic_dis_eq}
\Delta_3 = a_0^2a_3^2 - 6 a_0a_1a_2a_3+4a_0a_2^3-3a_1^2a_2^2+4a_1^3a_3
\end{equation}
Compactly, we have $\mathcal{B}_3=\mathbb{C}[\Delta_3]$.

\subsection{First vs Second Quantization: The Isomorphism  $\mathcal{B}_n^m\simeq \mathcal{U}_m^n$}
\label{sec_iso}
To construct an ansatz for principal $\mathbb{Z}_k^{(r)}$-wavefunctions, we will first construct an isomorphism between the spaces $\mathcal{B}_n^m$ and $\mathcal{U}_m^n$. As we will describe in this subsection, this isomorphism is a restriction of the isomorphism between first quantization (uniform states) to second quantization (binary invariants). Let us begin by exploring this quantum language in the context of states lying in the lowest Landau level.

\subsubsection{First and Second Quantization}
The lowest Landau level consists of $N_\phi+1$ `orbitals' $\ket{m}$ with $0\leq m\leq N_\phi$. The $m$th orbital carries a $L_z$-angular momentum $m$. Let $\nu_m$ be the number of bosons occupying the $m$th orbital and $N=\nu_0+\nu_1+\cdots+\nu_{N_\phi}$ the total number of particles. We often abbreviate all of the information in a single array $\nu=(\nu_0, \nu_1, \cdots, \nu_\phi)$, referred to as the bosonic occupation. In the second quantized language, the state corresponding to $\nu$ is represented as
\begin{equation}
	\ket{\nu}= \: a_0^{\nu_0}a_1^{\nu_1}\cdots a_{N_\phi}^{\nu_{N_\phi}}\ket{0}\equiv \mathbf{a}^{\nu}\ket{0}
\end{equation}
Here, $a_m$ is interpreted as the operator creating a boson at $m$th orbital. The first quantized formulation can now be obtained as follows. Let $N_{m}=\nu_0+\nu_1+\cdots+\nu_{m-1}$ be the number of particle with angular momentum $<m$ (with $N_0=0$). We define the corresponding \emph{symmetric monomial} as (caution: our convention is a bit different than the standard one)
\begin{equation}
	m_\nu(z_1,z_2,\cdots, z_N)=\mathscr{S}\left[\prod_{m=0}^{N_\phi}\prod_{j=1}^{\nu_m}z_{N_{m}+j}^{m}\right]
\end{equation}
where $\mathscr{S}$ is the symmetrization operator. Then, one finds that
\begin{equation}
	\braket{z_1, z_2, \cdots, z_N}{\nu} = \mathscr{C}(N,N_\phi)\:
	N!\:
	m_\nu(z_1, \cdots, z_N)
\end{equation}
where $\mathscr{C}(N,N_\phi)$ is a universal normalization (not dependent on $\nu$) which is irrelevant for our considerations. 

The inverse process, going from first quantization to second quantization, is called the \emph{umbral evaluation} in the literature on classical invariant theory (going back $\sim$ 150 years ago, way before the discovery of quantum mechanics, let alone quantum Hall effect). The main idea, which is purely algebraic, is as follows. Let $\mathscr{A}=\mathbb{C}[a_0, a_1, a_2, \cdots]$ be the space of polynomials in infinite variables. The umbral evaluation is the map $\mathscr{U}: \mathbb{C}[z_1, \cdots, z_N]\to \mathscr{A}$ linearly extending the following:
\begin{equation}
	\mathscr{U}[z_1^{p_1}z_2^{p_2}\cdots z_N^{p_N}]=a_{p_1}a_{p_2}\cdots a_{p_N}
\end{equation}
As a consequence, we find that $\mathscr{U}[m_\nu(z_1, \cdots, z_N)]=N!\:\mathbf{a}^{\nu}$, where $\nu=(\nu_0, \cdots, \nu_{N_\phi})$ is any occupation (i.e. $\sum_i \nu_i=N$) in an LLL with $N_\phi$ flux quanta. As we next explain, this first/second quantization isomorphism descends to an isomorphism $\mathcal{Q}: \mathcal{U}_N^{N_\phi}\xrightarrow{\sim} \mathcal{B}_{N_\phi}^N$. The symbol $\mathcal{Q}$ is reminder that the isomorphism is a `quantization' (from first to second). However, it is often easier to work with $\mathcal{Q}^{-1}$ (a `de-quantization' if you will).

\subsubsection{Symmetrized/Umbralized Graph Polynomials}
Let us start by introducing \emph{(regular)-graph-monomials}. In a previous work \cite{pakatchi2018quantum}, we studied the connection between graph theory and model FQH ground states in some depth. Here, we will review the basic definitions. For our purposes, a (weighted directed) graph $G$ is a pair of data $(V,\mu)$. The set $V$ is called the node-set, and we take it to be $V=\{1,2,\cdots, m\}$. We often use the notation $|G|$ (instead of $|V|$) for the number of nodes in $G$. The function $\mu: V\times V\to \mathbb{Z}_{\geq}$ is the multiplicity function of $G$. We understand $\mu(x,y)$ as the multiplicity of an arrow $x\to y$ in graph $G$. We further specialize to graphs where $\mu(i,j)\mu(j,i)=0$ for any pair $i,j\in V$ (i.e. all arrows between $i,j$ go in the same direction). A graph $G$ is \emph{$n$-regular} (or regular of degree $n$) if for all $j\in V$ we have
\begin{equation}
	n = \sum_{i\in V}\mu(i,j)+\mu(j,i)
\end{equation}
Given an $n$-regular graph $G$, we define its \emph{graph-monomial} as
\begin{equation}
	\mathscr{M}[G](z_1, \cdots, z_m)=\prod_{i,j=1}^m(z_i-z_j)^{\mu(i,j)}
\end{equation}
Suppose $G$ is such that $I=(\mathscr{U}\circ\mathscr{M})[G]\neq 0$. Then $I$ is a polynomial of the form
\[
I=\sum_{\substack{\nu_0+\nu_1+\cdots+\nu_n=m\\
		\nu_1+2\nu_2+\cdots +n\nu_n=nm/2
}}
N_{\vec{\nu}}[G]
\prod_{j=0}^{n}a_j^{\nu_j}
\]
with $N_{\vec{\nu}}[G]\in \mathbb{Z}$. For future convenience, we define $u(G)=\mathrm{gcd}(N_{\vec{\nu}}[G])$ (greatest common divisor of all coefficients). These are combinatorical factors that reflect the symmetries of $G$. If $(\mathscr{U}\circ\mathscr{M})[G]=0$, let $u(G)=\infty$. It is straightforward to show that $u(G_1\sqcup G_2)=u(G_1)u(G_2)$ with $\sqcup$ being the disjoint union. We define
\begin{equation}
	\Upsilon[G] = \frac{1}{u(G)}(\mathscr{U}\circ \mathscr{M})[G], \qquad
	\Psi[G] = \frac{1}{u(G)}(\mathscr{S}\circ \mathscr{M})[G]
\end{equation}
called \emph{umbralized graph polynomial} (UMP) and \emph{symmetrized graph monomial} (SGP) respectively. While using the symbol $\Psi$ for SGP is an abuse of language, we use it to remind ourselves that $\Psi[G]$ is wavefunction. Clearly, $\mathcal{Q}^{-1}(\Upsilon[G])=\Psi[G]$ ($\mathcal{Q}$ being the `quantization' isomorphism). This can be extended to an isomorphism $\mathcal{Q}^{-1}:\mathcal{B}_n^m\xrightarrow{\sim} \mathcal{U}_m^n$ as a corollary to the following theorem:
\begin{thm} \label{thm_iso}
	Throughout, all graphs are $n$-regular in $m$ nodes and sums are finite. The following are true:
	\begin{enumerate}
		\item For any invariant $I\in \mathcal{B}_n^m$, there exists graphs $G_i$ and  $t_i\in \mathbb{C}$ such that $X=\sum_{i} t_i \Upsilon[G_i]$.
		
		\item For any uniform state $X\in \mathcal{U}_m^n$, there exists graphs $G_i$ and  $t_i\in \mathbb{C}$ such that $X=\sum_{i} t_i \Psi[G_i]$.
		
		\item Let $P=\sum_i t_i \mathscr{M}[G_i]$ for some graphs $G_i$ and $t_i\in \mathbb{C}$. Then $\mathscr{U}[P]=0$ if and only if $\mathscr{S}[P]=0$.
	\end{enumerate}
\end{thm}
\noindent
For a proof of part 1, see \cite[Theorem 3.1 -- Part I]{kungrota} or 	
\cite[Ch. 6--7]{olverinv}. For a proof of part 2, see \cite[\S 88--89]{elliott} or \cite{pakatchi2018quantum}. Finally, part 3 can be found in \cite[Lemma 3.4 -- Part I]{kungrota}. We emphasize that regular graphs (and their graph-monomials) act as a middleman in the isomorphism $\mathcal{Q}^{-1}:\mathcal{B}_n^m\xrightarrow{\sim} \mathcal{U}_m^n$. The relation between the three concepts is shown in the diagram  \eqref{fig_diag}.
\begin{figure}[t]
	\centering
	\includegraphics{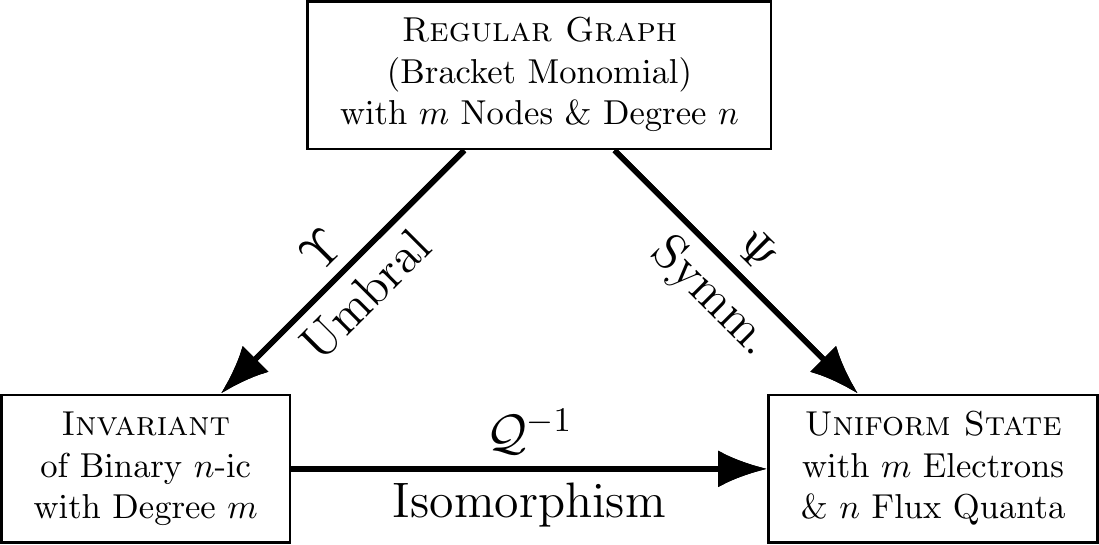}
	\caption{This diagram symbolically illustrates the relation between regular graphs, binary invariants, and uniform states.}
	\label{fig_diag}
\end{figure}

\begin{rmk}
	There exists yet another isomorphism $\mathscr{E}:\mathcal{U}_m^n\xrightarrow{\sim}\mathcal{B}_m^n$. In the literature on binary invariants, the combination of the two isomorphisms $\mathscr{E}\circ\mathcal{Q}^{-1}: \mathcal{B}_n^m\xrightarrow{\sim} \mathcal{B}_m^n$ is called \emph{Hermite reciprocity theorem}. We do not need this isomorphism in the current paper. Nonetheless, since this secondary mapping might also be relevant for future research, the details are provided in \ref{appendix_direct}.
\end{rmk}

\subsection{About the Structure of the $\mathcal{B}_n$ Algebra}
The algebra of invariants of binary forms is very well-studied and quite old. In 1868, Gordan proved $\mathcal{B}_n$ is finitely generated \cite{gordan}, though not utilizing the language of graded algebras. About twenty years later, Hilbert, in two revolutionary papers \cite{hilbert1, hilbert2}, not only proved that $\mathcal{B}_n$ is finitely generated as a graded algebra but also pioneered some of the foundations of modern commutative algebra. In his honor, the generators of $\mathcal{B}_n$ are called a \emph{Hilbert basis} (though the Hilbert basis is a set of generators, not a basis). Additionally, in Ref. \cite{hochster}, the authors proved that $\mathcal{B}_n$ algebras are Cohen-Macaulay. This subsection will briefly discuss some of the consequences of these properties that are relevant to us. To avoid breaking continuity, we do not include the technical definitions in the main body of the paper. Instead, we review the notions used here in \ref{appendix_graded}. The key takeaway is the unique decomposition \eqref{invariant_eq_1}, in which $	\mathfrak{I}^{(b)}_{\nu}$ are a basis for $\mathcal{B}_n^{\deg F}$.

Throughout, the notations $D$ and $\mu$ are used for the Krull dimension and multiplicity of $\mathcal{B}_n$, respectively. Let $I_1, \cdots, I_{D}$ (with degrees $d_1, \cdots, d_{D}$ respectively) be a homogeneous system of parameters (hsop). The full set of generators (i.e. the Hilbert basis) is denoted by $I_1, \cdots, I_{D}, J_{1}, \cdots, J_{g}$. Since $\mathcal{B}_n$ is a finitely generated Cohen-Macaulay graded $\mathbb{C}$-algebra, there exists a set of invariants
\[
\mathcal{Z}=\{1, Z_1, Z_2, \cdots, Z_{\mu-1}\}
\]
such that $\mathcal{Z}$ is a basis for $\mathcal{B}_n$ as a $\mathbb{C}[I_1, \cdots, I_{D}]$-module. We have $J_j\in \mathcal{Z}$ for all $1\leq j\leq g$. In general, any element in $\mathcal{Z}$ is a monomial $\prod_{j=1}^{g} J_j^{a_j}$ for some integer sequence $a=(a_1, a_2, \cdots, a_{g})$. We use the notation $e_b\equiv \deg Z_b$ (with $Z_0\equiv 1$). Finally, any invariant $F\in \mathcal{B}_n$ has a \emph{unique} decomposition
\begin{equation}
	\label{invariant_eq_1}
	F = 
	\sum_{b=0}^{\mu-1}
	\:
	\sum_{\substack{\nu_1, \cdots, \nu_{D}\geq 0\\
			\nu_1d_1+\cdots+\nu_{D}d_{D}=\deg F-e_b}}
	\tilde{\alpha}^{(b)}_{\nu}
	\mathfrak{I}^{(b)}_{\nu}, \qquad 
	\mathfrak{I}^{(b)}_{\nu}:= Z_b\prod_{i=1}^{D} I_i^{\nu_i}, \qquad (\tilde{\alpha}^{(b)}_{\nu}\in \mathbb{C})
\end{equation}
Note that once the Hilbert basis is concretely constructed (and the invariants $Z_s$ are identified), the computation of $F$ reduces to finding the complex numbers $\tilde{\alpha}^{(b)}_{\nu}$. We can use this, in conjunction with the isomorphism $\mathcal{Q}^{-1}:\mathcal{B}_n^m\xrightarrow{\sim} \mathcal{U}_m^n$, to obtain an ansatz for uniforms states; in particular, an ansatz for principal $\mathbb{Z}_k^{(r)}$-WFs. Some further properties about $\mathcal{B}_n$ can be found in \ref{appendix_graded}.

Let us now describe the algebras of binary invariants relevant to us. From now on, if there is no danger of ambiguity, the notations $I_d$ and $J_d$ represent a basic invariant of degree $d$. The information we are about to present is the findings of various authors. However, for an alternative reference listing all that will follow, see Ref. \cite{draisma}.

\paragraph{Example: Binary Quadratics and Cubics} As discussed before, $\mathcal{B}_2=\mathbb{C}[\Delta]$ ($\deg \Delta=2$)
and 
$\mathcal{B}_3=\mathbb{C}[\Delta_3]$ ($\deg \Delta_3=4$), with $\Delta, \Delta_3$ the quadratic/cubic discriminants. The Poincaré series for these two examples are $P_2(t)=(1-t^2)^{-1}$ and $P_3(t)=(1-t^4)^{-1}$. 

\paragraph{Example: Binary Quartics (see \cite{olverinv})} The algebra $\mathcal{B}_4$ has a Hilbert basis consisting of two invariants $I_2, I_3$ of degree $2,3$ respectively. There are no relations between $I_2, I_3$ (meaning $I_2, I_3$ constitute an hsop). The Poincaré series is obtained to be
\begin{equation}
	P_4(t)=\frac{1}{(1-t^2)(1-t^3)}
\end{equation}

\paragraph{Example: Binary sextics (see \cite{gordan})} The algebra $\mathcal{B}_{6}$ has a Hilbert basis $I_2, I_4, I_{6}, I_{10}$ (hsop), together with $J_{15}$. One computes $\mu(6)=2$, meaning, $J_{15}$ squares to a polynomial in $I_2, I_4, I_6, I_{10}$. In other words, $\mathcal{B}_6=(1+J_{15})\mathbb{C}[I_2, I_4, I_6, I_{10}]$ and the Poincaré series is
\begin{equation}
	P_{6}(t)= \frac{1+t^{15}}{(1-t^2)(1-t^4)(1-t^6)(1-t^{10})}
\end{equation}
\paragraph{Example: Binary octavics (see \cite{shioda})} The algebra $\mathcal{B}_8$ has a Hilbert basis $I_2,I_3, I_4, I_5, I_6, I_7$ (hsop) and $J_{8},J_{9}, J_{10}$. The Poincaré series of the octavics is
\begin{equation}
	P_8(t)=\frac{1+t^8 + t^9 + t^{10}+t^{18}}{(1-t^2)(1-t^3)(1-t^4)(1-t^5)(1-t^6)(1-t^7)}
\end{equation}
The space $\overline{\mathcal{B}}_8=\mathcal{B}_8/(\mathrm{hsop})$ is $5$-dimensional (i.e. $\mu(8)=5$), and has a basis $1, J_8, J_9, J_{10}, J_{8}J_{10}$ (that is with our choice of Hilbert basis in \S\ref{sec_octavic}). Note that $J_8J_{10}$ is not basic.

\subsection{A Graphic Ansatz for Principal $\mathbb{Z}_k^{(r)}$-Wavefunctions}
Irrespective of which $n$ is chosen, it is always possible to choose the Hilbert basis of $\mathcal{B}_n$ to be graphic. Concretely, if $D$ is Krull dimension of $\mathcal{B}_n$, one can find $D$ $n$-regular graphs $G_1, G_2, \cdots, G_{D}$ such that $\mathcal{I}=\{I_i\equiv \Upsilon(G_i)\mid 1\leq i\leq D\}$ is a homogeneous system of parameters. One can then augment the hsop with $g$ $n$-regular graphs $G'_1, \cdots, G'_{g}$ so that, with $\mathcal{J}=
\{J_i\equiv \Upsilon(G'_i)\mid 1\leq i\leq g\}
$, the union $\mathcal{I}\cup\mathcal{J}$ constitutes a full Hilbert basis for $\mathcal{B}_n$. This is called a  \emph{graphic} (Hilbert) basis. Similarly, let $H_0,H_1, \cdots, H_{\mu-1}$ be $n$-regular graphs such that
\[
\mathcal{Z}=\{\Upsilon(H_i)\mid 0\leq i\leq \mu-1\}
\]
is a basis for $\mathcal{B}_n$ as a $C[I_1, \cdots, I_D]$-module. We have $H_0=\emptyset$ (no nodes in $H_0$). The $H_i$ graphs are some disjoint unions of $G'_i$ graphs. In accordance with our previous notations, we denote the number of nodes of $G_i, G'_i$, and $H_i$ by $d_i, d'_i$, and $e_i$ respectively.

The principal $\mathbb{Z}_k^{(r)}$-wavefunction $\Psi$ is a $(2k,2r)$ uniform state. Therefore, $\mathcal{Q}(\Psi)$ is an element of $\mathcal{B}_{2r}^{2k}$. Using the graphic Hilbert basis of $\mathcal{B}_{2r}$, we can write $\Psi$ in the form
\begin{equation}
	\Psi(z_1,\cdots, z_{2k}) = \frac{1}{k!^2}\sum_{b=0}^{\mu-1}\sum_{\substack{\nu_1,\cdots, \nu_D\geq 0\\
			\nu_1 d_1+\cdots+\nu_D d_D=2k-e_b}} \alpha^{(b)}_{\nu} \mathfrak{f}_\nu^{(b)}(z_1,\cdots, z_{2k}), \qquad
		\mathfrak{f}_\nu^{(b)}:=\Psi\left[
	H_b\sqcup \bigsqcup_{i=1}^D G_i^{\sqcup \nu_i}
	\right]
\end{equation}
In particular, $\mathfrak{f}_\nu^{(b)}$ is a basis for $\mathcal{U}_{2k}^{2r}$. The above is called the ansatz for principal $\mathbb{Z}_k^{(r)}$-WFs. we call the complex numbers $\alpha^{(b)}_\nu$ the \emph{ansatz constant} or $\alpha$-constants of $\Psi$. We may also find the corresponding invariant:
\begin{equation}
	\label{eq_ansatz}
	\mathcal{Q}(\Psi) = \frac{1}{k!^2}\sum_{b=0}^{\mu-1}\sum_{\substack{\nu_1,\cdots, \nu_D\geq 0\\
			\nu_1 d_1+\cdots+\nu_D d_D=2k-e_b}} \alpha^{(b)}_{\nu} \Upsilon[H_b] \prod_{i=1}^D \Upsilon[G_i]^{\nu_i}
\end{equation}
 In general, our strategy in \S\ref{sec_about} would be to first find a (convenient) graphic Hilbert basis for $\mathcal{B}_{2r}$, and then try to find as many of the $\alpha$-constants of the principal $\mathbb{Z}_k^{(r)}$-wavefunction as we can.

\paragraph{Example: Binary Quadratics} We have seen before that the invariants of binary quadratics are generated by the quadratic discriminant $\Delta = a_0a_2-a_1^2$. Graphically, we have
\begin{equation}
	\label{disc_eq}
	\Delta \equiv \frac{1}{2}(\mathscr{U}\circ \mathscr{M})[G]\equiv \frac{1}{2} (\mathscr{U}\circ \mathscr{M})\left[
	\vcenter{
		\hbox{
			\includegraphics[scale=.33]{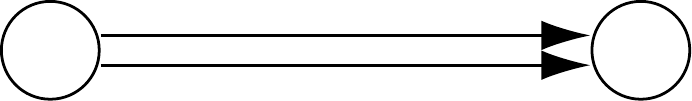}
		}
	}
	\right]= \frac{1}{2}\mathscr{U}[(z_1-z_2)^2]
\end{equation}
For this simple example, let us explicitly work out the effect of $\mathscr{U}$-operator:
\[
\mathscr{U}[(z_1-z_2)^2]=
\mathscr{U}[
z_1^2z_2^0+z_1^0z_2^2-2z_1z_2
]=a_2a_0+a_0a_2-2a_1a_1 = 2(a_0a_2-a_1^2)
\]
Therefore $u(G)=2$, and $\Delta=\Upsilon[G]$. The principal $(k,1)$-wavefunction is then equal to:
\begin{equation}
	\Psi^{(k,1)}(z_1, \cdots, z_{2k}) = \frac{1}{k!^2}\Psi\left[\vcenter{
		\hbox{
			\includegraphics[scale=.4]{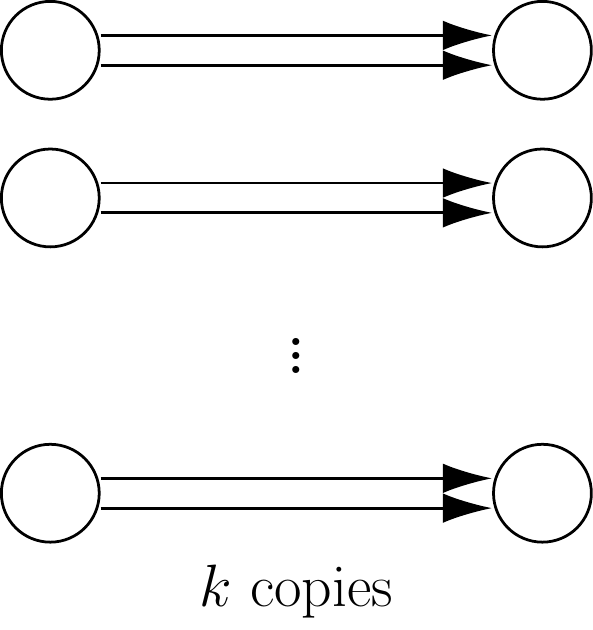}
		}
	}
	\right]=
	\frac{1}{2^k k!^2}\mathscr{S}\Big[
	\prod_{i=1}^k (z_i-z_{k+i})^2
	\Big]
\end{equation}
which is nothing but the principal (i.e. $N=2k$) $\mathbb{Z}_k$ Read-Rezayi state \cite{Read-Rezayi,cappelli}.

\paragraph{Example: Binary Cubics} To compute the principal wavefunctions of $\mathbb{Z}_k^{(r)}$ (with our particular definition of it), we only need to study invariants of binary $2r$-ics (i.e. binary \emph{even}-ics). Nonetheless, let us take a moment and discuss the simplest binary \emph{odd}-ic; namely the binary cubic. We have already seen that $\mathcal{B}_3$ is generated by the cubic discriminant $\Delta_3$ given by Eq. \eqref{cubic_dis_eq}. We have
\begin{equation}
	\Delta_3 \equiv \frac{1}{2}(\mathscr{U}\circ \mathscr{M})[G]\equiv \Upsilon\left[
	\vcenter{
		\hbox{
			\includegraphics[scale=.33]{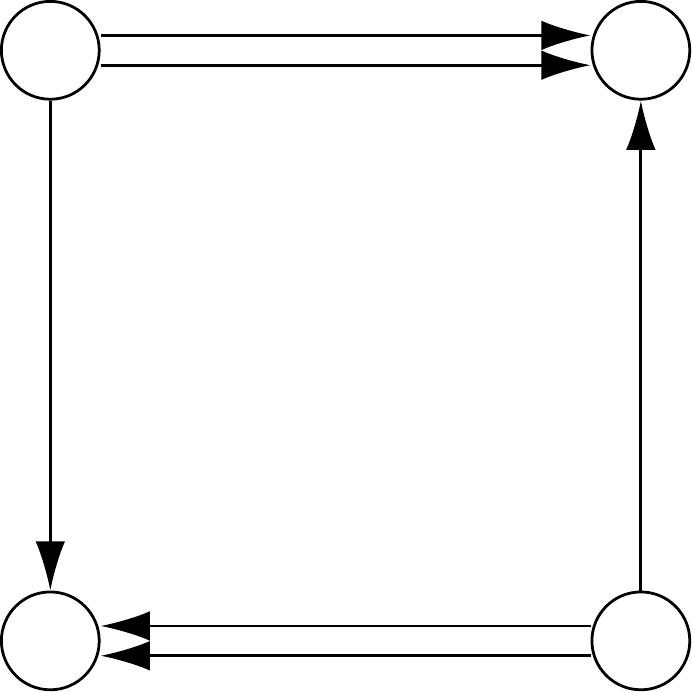}
		}
	}
	\right]
\end{equation}
The $(4,3)$ uniform state corresponding to $\Delta_3$ is the principal (i.e. $N=4$) Gaffnian wavefunction \cite{Gaffnian}. The $(4\ell, 3)$ uniform state corresponding to $\Delta_3^\ell$ is the principal (i.e. $N=4\ell$) Jack with parameter $\alpha = -(2\ell+1)/2$ \cite{Bernevig_Haldane_Model}.

\section{About Principal $\mathbb{Z}_k^{(r)}$-Wavefunctions with $2\leq r\leq 4$}
\label{sec_about}
In this section, we will present a graphic basis for $\mathcal{B}_4, \mathcal{B}_6$, and $\mathcal{B}_8$. In addition, we will develop some technical tools to find the ansatz constants of the principal $\mathbb{Z}_k^{(r)}$-WFs, with $r=2,3,4$, based on the free parameters of $\mathbb{Z}_k^{(r)}$-algebra. The first tool we will introduce is the tricoloring homomorphism (\S\ref{section_tri}), which is our own design. This machinery allows us to find a subset of $\alpha$-constants systematically. The second concept we introduce in this section is that of graph residue (\S\ref{sec_residue}). Residues help facilitate visualizing the link between principal $\mathbb{Z}_k^{(r)}$-wavefunctions and their corresponding minimal polynomial $\chi$. Finally, we will present the detailed discussion of principal $\mathbb{Z}_k^{(r)}$-wavefunctions for $r=2,3,4$ at the end.

\begin{rmk}
	In the main body of the paper, we always fix $r$ first and then study the entire $k$-family (i.e., all $k$ are allowed) simultaneously. It is also possible to fix $k$ and consider the $r$-family. The tools introduced in this section do not apply to the latter approach. In \ref{appendix_direct}, we present the special case $k=2$, i.e. the paired states ($k=1$ are the Laughlin states). 
\end{rmk}

\subsection{Tricoloring Homomorphism $K$}
\label{section_tri}
Let us start by fixing some terminology. Throughout, let $X,Y,Z$ be three formal variables. We use the notations $E_1=X+Y+Z$, $E_2=XY+XZ+YZ$, and $E_3=XYZ$;  the elementary symmetric polynomials in three variables. The ring of symmetric polynomials $\Lambda_3=\mathbb{C}[E_1, E_2, E_3]$ is to be understood as a graded algebra with degrees $\deg E_i=i$. The \emph{even subalgebra} is defined as $\mathcal{B}_{n}^E=\bigoplus_{k\geq 0}\mathcal{B}_n^{2k}$. We understand $\mathcal{B}_n^{2k}$ as the vector space of homogeneous elements of pseudo-degree $k$ (not $2k$). With $A,B$ being two graded algebras, we say a homomorphism $\rho:A\to B$ is \emph{degree-preserving} if $\deg_A x =\deg_B \rho(x)$ for any homogeneous element $x\in A$.

In this subsection, we will build a degree preserving homomorphism $K:\mathcal{B}_{2r}^E\to \Lambda_3$ called the \emph{tricoloring homomorphism}. Tricoloring homomorphism, at its core, is the generating function for 3-point specializations of the uniforms states. More concretely, let $P$ be an arbitrary $(2k,2r)$ uniform state. For any triple of integer $a,b,c\geq 0$ such that $a+b+c=k$ we define the specialization $P_{a,b,c}$ as
\begin{equation}
	\label{eq_sp}
	P_{a,b,c}:=\lim_{w\to \infty} w^{-2r(k-a)}P(w^{\times k-a}, 1^{\times k-b}, 0^{\times k-c})
\end{equation}
Due to how $P$ transforms under M\"{o}bius transformations (being a uniform state), it is straightforward to see that $P_{a,b,c}$ is symmetric with respect to the permutations of the indices. This allows us to construct a symmetric polynomial in $X,Y,Z$:
\begin{equation}
	K[P](X,Y,Z)=\sum_{\substack{a,b,c\geq 0\\
			a+b+c=k}}\frac{P_{a,b,c}}{(k-a)!(k-b)!(k-c)!}X^aY^bZ^c
\end{equation}
Clearly, $K[P]$ acts as a generating function for the specializations $P_{a,b,c}$. Note that $K[P]$ is homogeneous of degree $a+b+c=k$. By abuse of notation, let $K[I]\equiv K[\mathcal{Q}^{-1}(I)]$. So far, we have constructed a collection of linear mappings $\mathcal{B}_{2r}^{2k}\to \Lambda_3^k$ given by $I\mapsto K[I]$. To show that $K$ actually extends to a homomorphism $K:\mathcal{B}_{2r}^E\to \Lambda_3$ we need to explain the naming ``tricoloring''.

Let $G$ be a $2r$-regular graph with nodes $V={1,2,\cdots, |G|}$ (with $|G|$ a positive integer, not necessarily even). Let us illustrate the meaning of $K[G]\equiv K[\Psi[G]]$ in terms of graphs. A (proper) \emph{tricoloring} of $G$ is a map $\eta: V\to \{\infty, 1, 0\}$ so that no two adjacent nodes have the same color (the colors being $\infty,1, 0$). In contrast to the typical definition of graph-coloring, we assume the colors are ordered: $\infty>1>0$. Given an arrow $e=x\to y$ in $G$, the quantity $\sgn(e)=\mathrm{sgn}(\eta(x)-\eta(y))\in \{-1, +1\}$ is called the sign of $e$ induced by $\eta$. The sign of the tricoloring $\eta$ is then defined as  $\mathrm{sgn}(\eta)=\prod_{i,j=1}^{|G|}\mathrm{sgn}(i\to j)^{\mu(i,j)}$ (i.e. the multiplication of all arrow signs). If we denote by $v_\infty, v_1, v_0$ as the number of nodes colored by $\infty,1,0$ respectively, then 
\begin{equation}
	K[G]\equiv K[\Upsilon[G]]=
	\frac{1}{u(G)}\sum_{\eta}\mathrm{sgn}(\eta) X^{|G|/2-v_\infty(\eta)}
	Y^{|G|/2-v_1(\eta)}
	Z^{|G|/2-v_0(\eta)}
\end{equation}
Here, the sum is over all proper tricolorings $\eta$. 
$K[G]$ is homogeneous of degree $|G|/2$
as a function in $X,Y,Z$. A consequence of regularity is that $v_i\leq |G|/2$ for any of $i=\infty, 1, 0$. Therefore, when $|G|=2k$ is even, this is a polynomial, as expected. However, when $|G|$ is odd, it is $K[G]/\sqrt{E_3}$ that is a symmetric polynomial. An important property of $K$ is that $K[G_1\sqcup G_2]=K[G_1]K[G_2]$. In other words,
\begin{equation}
	K[\Upsilon[G_1]\Upsilon[G_2]]=K[\Upsilon[G_1\sqcup G_2]]=K[G_1\cup G_2]=K[G_1]K[G_2]
\end{equation}
Combining this observation with Theorem \ref{thm_iso} shows that $K:\mathcal{B}_{2r}^E\to \Lambda_3$ is indeed a degree-preserving homomorphism of graded algebras.

Though the setting of the definition of tricoloring homomorphism is general, it is the action of $K$ on the principal $\mathbb{Z}_k^{(r)}$-WFs that reveals their true strength. Going to the specializations \eqref{eq_sp}, with find
\begin{equation}
	\Psi_{a,b,c}=\lim_{w\to \infty} w^{-2r(k-a)}\Psi(w^{\times k-a}, 1^{\times k-b}, 0^{\times k-c}) = \frac{g_{k-a}g_{k-b}g_{k-c}}{g_k^2}
	\avg{\psi_a^+(\infty)\psi^+_b(1)\psi^+_c(0)}=C_{a,b,c}^2
\end{equation}
The formula for $C_{a,b,c}^2$ in Eq.\eqref{Cabc_eq} now yields
\begin{equation}
	\begin{aligned}
		\label{K_CFT_eq}
		k!^2 K[\Psi]&=
		k!^2 
		\sum_{\substack{a,b,c\geq 0\\
				a+b+c=k}}
		\frac{C_{a,b,c}^2}{(k-a)!(k-b)!(k-c)!}
		X^a Y^b Z^c\\
		&=
		\sum_{\substack{a,b,c\geq 0\\
				a+b+c=k}}
		\binom{a+b+c}{a,b,c}
		\prod_{p=1}^{r-1}\frac{(t_p)_{a+b}(t_p)_{a+c}(t_p)_{b+c}}{(t_p)_{a}(t_p)_{b}(t_p)_{c}(t_p)_{a+b+c}}
		X^{a}Y^b Z^c
	\end{aligned}
\end{equation}
On the other hand, using the graphic ansatz \eqref{eq_ansatz} we find ($G_i, H_b$, etc. being the same as before)
\begin{equation}
	\label{K_hilbert_eq}
	k!^2 K[\Psi]=\sum_{b=0}^{\mu-1}
	\:
	\sum_{
		\substack{\nu_1, \cdots,\nu_D\geq 0\\
			\nu_1 d_1+\cdots+\nu_D d_D=2k-e_b}}
	\alpha_{\nu}^{(b)}K[H_b]\prod_{i=1}^{D}
	K[G_i]^{\nu_i}
\end{equation}
This leads to what we call the \emph{tricoloring identity}:
\begin{equation}
	\label{tricoloring_identity}
	\sum_{\substack{a,b,c\geq 0\\
			a+b+c=k}}
	\binom{a+b+c}{a,b,c}
	\prod_{p=1}^{r-1}\frac{(t_p)_{a+b}(t_p)_{a+c}(t_p)_{b+c}}{(t_p)_{a}(t_p)_{b}(t_p)_{c}(t_p)_{a+b+c}}
	X^{a}Y^b Z^c
	=
	\sum_{b=0}^{\mu-1}
	\:
	\sum_{
		\substack{\nu_1, \cdots,\nu_D\geq 0\\
			\nu_1 d_1+\cdots+\nu_D d_D=2k-e_b}}
	\alpha_{\nu}^{(b)}K[H_b]\prod_{i=1}^{D}
	K[G_i]^{\nu_i}
\end{equation}
We will use this identity to determine a significant number of ansatz constants $\alpha_\nu^{(b)}$, though not necessarily all, in our analysis of principal $\mathbb{Z}_k^{(r)}$-wavefunctions with $r=2,3,4$.

\subsection{Graph Residue and Residual Basis}
\label{sec_residue}
As mentioned before, a partial reason to explicitly find the principal $\mathbb{Z}_k^{(r)}$ wavefunction $\Psi$ is to investigate whether $\chi\to \Psi_2$ is true. In such a study, it would be convenient to define a basis for $\mathcal{T}_{k+1}^{2r}$ that is naturally obtained from the Hilbert basis of $\mathcal{B}_{2r}$-algebras. We will develop this basis in this subsection.

In general, the minimal polynomial $\chi$ is a specialization of the principal wavefunction:
$$
\chi(z_1, \cdots, z_{k+1})=
\lim_{w\to \infty} w^{-2r(k-1)} \Psi(w^{\times (k-1)},z_1, \cdots, z_{k+1})$$
On the other hand, the principal wavefunction $\Psi$ is a superposition of symmetrized graph polynomials $\Psi[\Gamma]$, with $\Gamma$ a $2r$-regular graph in $2k$ nodes. Hence, we routinely need to calculate specializations of the form:
\begin{equation}
	\label{eq:chiG}
	\chi_{\Gamma}(z_1, \cdots, z_{k+1})\equiv\frac{1}{k!^2} \lim_{w\to \infty} w^{-2r(k-1)} \Psi[\Gamma](w^{\times k-1},z_1, \cdots, z_{k+1})
\end{equation}
We would like to reduce this algebraic calculation to a graph-theoretic one. To do so, we need to introduce a few notions.

\paragraph{Co-Independence Number}
Let $\Gamma=(V,\mu)$ be an arbitrary graph. A subset of nodes $S$ in $\Gamma$ is called an \emph{independent set} if no two nodes in $S$ are adjacent. The \emph{independent number} $\alpha(\Gamma)$ is defined as the size of the largest independent set. The independent sets $S$ with $|S|=\alpha(\Gamma)$ are called \emph{maximum}. We denote the number of maximum independent sets by  $m(\Gamma)$. Note that $\alpha(\Gamma_1\sqcup \Gamma_2)=\alpha(\Gamma_1)+\alpha(\Gamma_2)$ and $m(\Gamma_1\sqcup \Gamma_2)=m(\Gamma_1)m(\Gamma_2)$. While the definitions so far are standard, the following definition of \emph{co-independence number} $\varepsilon(\Gamma)$ is tailored for our specific purposes:
\begin{equation}
	\varepsilon(\Gamma)\equiv \frac{|\Gamma|}{2}-\alpha(\Gamma)
\end{equation}
Using this notion, if $\Gamma$ has $2k$ nodes, a necessary condition for $\chi_\Gamma\neq 0$ is $\varepsilon(\Gamma)=0,1$. If $\Gamma$ is a regular graph, its co-independence number is non-negative. The quantity $\varepsilon(\Gamma)$ is an integer/half-integer iff $|\Gamma|$ is even/odd respectively. Moreover, we have $\varepsilon(\Gamma_1\sqcup \Gamma_2)=\varepsilon(\Gamma_1)+\varepsilon(\Gamma_2)$.

\paragraph{Graph Residue} Suppose $\Gamma$ is a $2r$-regular graph with $\varepsilon(\Gamma)=1$, say in $2k$ nodes. Given any maximum independent set $S$ (i.e., $|S|=k-1$), the graph $\Gamma-S$ is obtained from $\Gamma$ by deleting the nodes in $S$ (and all arrows with one end in $S$). We also define $n(S)$ as the number of arrows\ of the form $(V-S)\to S$ (i.e., arrows ending in $S$). We say $\Gamma$ is reductive if for all maximum independent sets $S$:
\begin{enumerate}
	\item $n(S)=$even.
	\item There exists a graph $\Gamma^\star$ with $k+1$ nodes, called the \emph{residue} of $\Gamma$, such that $\Gamma-S$ is isomorphic to $\Gamma^\star$.
\end{enumerate}
If $\Gamma$ is reductive in $2k$ nodes, then the corresponding $\chi$ polynomial is relatively simple (normalization chosen for future convenience)
\begin{equation}
	\chi_\Gamma(z_1, \cdots, z_{k+1})=\frac{1}{kk!}\frac{m(\Gamma)}{u(\Gamma)}(\mathscr{S}\circ \mathscr{M})[\Gamma^\star]
\end{equation}
Moreover, since $\Gamma^\star$ has $2r$ arrows, this polynomial is an element of $\mathcal{T}_{k+1}^{2r}$ as expected. Note that if $\varepsilon(\Gamma)=1$, then $\chi_{\Gamma}(1, 0^{\times k})=0$.

Let us now study the concepts mentioned above in the context of the graphic bases of $\mathcal{B}_{2r}$. We want to establish three rules of thumb for a ``good'' graphic basis. We do not know if these rules can always be satisfied, but in our explicit construction for $r=2,3,4$, it is possible to abide by them. Suppose $\mathcal{B}_{2r}$ has $N(d)\neq 0$ basic invariants of degree $d$. We denote the corresponding graphs by $G_d^{(i)}$, with $1\leq i\leq N(d)$. We will omit the superindex $(i)$ if $N(d)=1$. Due to ansatz \eqref{eq_ansatz}, we have a special interest in graphs $\Gamma$ with an even number of nodes that are some disjoint union of basic graphs. We call such a graph \emph{constructible}. To get to the first rule, let us first recite a fact: The algebra $\mathcal{B}_{2r}$ always has exactly one basic invariant of degree two. Graphically, this invariant is represented by
\begin{equation}
	G_2 =
	\vcenter{
		\hbox{
			\includegraphics{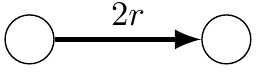}
		}
	}
\end{equation}
and clearly has $\varepsilon(G_2)=0$. For the first rule, we desire that ``$\varepsilon(G_d^{(i)})>0$ for all $d>2$''. This would make $G_2$ the only basic graph with a co-independence number of zero. As the second rule, we put in place: ``any constructible graph either has $\varepsilon> 1$ or is reductive''. Among all of the constructible graphs which do not have a component equal to $G_2$, only two types have co-independence number one; namely
\begin{subequations}
	\begin{align}
		&\mathcal{G}^{(i,j)}_{d_1,d_2} = G_{d_1}^{(i)}\sqcup
		G_{d_2}^{(j)}&& \text{if }\varepsilon(G_{d_1}^{(i)})=\varepsilon(G_{d_2}^{(j)})=\frac{1}{2}\\
		&\mathcal{G}^{(i)}_{d} = G_{d}^{(i)}&& \text{if }\varepsilon(G_{d}^{(i)})=1
	\end{align}
\end{subequations}
In general, any constructible graph $\Gamma$, which has $2k$ nodes and satisfies $\varepsilon(\Gamma)=1$, is of the form  $\Gamma=H\sqcup G_2^{\sqcup(k-l)}$, where $H$ is one of the above graphs satisfying $|H|=2l\leq 2k$. Computing the $\chi$ polynomial corresponding to $\Gamma^\star$ (recall that by hypothesis such $\Gamma$ is reductive)
\begin{equation}
	\chi_{H\sqcup G_2^{\sqcup (k-l)}}(z_1,\cdots, z_{k+1})
	= 
	\frac{1}{k(k)^-_{l}}
	\frac{m(H)}{u(H)}
	\sum_{i_1<i_2<\cdots<i_{l+1}}
	(\mathscr{S}\circ \mathscr{M})[
	H^\star
	](z_{i_1}, \cdots z_{i_{l+1}})
	\label{eq_chiH}
\end{equation}
reduces, basically, to the computation of symmetrized graph polynomial of the residue $H^\star$. From now on, if $H=\mathcal{G}^{(i,j)}_{d_1,d_2}$ (resp. $\mathcal{G}^{(i)}_{d}$), we use the notation $\chi_{d_1, d_2}^{(i,j)}$ (resp. $\chi_{d}^{(i)}$) instead of $\chi_{H\sqcup G_2^{\sqcup (k-l)}}$. There is one more constructible graph $\Gamma$ in $2k$ variables with non-vanishing $\chi_\Gamma$; namely $\Gamma=G_2^{\sqcup k}$. The $\chi_2\equiv \chi_{G_2^{\sqcup k}}$ has a simple universal form:
\begin{equation}
	\chi_2(z_1,\cdots, z_{k+1}) = \frac{1}{k}\sum_{i<j} (z_i-z_j)^{2r}
	\label{eq_chi2}
\end{equation}
This brings us to the third rule: ``For any $k\geq 1$, the collection
\begin{equation}
	\mathscr{B} = \{\chi_2\}\cup \{\chi_{d_1,d_2}^{(i,j)}\neq 0\mid d_1, d_2=\text{odd},\; d_1+d_2\leq 2k\}
	\cup \{\chi_{d}^{(i,j)}\neq 0\mid d=\text{even},\; 2<d\leq 2k\}
\end{equation}
is not redundant (i.e., there is no polynomial appearing twice), is linearly independent, and spans $\mathcal{T}_{k+1}^{2r}$.'' If these three rules can be satisfied, which is the case at least for $r=1,2,3,4$, then we can find a natural basis for $\mathcal{T}_{k+1}^{(2r)}$ directly from the Hilbert basis. We call this the \emph{residual basis} of $\mathcal{T}_{k+1}^{(2r)}$.

\begin{rmk}
	Due to universality of $\chi_2$ we can compute the specialization $\xi^2_{a}=\chi_2(1, x^{\times a}, 0^{\times k-a})$ in general:
	\begin{equation}
		\xi^2_{a}(x)=\frac{(a+1)(k-a)}{k}+\frac{a(k-a)}{k}\left(\frac{(x-1)^{2r}}{k-a}+x^{2r}-1\right)=1 - \frac{2ar}{k}x + \cdots
	\end{equation}
	We will use this formula on a few occasions.
\end{rmk}

\begin{rmk}
	Regarding the graph drawing conventions:\\
	Before starting our construction of Hilbert bases, we need to clarify our drawing convention for the graphs to appear. Let $G=(V,\mu)$ be a graph. Recall that we only care about those graphs which satisfy $\mu(i,j)\mu(j,i)=0$ for any pair $i,j\in V$. This means that all arrows shared between $i,j$ will have the same orientation. In our convention, if $i,j$ is such that $\mu(i,j)+\mu(j,i)=2m$, then we will 
	draw these as $2m$ \emph{undirected} edges between $i,j$. One is free to orient those edges in either of the two ways (as long as $\mu(i,j)\mu(j,i)=0$). We can afford this ambiguity since the invariant $\Upsilon[G]$ is insensitive to inverting the direction of an arrow with even multiplicity.
\end{rmk}

\subsection{Binary Quartics ($2r=4$) and the Principal $(k,2)$-Wavefunction}
\label{sec_quartic}

\paragraph{Hilbert Basis} The algebra of invariants of the binary quartics $\mathcal{B}_4$ has two basic invariants $I_2, I_3$ with degrees $2,3$ respectively. The graphs for these invariants are
\begin{equation}
	G_2 = 
	\vcenter{
		\hbox{
			\includegraphics{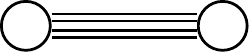}
		}
	}
	,\qquad  G_{3}=
	\vcenter{
		\hbox{
			\includegraphics{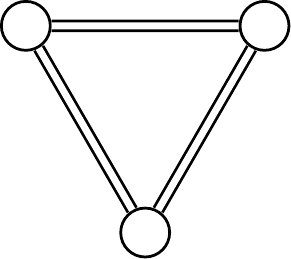}
		}
	}
\end{equation}
Computing the image of $G_2, G_3$ under the tricoloring homomorphism $K$ we have
\begin{equation}
	K[G_2]=E_1, \qquad  K[G_3]=\sqrt{E_3}
\end{equation}
The co-independence number of these graphs are $\varepsilon(G_2)=0$ (as expected) and $\varepsilon(G_3)=1/2$. We have $u(G_2)=m(G_2)=2$,  $u(G_3)=6$ and $m(G_3)=3$. Using the ansatz \eqref{eq_ansatz}, the principal $\mathbb{Z}_k^{(2)}$-wavefunction can be written as
\begin{equation}
	\Psi = \frac{1}{k!^2}\sum_{m=0}^{\lfloor k/3\rfloor} \alpha_m \Psi\left[
	\vcenter{
		\hbox{
			\includegraphics[scale=.75]{quartic_2}
		}
	}^{\sqcup(k-3m)}
	\bigsqcup
	\vcenter{
		\hbox{
			\includegraphics[scale=.75]{quartic_3}
		}
	}^{\sqcup(2m)}
	\right]
\end{equation}
We now proceed to find the constants $\alpha_m$.

\paragraph{Tricoloring Identity}
To determine the $\mathbb{C}$-numbers $\alpha_m$, we
put the tricoloring identity \eqref{tricoloring_identity} to use. We find that
\begin{equation}
	\label{accidental_eq}
	k!^2 K[\Psi] = \sum_{\substack{a,b,c\geq 0\\
			a+b+c=k}}
	\binom{k}{a,b,c}
	\frac{(t)_{k-a}(t)_{k-b}(t)_{k-c}}{(t)_{a}(t)_{b}(t)_{c}(t)_k}
	X^{a}Y^b Z^c=
	\sum_{m=0}^{\lfloor k/3\rfloor}
	\alpha_m E_1^{k-3m}E_3^m
\end{equation}
Expanding the right-hand side in powers of $X,Y,Z$, this identity reduces to (recall that $(x)_n^-=x(x-1)\cdots (x-n+1)$ is the falling factorial, while $(x)_n=x(x+1)\cdots (x+n-1)$ is the Pochhammer symbol)
\begin{equation}
	\label{place-holder_eq}
	\sum_{m=0}^b \alpha_m \frac{(b)_m^-(a)_m^-(k-a-b)_m^-}{(k)_{3m}^-}=\frac{(t+a)_b(t+k-b-a)_b}{(t)_b(t+k-b)_b}
\end{equation}
At the same time, due to the Saalsch\"{u}tz's theorem (see \cite[\S 2.2]{bailey1935hypergeometric}), for any non-negative integer $b$, we have the hypergeometric identity
\begin{equation}
	{}_3F_2(-b,-a, b+a-k; t, 1-k-t;1)\equiv
	\sum_{m=0}^b \binom{b}{m} \frac{(a)_m^-(k-a-b)_m^-}{(t)_m(t+k-1)_m^-}=
	\frac{(t+a)_b(t+k-a-b)_b}{(t)_b(t+k-b)_b}
\end{equation}
where ${}_3F_2(\alpha_1, \alpha_2; \alpha_3, \beta_1, \beta_2;x)$ is the generalized Gauss hypergeometric function. Consequently, the ansatz constants are found to be
\begin{equation}
	\alpha_m = \frac{1}{m!}\frac{(k)_{3m}^-}{(t)_m(t+k-1)_m^-}
\end{equation}
To summarize, the principal $(k,2)$-wavefunction is the following polynomial:
\begin{equation}
	\Psi=
	\frac{1}{k!}
	\sum_{m=0}^{\lfloor k/3\rfloor}\frac{1}{(k-3m)!m!}\frac{1}{(t)_m(t+k-1)_m^-}\:\Psi\left[
	\vcenter{
		\hbox{
			\includegraphics[scale=.75]{quartic_2}
		}
	}
	^{\sqcup (k-3m)}
	\sqcup
	\vcenter{
		\hbox{
			\includegraphics[scale=.75]{quartic_3}
		}
	}
	^{\sqcup 2m}
	\right]
\end{equation}
Note that $\varepsilon(G_2^{\sqcup k-3m}\sqcup G_{3,3}^m)=m$. So only the summands for $m=0,1$ contribute to the minimal polynomial of the $(k,2)$-algebra.
\paragraph{Minimal Polynomial} It is obvious from the high degree of symmetry of $G_3$ that it leads to reductive graphs. The residue of $G_{3,3}=G_3\sqcup G_3$ is then obtained to be
\begin{equation}
	G_{3,3}^\star =
	\vcenter{
		\hbox{
			\includegraphics[scale=.5]{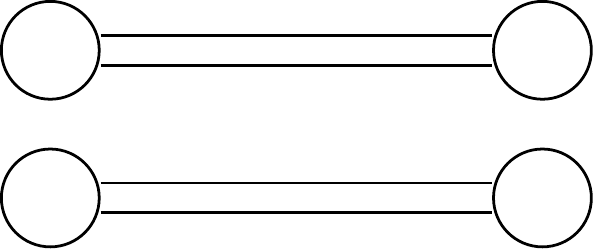}
		}
	}
\end{equation}
We have  $\chi_2(1, 0^{\times k})=1$ and $\chi_{3,3}(1, 0^{\times k})=0$. This, in particular, shows that $\chi_2, \chi_{3,3}$ are linearly independent. They also span $\mathcal{T}_{k+1}^{4}$ since $\dim \mathcal{T}_{k+1}^{4}=2$ (for $k\geq 3$). Thus, our three rules of thumb are satisfied. The minimal polynomial can now be written as
\[
\chi = \chi_2 + \beta \chi_{3,3}
\]
for some coefficient $\beta$. To find $\beta$ and the central charge $c$, consider the specialization $\xi_a=\chi(1,x^{\times a}, 0^{\times k-a})$. We calculate
\begin{align}
	&\xi_{a}^{2}\equiv \chi_{2}(1,x^{\times a}, 0^{\times k-a})=
	\frac{(a+1)(k-a)}{k}+\frac{a(k-a)}{k}\left(\frac{(x-1)^{4}}{k-a}+x^{4}-1\right)
	\label{xi22}\\
	&\xi_{a}^{3,3}\equiv \chi_{2}(1,x^{\times a}, 0^{\times k-a})=
	\frac{(a+1)_2^-(k-a)_2^-}{k(k)_3^-}x^2
	+\frac{(a)_2^-(k-a)_2^-}{k(k)_3^-}x^2\left(
	\frac{2(x-1)^2}{k-a-1}+x^2-1
	\right)
\end{align}
and $\xi_a(x) = \xi_a^2(x) + \beta \xi^{3,3}_a(x)$. Now, since $\xi_a(1)=C_{1,a}^2$, we can find $\beta$ as a function of $t$:
\begin{equation}
	\beta = \frac{(k)_3^-}{t(t+k-1)}
\end{equation}
On the other hand, interpreting $\xi_1(x)$ as a conformal block, we have $\xi_1(x)=1-\gamma x + [\gamma(\gamma-1)/2+2h^2/c] x^2+O(x^3)$. Therefore, as discussed in subsection \ref{sec_central}, since
\begin{equation}
	\xi_1(x) = 1 - \frac{4}{k}x+ \left(\frac{6}{k}+\frac{2(k-1)(k-2)}{k t (t+k-1)}\right)x^2-\frac{4}{k}x^3+x^4
\end{equation}
we can read off the central charge 
\begin{equation}
	c = \frac{4(k-1)t(t+k-1)}{(k+2t-2)(k+2t)}
\end{equation}
This is the same central charge Zamolodchikov-Fateev report in Ref. \cite{Zamolodchikov_Fateev_Parafermion}. One can, in principle, use the formula for $\xi_a(x)$, in conjunction with the conformal block interpretation of $\xi_a(x)$ to find certain constraints on the chiral weights of $\mathbb{Z}_k^{(2)}$-algebra. Nevertheless, due to the complexity of the formulas, we will refrain from doing so.

\begin{rmk}
	\label{remark_2}
	We have shown that, in the case of $\mathbb{Z}_k^{(2)}$-algebras, both the minimal polynomial and the principal wavefunction is parametrized by $t$. In other words, both polynomials carry the same information. The major agents responsible for making this equivalence happen are (1) the constraint $C_{a,b}=g_{a+b}/g_a g_b$, (2) the polynomiality of $C_{1,a}^2$. Polynomiality is a consequence of the separability of the principal wavefunction. We can reformulate the hidden constraint in the form ($a+b\leq k$)
	\begin{equation}
		\chi^{(k+b)}(1^{\times a+b}, 0^{\times k-a})
		=
		\chi(1^{\times b+1}, 0^{\times k-b})
		\prod_{s=1}^{a-1}
		\frac{\chi(1^{\times s+b+1}, 0^{\times k-s-b})}{
			\chi(1^{\times s+1}, 0^{\times k-s})
		}
		\label{eq_nontriv}
	\end{equation}
	These puts some non-trivial constraints on the  $(k+b)$-factors $\chi^{(k+b)}$. It is unclear whether these constraints are also a result of separability. If they are not, then $\mathbb{Z}_k^{(r)}$-WFs have even more structure than originally anticipated. This raises the question: Suppose $\chi$ is the minimal polynomial of $\mathbb{Z}_k^{(r)}$-algebra for some $t$. Is it possible to build a separable $(k,4)$-clustering uniform sequence of wavefunctions that is different from the $\mathbb{Z}_k^{(r)}$ sequence but has minimal polynomial $\chi$? If such a sequence exists, then the Hamiltonian $H_\chi$ will not have a unique ground state.
\end{rmk}

\subsection{Binary Sextics ($2r=6$) and Principal $(k,3)$-Wavefunctions}
\label{sec_sextic}
\begin{table}[t]
	\centering
	\begin{tabular}{|c|c|c|c|c|}
		\hline
		\textbf{Invariant} & \Large$I_2$ & \Large$I_4$ & \Large$I_{6}$ & \Large$I_{10}$\\
		\hline
		\parbox[c]{1.3cm}{\centering\textbf{Graph}} &
		\parbox[c]{2cm}{\centering\includegraphics[scale=.8]{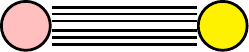}}&
		\parbox[c]{2.5cm}{\centering\includegraphics[scale=.8]{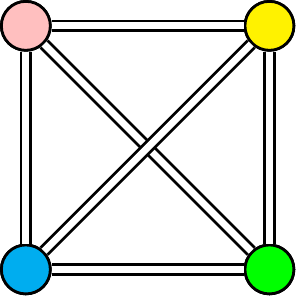}}&
		\parbox[c]{3.2cm}{\centering\includegraphics[scale=.8]{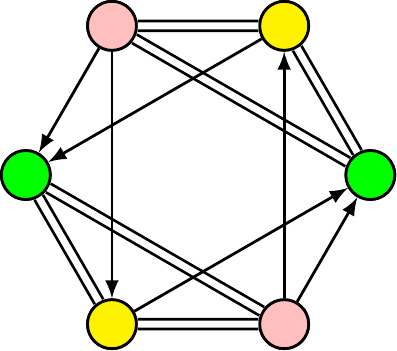}}&
		\parbox[c]{3.6cm}{\centering\includegraphics[scale=.8]{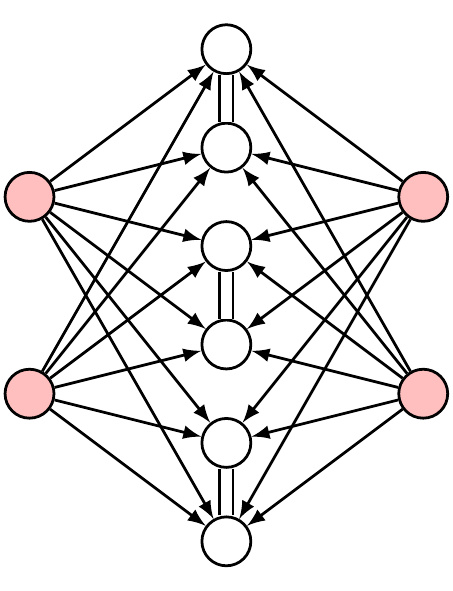}}
		\\
		\hline
		\Large $u$ & \large$2$ & \large$24$& \large$6$ & \large$8$ 
		\\
		\hline
		\Large $m$ & \large$2$ & \large$4$& \large$3$ & \large$1$ 
		\\
		\hline
		\Large $K$ & \large$E_1$ & \large$0$ & \large$E_3$ & \large$E_2E_3$
		\\
		\hline
		\Large $\varepsilon$ & \large$0$ & \large$1$ & \large$1$ & \large$1$
		\\
		\hline
		\parbox[c]{1.3cm}{\centering\textbf{Residue}} & 
		N/A&
		\parbox[c]{2.5cm}{\centering\includegraphics[scale=.8]{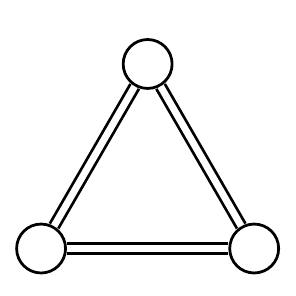}}&
		\parbox[c]{3.2cm}{\centering\includegraphics[scale=.8]{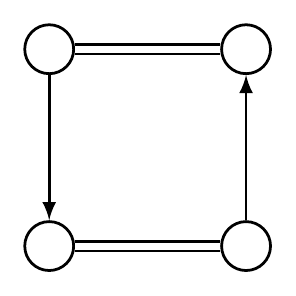}}&
		\parbox[c]{3.6cm}{\centering\includegraphics[scale=.8]{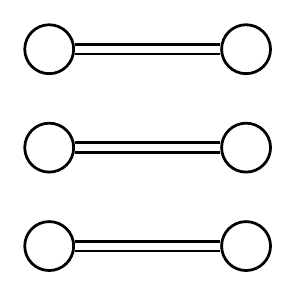}}\\
		\hline
	\end{tabular}
	\caption{A graphic Hilbert basis for the algebra of invariants of binary sextics. For each graph, its $u$-value, $m$-value, co-independence number $\varepsilon$, and its image under $K$ is listed. The residue for graphs with $\varepsilon= 1$ is included as well. The colorings in the first row represent all possible maximum independent sets (the number of colors is the $m$-value). Note that removing each color leads to the graph residue.}
	\label{tab:sextic}
\end{table}
As discussed previously, invariants of binary sextics have a Hilbert basis $I_2, I_4, I_6, I_{10}$ (hsop) together with an extra invariant $J_{15}$. The degree-15 invariant $J_{15}$ squares to a polynomial in $I_2,I_4,I_6, I_{10}$. Since we are only interested in invariants with even degrees, we may ignore $J_{15}$. We have gathered our choice for the basic graphs in Table \eqref{tab:sextic}. 

\paragraph{Minimal Polynomial} Although one can check the linear independence of $\chi_2, \chi_4, \chi_6, \chi_{10}$ directly, we prove this by computing the specializations $\xi^d_a\equiv \chi_d(1, x^{\times a}, 0^{\times k-a})$. 
\begin{subequations}
	\begin{align}
		\xi^2_a(x)&=\frac{(a+1)(k-a)}{k}+\frac{a(k-a)}{k}\left(\frac{(x-1)^6}{k-a}+x^6-1\right)\\
		\xi^4_a(x)&=\frac{a(k-a)}{k^2(k-1)}x^2(x-1)^2\\
		\xi_{a}^6(x) &=\frac{(a+1)_2^-(k-a)_2^-}{k(k)_3^-}x^3+\frac{(a)^-_2(k-a)^-_2}{k(k)_3^-}x^3\left(\frac{2(x-1)^3}{k-a-1}+x^3-1\right)\\
		\xi_{a}^{10}(x)&= \frac{(a+1)_3^-(k-a)}{k(k)_5^-}x^4 +\frac{(a)_3^-(k-a)_3^-}{k(k)_5^-}x^4\left(\frac{3(x-1)^2}{k-a-2}+x^2-1\right)
	\end{align}
\end{subequations}
Let us also write
\begin{equation}
	\label{eq_mini_sextic}
	\chi = \beta_2\chi_2 + \beta_4 \chi_4 + \beta_6 \chi_6+\beta_{10}\chi_{10}
\end{equation}
for some coefficients $\beta_d$. To check linear independence, we put $\chi\to 0$. Now $\xi_a^{2}, \xi_{a}^{4}, \xi_{a}^{6}$ and $\xi_{a}^{10}$ are respectively $O(1), O(x^2), O(x^3)$ and $O(x^4)$. This immediately shows linear independence. However, the true advantage of the residual basis is that it grants us a certain degree of interpretability for the coefficients $\beta_d$. To illustrate, since $\xi_{a}^2$ is the only specialization surviving when $a=0$ (in fact, $\xi_0^2=1$) the coefficient $\beta_2$ is controlled by the normalization of $\chi$ (namely, $\beta_2=1$). Realizing $\xi_a(x)=\chi(1, x^{\times a}, 0^{\times k-a})$ as a conformal block, one interprets $\beta_4, \beta_6$ to be controlled by the central charge and the $W_3$-weight of $\psi$ respectively. Finally, note that $\xi_a^{10}=0$ when $a=1$. Using the polynomiality of chiral weight, we can write the $W_4$-weight of $\psi_a$ in the form
\[
h^{W_4}_a = \frac{a(k-a)}{k-3}\left[-\frac{(a-2)(k-a-2)}{(k-1)}w_1 
+
\frac{(a-1)(k-a-1)}{2(k-2)} w_2
\right]
\]
where $w_1, w_2$ are the $W_4$-weights of $\psi, \psi_2$ respectively. The coefficient $\beta_{10}$ is controlled by $w_2$. However, the better way to interpret $\beta_6, \beta_{10}$ would be in terms of structure constant parameters $t_1, t_2$. Note that $\xi_a(1)=C_{1,a}^2$. We have $\xi_{a}^{4}(1)=0$; i.e. $\beta_4$ is a priori independent of $t_1, t_2$. The coefficients $\beta_6, \beta_{10}$ are just an alternative parametrization to $t_1, t_2$. In other words, while $\beta_4$ is controlled by the central charge, the $W_3$ and $W_4$ weights are controlled by the central charge $c$ and the structure constant parameters $t_1, t_2$. As this analysis shows, the vectors $\chi_{2}, \chi_{4}, \chi_{6}, \chi_{10}$ of the residual basis have (to some degree) specific roles in the language of CFT.

\paragraph{Tricoloring Homomorphism} Using the ansatz \eqref{eq_ansatz}, we can write the $(k,3)$-wavefunction in the form
\begin{equation}
	\Psi = \frac{1}{k!^2}\sum_{\nu_2+2\nu_4+3\nu_6+5\nu_{10}=k} \alpha(\nu_4, \nu_6, \nu_{10})
	\Psi[G_2^{\sqcup\nu_2}\sqcup G_4^{\sqcup\nu_4}\sqcup G_6^{\sqcup\nu_6}\sqcup G_{10}^{\sqcup\nu_{10}}]
\end{equation}
In terms of ansatz coefficients, the $\beta$-coefficients of the minimal polynomial $\chi$ \eqref{eq_mini_sextic} are $\beta_4=\alpha(1,0,0)$, $\beta_6=\alpha(0,1,0)$ and $\beta_{10}=(0,0,1)$. The normalization also forces $\alpha(0,0,0)=1$. As summarized in Table \ref{tab:sextic}, we have $K(G_2)=E_1, K(G_4)=0, K(G_6)=E_3$ and $K(G_{10})=E_2E_3$. The tricoloring identity \eqref{tricoloring_identity} now yields
\begin{equation}
	\label{K_sextic_eq}
	\sum_{\substack{a,b,c\geq 0\\
			a+b+c=k}}
	\binom{k}{a,b,c}
	\prod_{i=1,2}\frac{(t_i)_{k-a}(t_i)_{k-b}(t_i)_{k-c}}{(t_i)_{a}(t_i)_{b}(t_i)_{c}(t_i)_k}
	X^{a}Y^b Z^c=
	\sum_{\nu_2+3\nu_6+5\nu_{10}=k}
	\alpha(0, \nu_6, \nu_{10})
	E_1^{\nu_2}E_2^{\nu_{10}}E_3^{\nu_6+\nu_{10}}
\end{equation}
Although we do not have a closed formula, this identity can be used to compute the constants $\alpha(0, \nu_6, \nu_{10})$ in terms of $t_1, t_2$. Unfortunately, however, except for $\alpha(1,0,0)$, we do not know of a general formalism that allows us to compute $\alpha(\nu_4, \nu_6, \nu_{10})$ with $\nu_4\neq 0$.

\begin{rmk}
	The machinery that allows us to extend our knowledge of $\alpha(0,1,0)$ and $\alpha(0,0,1)$ is the tricoloring identity. Tricoloring identity takes advantage of the constraint $C_{a,b}=g_{a+b}/g_ag_b$. As discussed in Remark \ref{remark_2}, this is highly non-trivial condition in terms of $(k+b)$-factors. In other words, this constraint reveals some hidden structure in the principal $\mathbb{Z}_k^{(r)}$-polynomial. If we did not know about this structure, we would not have had much success in analysing the principal $\mathbb{Z}_k^{(2)}$ wavefunctions either. It is possible that $\mathbb{Z}_k^{(r)}$-algebra/wavefunctions, at least when $r\geq 3$, possess extra structure/hidden constraints that would allow us to strengthen the method presented in this paper and, for example, let us compute all of the constants $\alpha(\nu_4, \nu_6, \nu_{10})$.
\end{rmk}

\begin{rmk}
	A possible approach for finding the previously unknown ansatz constants (at least a few more of them) is through conformal blocks similar to $\xi_a(x)$. Define
	\begin{equation}
		\begin{aligned}
			\xi_{a,b}(x)&=\lim_{w\to \infty}w^{-2r(k-b)}\Psi(w^{\times(k-b)}, 1^{\times b}, x^{\times a}, 0^{\times k-a})\\
			&=\avg{\psi_b^+(\infty)\psi_b(1)\psi_a(x)\psi_a^+(0)}(1-x)^{ab\gamma}x^{2h_a}\\
			&=(1-x)^{ab\gamma}\left[1+\sum_{s\geq 2}\left(\sum_{i}\widehat{h}_a^ih^i_b\delta_{s, \Delta_i}\right)\right](-x)^s{}_2F_1(s,s;2s;x)
		\end{aligned}
	\end{equation}
	The idea is to use the chiral weights' polynomialities to reduce this block's unknowns to a ``handful'' of constants. Unfortunately, this method comes with severe limitations. Aside from the fact that a lot of tedious computations are involved in this method, not knowing any details about the chiral algebras is restrictive as well. For one, at the very least, we need to know how many simple fields exist at each level. Equally problematic is that to compute the weight of a qNOP of two simple fields (neither being the energy-momentum tensor), some limited knowledge of the Lie algebra of the modes is required.
\end{rmk}

\begin{table}
	\centering
	\begin{tabular}{|c|c|c|c|c|}
		\hline
		\textbf{Invariant} & \Large$I_2$ & \Large$I_3$ & \Large$I_4$ &  \Large$I_5$\\
		
		\hline
		\parbox[c]{1.1cm}{\centering \textbf{Graph}}  &\parbox[c]{2.5cm}{\centering\includegraphics[scale=.75]{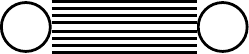}} & \parbox[c]{3cm}{\centering\includegraphics[scale=.75]{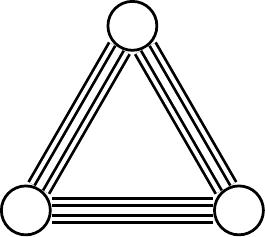}} &
		\parbox[c]{3cm}{\centering\includegraphics[scale=.75]{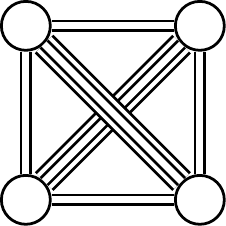}} &
		\parbox[c]{3cm}{\centering\includegraphics[scale=.75]{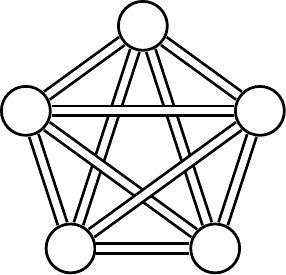}}
		\\
		\hline
		\Large$u$ & \large$2$ & \large$6$ & \large$16$ & \large$120$ 
		\\
		\hline
		\Large$\varepsilon$ & \large$0$ & \large$1/2$ & \large$1$  & \large$3/2$
		\\
		\hline
		\Large$K$ & \large$E_1$ & \large$\sqrt{E_3}$ & \large$0$ & \large$0$\\
		\hline
	\end{tabular}
	
	\begin{tabular}{|c|c|c|c|}
		\hline
		\textbf{Invariant} &\Large$I_6$ & \Large$I_7$ & \Large$J_8$\\
		\hline
		\parbox[c]{1.1cm}{\centering \textbf{Graph}} 
		&\parbox[c]{2.5cm}{\centering \includegraphics[scale=.75]{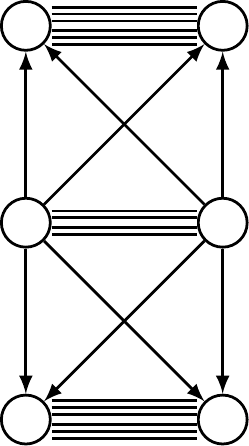}}
		&\parbox[c]{4.4cm}{\centering\includegraphics[scale=.75]{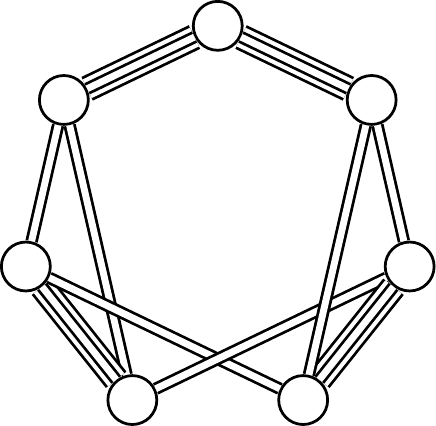}} & \parbox[c]{5.04cm}{\centering \includegraphics[scale=.75]{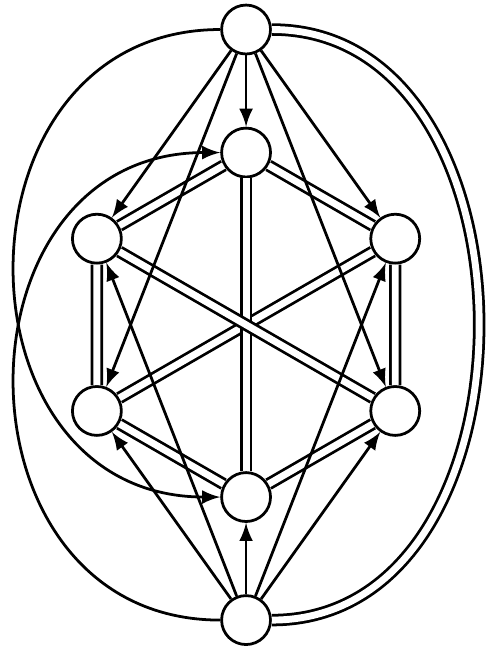}} 
		\\
		\hline
		\Large$u$ & \large$8$ & \large$8$ & \large$4$
		\\
		\hline
		\Large$\varepsilon$ & \large$1$ & \large$1/2$ & \large$1$
		\\
		\hline
		\Large$K$ & \large$0$ & \large$E_2\sqrt{E_3}$ & \large $0$ \\
		\hline
	\end{tabular}

	\begin{tabular}{|c|c|c|}
		\hline
		\textbf{Invariant} & \Large $J_9$ & \Large $J_{10}$\\
		\hline
		\parbox[c]{1.1cm}{\centering \textbf{Graph}}  &\parbox[c]{6.2cm}{\centering\includegraphics[scale=.75]{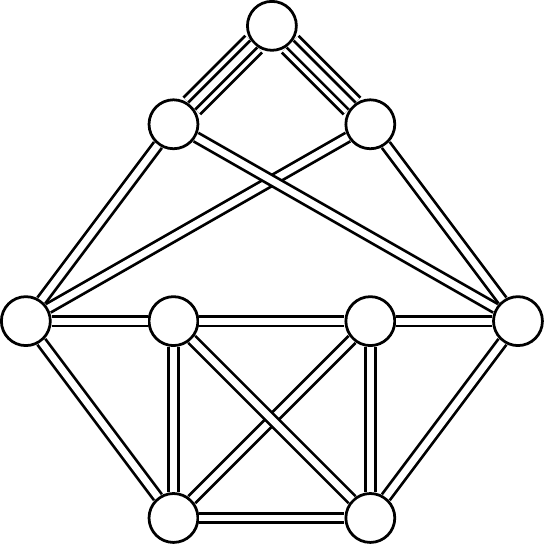}} & \parbox[c]{6.17cm}{\centering\includegraphics[scale=.75]{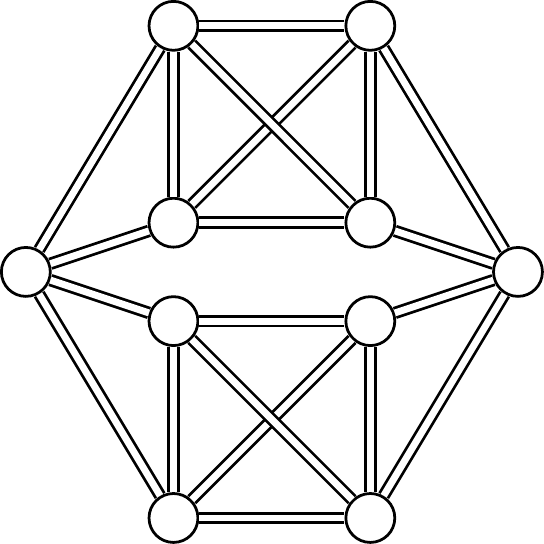}} 
		\\
		\hline
		\Large$u$ & \large$8$ & \large$32$ 
		\\
		\hline
		\Large$\varepsilon$ & \large$3/2$ & \large$2$ 
		\\
		\hline
		\Large$K$ & \large$0$ & \large$0$ \\
		\hline
	\end{tabular}
	\caption{A graphic Hilbert basis for the algebra of invariants of binary octavics. For each graph, its $u$-value, its co-independence number $\varepsilon$, and its image under $K$ are also listed. One reads from this table that $I_5, J_9, J_{10}$ have $\varepsilon>1$ and therefore will not contribute to the residual basis. One also reads that the only basic invariants with $K\neq 0$ are $I_2, I_3$ and $I_7$.}
	\label{tab:octavics_hilbert}
\end{table}
\subsection{Binary Octavics ($2r=8$) \& Principal $(k,4)$ Wavefunctions}
\label{sec_octavic}
Our inability to compute the ansatz constants gets even worse in the case of $\mathbb{Z}_k^{(4)}$-algebras. Despite this, the brief analysis of binary octavics brings to light new oddities we should expect for $r\geq 4$. $\mathcal{B}_8$ has an hsop consisting of invariants $I_2,I_3,I_4,I_5,I_6, I_7$. In addition to the hsop, the Hilbert basis contains invariants $J_8,J_9, J_{10}$. The algebra $\mathcal{B}_8$ can be written as
\begin{equation}
	\mathcal{B}_8 = (1+J_8+J_9+J_{10}+J_8 J_{10})\:\mathbb{C}[I_2, I_3, I_4, I_5, I_5, I_6, I_7]
\end{equation}
This is the first instance that we encounter a situation where we need to keep track of the $J$-invariants appearing on a Hilbert basis. This will be a common sight in $r\geq 4$. The graphic representation of these invariants can be found in Table \eqref{tab:octavics_hilbert}. The residual basis is described in Fig. \ref{fig:octavic_res}.   A general invariant $F$ of even degree $2k$ can now be written in the form
\begin{equation}
	\begin{aligned}
		k!^2 F&=
		\sum_{2\nu_2+4\nu_4+6\nu_{3,3}+6\nu_6+14\nu_{7,7}=2k}
		\alpha^{(0)}_{\nu}I_2^{\nu_2}I_4^{\nu_4}I_{3,3}^{\nu_{3,3}}I_{6}^{\nu_6}I_{7,7}^{\nu_{7,7}}\\
		&+J_8
		\sum_{2\nu_2+4\nu_4+6\nu_{3,3}+6\nu_6+14\nu_{7,7}=2k-8}
		\alpha^{(1)}_{\nu}I_2^{\nu_2}I_4^{\nu_4}I_{3,3}^{\nu_{3,3}}I_{6}^{\nu_6}I_{7,7}^{\nu_{7,7}}\\
		&+I_{3,7}
		\sum_{2\nu_2+4\nu_4+6\nu_{3,3}+6\nu_6+14\nu_{7,7}=2k-10}
		\alpha^{(2)}_{\nu}I_2^{\nu_2}I_4^{\nu_4}I_{3,3}^{\nu_{3,3}}I_{6}^{\nu_6}I_{7,7}^{\nu_{7,7}}+\cdots
	\end{aligned}
\end{equation}
where the ellipsis involves basic invariants that have a co-independence number greater than $1$ and $K$-value zero. Using the tricoloring identity \eqref{tricoloring_identity} the constants $\alpha^{(0)}_{\nu}, \alpha^{(2)}_{\nu}$ with $\nu_4=\nu_5=0$ can be obtained in terms of $t_1, t_2, t_3$. Other $\alpha$-constants can be obtained by computing the specialization $\xi_a(x)$ (and relating the surviving $\alpha$-constants to the central charge and chiral weights). However, many ansatz constants remain undetermined. In particular, since the basic invariants, $I_5, J_9, J_{10}$ have a co-independence number greater than one and a $K$-value zero, we cannot say anything about the $\alpha$-constant of any component involving one of these three invariants.

\begin{rmk}
	Is it perhaps possible that the principal $\mathbb{Z}_k^{(4)}$-wavefunctions have no component involving $I_5, J_9, J_{10}$? In other words, in building the ansatz \eqref{eq_ansatz}, do we need to involve \emph{all} of the basis invariants? Let us briefly discuss the $\mathcal{B}_{10}$-algebra to expand on this question. The invariants of binary decimics has an hsop with degrees $2,4,6,6,8,9,10,14$ \cite{brouwer2010invariants} and a total of 104 basic invariants \cite[p. 40]{olverinv}. If the principal wavefunction were to involve every single one of these invariants, then we would expect $\mathbb{Z}_k^{(5)}$-algebra to have, give or take, $104$ degrees of freedom! We certainly do not expect this to be the case. If our expectation is true, then the ansatz \eqref{eq_ansatz} has to be improved, and we possibly do not need some of the basic invariants.
\end{rmk}

\begin{figure}
	\centering
	\includegraphics{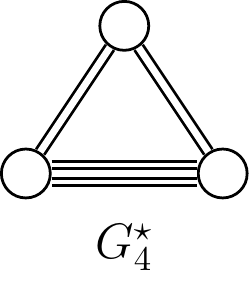}\hspace{15pt}
	\includegraphics{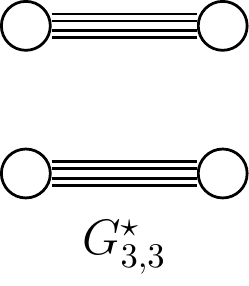}\hspace{15pt}
	\includegraphics{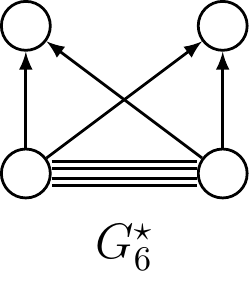}\hspace{15pt}
	\includegraphics{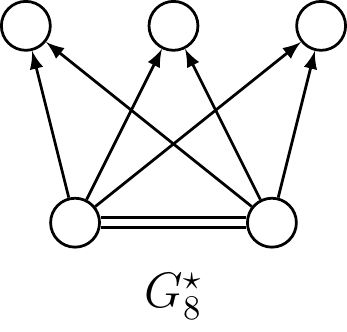}
	
	\includegraphics{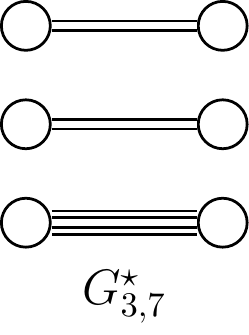}\hspace{15pt}
	\includegraphics{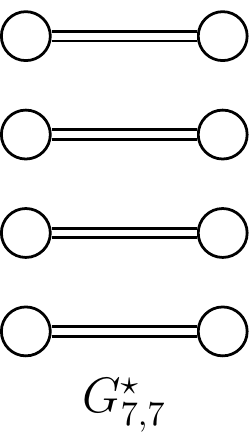}
	
	\caption{The residue of the graphs representing $I_2, I_{4}, I_{3,3}, I_{6}, J_8, I_{3,7}, I_{7,7}$.}
	\label{fig:octavic_res}
\end{figure}

\paragraph{Acknowledgments} I am thankful for insightful discussions with my research supervisor Xiao-Gang Wen. This work is supported by NSF FRG grant:  DMS-1664412.

\appendix
\section{Scaling Dimensions of Zamolodchikov-Fateev Parafermions}
\label{appendix_pattern}
For the purposes of this appendix, we will work with a more general parafermionic CFT. We still have $k$ primary fields $1, \psi, \psi_2, \cdots, \psi_{k-1}$, together with cyclic fusion rules $\psi_a\star \psi_b=\psi_{a+b}$ and charge conjugation symmetry $\psi^+_a=\psi_{k-a}$. However, we allow the scaling dimensions $h_a=h(\psi_a)$ to be general. In other words, we are looking for the most general $h_a$ that would be consistent with an associative solution of the following OPA
\begin{subequations}
	\begin{align}
		&\psi_a(z)\psi_b(w)(z-w)^{\gamma(a,b)}=C_{a,b}\psi_{a+b}(w)+\cdots &(0< a+b<k)\\
		&\psi^+_{a}(z)\psi^+_{b}(w)(z-w)^{\gamma(a,b)}=C_{a,b}\psi^+_{a+b}(w)+\cdots &(0< a+b<k)\\
		&\psi_a(z)\psi_a^+(w)(z-w)^{2h_a}= 1+ 
		{\textstyle\frac{2h_a}{c}}(z-w)^2L(w)+\cdots
	\end{align}
\end{subequations}
Here, $\gamma(a,b)=h_{a}+h_b-h_{a+b}$ is the relative scaling dimension. 
Now, in studying the local properties of wavefunctions, we often find ourselves in need of multi-parafermion OPEs of the form
\begin{equation}
	\Phi_{\vv{a}}(\vv{z})\equiv
	\psi_{a_1}(z_1)\psi_{a_2}(z_2)\cdots \psi_{a_\ell}(z_\ell) \prod_{i<j}(z_i-z_j)^{\gamma(a_i, a_j)}
\end{equation}
where $\vv{a}=(a_1, \cdots, a_\ell)$ is some sequence of integers $0\leq a_i<k$. Due to the simple fusions $\psi_a\star \psi_b=\psi_{a+b}$, the result of $\Phi_{\vv{a}}$ fusion would be $\psi_{|\vv{a}|}$ (where $|\vv{a}|=a_1+\cdots+a_\ell$) and its $\mathcal{W}$-descendants. Therefore, we have
\begin{equation}
	\Phi_{\vv{a}}(\vv{z})= \mathcal{F}_{\vv{a}}(z_1, \cdots, z_\ell)\psi_{a_1+\cdots+a_\ell}\left(\widehat{z}\right)+\cdots
\end{equation} 
Here, $\widehat{z}$ is the weighted center-of-mass $\widehat{z}=(a_1 z_1+\cdots+a_\ell z_\ell)/|\vv{a}|$. With this choice of $\widehat{z}$, the function $\mathcal{F}_{\vv{a}}(z_1, \cdots, z_\ell)$
is a translational invariant homogeneous polynomial. In what follows, we will calculate the degrees $\Delta(\vv{a})=\deg \mathcal{F}_{\vv{a}}$ and relate them to the scaling dimensions $h_a$ and pattern of zeros \cite{Pattern_of_Zeros}.

Before analyzing $\Delta(\vv{a})$, let's make a basic observation. Recall that, with $\gamma=2h-h_2$, the wavefunctions are defined via
\[
\Psi_n (z_1, \cdots, z_{nk})=\frac{1}{g_k^n}\avg{\psi(z_1)\cdots\psi(z_{nk})}\prod_{i<j}(z_i-z_j)^{\gamma}
\]
Let $f(z)=(az+b)/(cz+d)$ be a M\"{o}bius map ($ad-bc\neq 0$). Then, using the transformations of correlation functions
\[
\Psi_n (f(z_1), \cdots, f(z_{nk}))=\prod_{i=1}^{nk}[f'(z_i)]^{N_\phi/2}
\Psi_n (z_1, \cdots, z_{nk})
\]
with $N=nk$ and $N_\phi=N\gamma - (2h+\gamma)$. Thus,  $\gamma$ is the inverse of the ``filling fraction'' (in quotation because we have not proven that $\gamma$ is a rational number yet), and the shift is $\mathcal{S}=2h+\gamma$. Let $d\equiv k\gamma$ and $\mathcal{S}\equiv d-\sigma$. Now since $\Psi_n$ is a polynomial for all $n$, we need $N_\phi=nd-\mathcal{S}$ to be non-negative integer for all $n$ (choose $f(z)=-1/z$ to see this). Thus, we find that $d,\sigma$ are both integers and $0\leq \sigma\leq d$. In particular, $\gamma=d/k$ is rational. Rearranging $2h+\gamma = d-\sigma$, we find
\begin{equation}
	h = \frac{k-1}{2}\gamma - \frac{\sigma}{2}
	\label{eq:app_h}
\end{equation}
We will come back to this equation shortly.

We now begin our analysis of $\Delta(\vv{a})$. Given a sequence $\vv{a}=(a_1, \cdots, a_\ell)$ with $|\vv{a}|=pk+q$ (throughout $p\geq 0$ and $0\leq q<k$), define the \emph{augmentation} of $\vv{a}$  as $\vv{a}'=(a_0=k-q, \vv{a})$. Note that $|\vv{a}'|$ is always a multiple of $k$. Using the multi-parafermion OPE, we find
\begin{equation}
	\mathcal{F}_{\vv{a}'}(w, \vv{z})=\avg{\psi^+_q(w)\Phi_{\vv{a}}(\vv{z})}\prod_{i=1}^\ell (w-z_i)^{\gamma(a_0, a_i)}=
	\frac{\prod_{i=1}^\ell (w-z_i)^{\gamma(a_0, a_i)}}{(w-\widehat{z})^{2h_q}}
	\mathcal{F}_{\vv{a}}(\vv{z})
\end{equation}
Since $\mathcal{F}_{\vv{a}'}$ is a polynomial in $w$, and because it is homogeneous in all of its variables, by looking at its leading diverging term as $w\to \infty$ (on the right-hand-side), we find that
\begin{equation}
	\label{eq:degF2}
	\deg \mathcal{F}_{\vv{a}'} = \sum_{i=1}^\ell \gamma(a_0, a_i)-2h_{|\vv{a}|}+\Delta(\vv{a})
\end{equation}
At the same time, since $\mathcal{F}_{\vv{a}'}(w, \vv{z})$ is a correlation function, using the transformation of Virasoro primary fields under scaling, we find that
\begin{equation}
	\label{eq:degF1}
	\deg \mathcal{F}_{\vv{a}'}= \sum_{i=1}^\ell \gamma(a_0,a_i)+\sum_{1\leq i<j\leq \ell}\gamma(a_i,a_j)
	- h_{|\vv{a}|}-\sum_{i=1}^\ell h_{a_i}
\end{equation}
Combining \eqref{eq:degF1} and \eqref{eq:degF2} yields:
\begin{equation}
	\Delta(\vv{a})=h_{|\vv{a}|}-\sum_{i=1}^\ell h_{a_i}+
	\sum_{1\leq i<j\leq \ell}\gamma(a_i,a_j)\in \mathbb{N}
	\label{eq:Delta}
\end{equation}
The quantity on the right-hand-side is a non-negative integer for all sequences $\vv{a}$. This serves as a constraint on the scaling dimensions.

For now, let us limit ourselves to studying \eqref{eq:Delta} when $\vv{a}=(1^{\times{a}})$. We use the notation $S_a\equiv \Delta(1^{\times a})$. The infinite sequence $S=(S_1, S_2, S_3, \cdots)$ is the pattern of zeros. We now compute $h_a$ in terms of $S_a$. Combining \eqref{eq:app_h} and \eqref{eq:Delta}, we have
\begin{equation}
S_a = h_a - ah+\frac{a(a-1)}{2}\gamma = 
h_a - a\left(
\frac{k-1}{2}\gamma - \frac{\sigma}{2}
\right)+\frac{a(a-1)}{2}\gamma
\end{equation}
First of all, note that $d=k\gamma=S_{k+1}-S_k$. Moreover,
rearranging yields a formula for $h_a$
\begin{equation}
	\label{eq:scaling_dim}
	h_a =\frac{a(k-a)}{2}\gamma-\frac{a\sigma}{2}+S_a
\end{equation}
This formula is supposed to work for any $a$ (not just $0\leq a<k$). There is a series of consistency relations on the pattern of zeros $S$. The determination of scaling dimensions \eqref{eq:scaling_dim} is therefore reduced to finding a ``good'' sequence $S$. We now list the consistency conditions. Most of these conditions, albeit with a different method, have been found a long time ago by Wen \& Wang \cite{Pattern_of_Zeros}. 

\paragraph{First Two Elements} It can be seen from the definition of $h$ that $S_1=0$. Moreover, since by definition $\gamma = \gamma(1,1)=2h-h_2$, a quick check shows that $S_2=0$ as well.

\paragraph{Periodicity} Let $a=pk+q$ with $0\leq q$. Since we have $h_{a}=h_q$, we find
\begin{subequations}
	\begin{align}
		S_{pk+q}&=
		pS_k+S_{q}+\mathcal{S}_k(a)d\\
		\mathcal{S}_k(a)&=\frac{p(p-1)}{2}k+pq
		\label{eq:curlyS}
	\end{align}
\end{subequations}
 In other words, knowing $(\frac{d}{k};S_1=S_2=0, S_3, \cdots, S_k)$ is enough to determine the entire pattern of zeros $S$. In the Wen-Wang terminology, this is referred to as `$k$-cluster condition'.

\paragraph{Charge Conjugation} Let $0\leq a\leq k$ (with convention $h_0=S_0=0$). We know that $h_{k-a}=h_a$ due to duality of $\psi_a$ and $\psi_{k-a}$. Therefore, we must have $S_{k-a}-S_a=S_k-a\sigma$. The condition $\sigma=2S_k/k\in \mathbb{N}$ is the $a=k$ special case of this condition.

\paragraph{Positivity of $\Delta(\vv{a})$ and Convexity} Now that we have obtained the scaling dimensions $h_a$, we can recast the formula for $\Delta(\vv{a})$ \eqref{eq:Delta} in terms of the pattern of zeroes:
\begin{equation}
	\Delta(\vv{a})=S_{|\vv{a}|}-\sum_{I<J}S_{a_I+a_J}+(\ell-2)\sum_{I}S_{a_I}\in \mathbb{N}
\end{equation}
where $\ell$ is the length of $\vv{a}$. The first non-trivial case appears when $\vv{a}=(a,b,c)$. This yields
\begin{equation}
	\begin{aligned}
		\Delta(a,b,c)&=h_{a+b+c}-(h_{a+b}+h_{a+c}+h_{b+c})+h_a+h_b+h_c\\
		&=S_{a+b+c}-(S_{a+b}+S_{a+c}+S_{b+c})+S_a+S_b+S_c\geq 0
	\end{aligned}
\end{equation}
In particular, using $S_1=S_2=0$, we find
\begin{equation}
	\Delta(a-1,1,1)=S_{a+1}-2S_a+S_{a-1}\geq 0
\end{equation}
In other words, $S$ is a \emph{convex sequence}. In particular, with $l_a=S_{a}-S_{a-1}$, we find $l_a\geq l_{a-1}$. Since $l_2=S_2-S_1=0$, we find that $l_a\geq 0$ for all $a\geq 2$. In other words, $S_a$ is a non-decreasing sequence.

\paragraph{Positivity of $D(a,b)$} We can also recast $D(a,b)$ in terms of pattern of zeros:
\begin{equation}
	D(a,b)=ab\gamma - \gamma(a,b)=S_{a+b}-S_a-S_b
\end{equation}
Note that $\Delta(a,b,c)=D(a,b+c)-D(a,b)-D(a,c)$. Moreover $D(a,1)=l_{a+1}\geq 0$.  Thus, $D(a,b+1)\geq D(a,b)+l_{a+1}\geq D(a,b)$. A simple induction now shows that $D(a,b)\geq 0$ for any pair $a,b$.

\paragraph{Intermission: Semi-Locality Conditions} Recall that the commutation factor $\mu_{a,b}$ is defined via the semi-locality condition
\[
\psi_a(z)\psi_b(w)(z-w)^{\gamma(a,b)}=
\mu_{a,b}\psi_b(w)\psi_a(z)(w-z)^{\gamma(a,b)}
\]
In this appendix we will show that $\mu_{a,b}=(-1)^{D(a,b)}$ for a general $\mathbb{Z}_k$ parafermionic theory. Consider the following reformulation of the semi-locality
\[
\psi_a(Z+\zeta/2)
\psi_b(Z-\zeta/2)\zeta^{\gamma(a,b)}=
\mu_{a,b}
\psi_b(Z-\zeta/2)
\psi_a(Z+\zeta/2)
(-\zeta)^{\gamma(a,b)}
\]
Taking the limit $\zeta\to 0$, this equality leads to $C_{a,b}=\mu_{a,b}C_{b,a}$, where $C_{a,b}$ is defined through $\psi_a(z)\psi_b(w)(z-w)^{\gamma(a,b)}=C_{a,b}\psi_{a+b}(w)+\cdots$. This immediately shows that $\mu_{a,a}=1$ (we prohibit $C_{a,a}=0$). Now, consider the (regular) $3$-field
\[
\begin{aligned}
	R(x,y,z) &= \psi_a(x)\psi_b(y)\psi_c(z)(x-y)^{\gamma(a,b)}(x-z)^{\gamma(a,c)}(y-z)^{\gamma(b,c)}\\
	&=\mu_{a,b}\mu_{a,c}
	\psi_b(y)\psi_c(z)\psi_a(x)(y-x)^{\gamma(a,b)}(z-x)^{\gamma(a,c)}(y-z)^{\gamma(b,c)}
\end{aligned}
\]
with $a,b,c$ arbitrary. Taking the limit $y\to z$ we find
\[
\begin{aligned}
	\lim_{y\to z}R(x,y,z) &= C_{b,c}\psi_{a}(x)\psi_{b+c}(z)
	(x-z)^{\gamma(a,b+c)}
	(x-z)^{\Delta(a,b,c)}\\
	&=C_{b,c}\mu_{a,b}\mu_{a,c}\psi_{b+c}(z)\psi_a(x) (z-x)^{\gamma(a,b+c)}(z-x)^{\Delta(a,b,c)}
\end{aligned}
\]
Defining $\widetilde{\mu}_{a,b}=(-1)^{D(a,b)}\mu_{a,b}$ and using $\Delta(a,b,c)=D(a,b+c)-D(a,b)-D(a,c)$ we find
\begin{equation}
\widetilde{\mu}_{a,b+c}=\widetilde{\mu}_{a,b}\widetilde{\mu}_{a,c}
\end{equation}
Furthermore, since $D(1,1)=S_2-2S_1=0$, we have $\widetilde{\mu}_{1,1}=1$. Choosing $a=c=1$, the above identity inductively leads to $\mu_{1,a}=1$. At the same time, since $\mu_{b,a}=1/\mu_{a,b}$, we also find $\widetilde{\mu}_{b+c,a}=\widetilde{\mu}_{b,a}\widetilde{\mu}_{c,a}$. Choosing $c=1$ in the latter identity, we inductively find that $\widetilde{\mu}_{a,b}=1$ for all $a,b$. In other words, $\mu_{a,b}=(-1)^{D(a,b)}$. In particular, since the identity commutes with everything, we have $1=\mu_{1,k}=(-1)^{D(1,k)}$. In other words, $d=S_{k+1}-S_k=D(1,k)$ has to be even. Throughout the main text, we have used the notation $d=2r$. Here, however, we will keep working with $d$.

\paragraph{Evenness of $D(a,a)$} The fact that $\mu_{a,a}=1$, together with the general formula $\mu_{a,a}=(-1)^{D(a,a)}$, forces $D(a,a)=S_{2a}-2S_a$ to be even. In other words, $S_{2a}$ needs to be even. Now, since $d=S_{k+1}-S_k$ is even, and either $k$ or $k+1$ is even, we find that both $S_k, S_{k+1}$ are even.

\paragraph{Total Angular Momenta on Sphere} As far as transformations of $\mathcal{F}_{\vv{a}'}$ under M\"{o}bius maps go, so far, we have discussed translations and scaling. To complete the discussion on the effects of conformal symmetry, we also need to consider the M\"{o}bius transformation $z\to -1/z$.  Under this transformation we find
\begin{equation}
	\label{inversion}
	\mathcal{F}_{\vv{a}'}(z_0, z_1, \cdots, z_\ell)=
	\prod_{I=0}^\ell z_I^{\delta_{I}(\vv{a})}
	\mathcal{F}_{\vv{a}'}\left(-\frac{1}{z_0}, -\frac{1}{z_1}, \cdots, -\frac{1}{z_\ell}\right)
\end{equation}
where
\begin{equation}
	\delta_I(\mathbf{a}) = \sum_{J\neq I}\gamma(a_I, a_J)-2h(a_I)
	\label{eq:delta_constraint_eq}
\end{equation}
Since $\mathcal{F}_{\vv{a}'}$ is a translational invariant polynomial, the quantity $\delta_I(\vv{a})$  is nothing but the degree of $z_I$ in $\mathcal{F}_{\vv{a}'}$. Thus, we must have $\delta_I(\mathbf{a})\geq 0$. To our knowledge, this condition has not been reported in the literature. Recasting this condition in terms of pattern of zeros, we get a reformulation: Let $\vv{b}$ be any sequence such that $|\vv{b}|=nk-a$ for some $n$ and $a$. Then
\begin{equation}
	\sum_{i}D(a,b_i)\leq D(a, nk-a)
\end{equation}
The quantity $D(a,nk-a)$ is $2J_a$, where $J_a=\frac{1}{2}aN_\phi(n)-S_a$ is the maximum possible total angular momentum of a bound state of $a$ particles (in a system of size $n$). Here $N_\phi(n)=(n-1)d+\sigma$ is the number of flux quanta.

\section{Basic Definitions of Chiral Algebra}
\label{appendix_chiral}
\paragraph{Quasiprimary Fields and Field-Space $\mathcal{W}$} Let $L(z)$ be the (chiral) energy-momentum tensor in a conformal field theory. This is the conserved current corresponding to the conformal symmetry. We use the convention $L(z)=\sum_n z^{-n-2} L_n$ for the mode expansion of $L(z)$. The modes of $L(z)$ generate a copy of the Virasoro algebra $\mathfrak{Vir}$; i.e.
\begin{equation}
	[L_m, L_n] = (m-n)L_{m+n}+\frac{m(m-1)(m+1)}{12}c\:\delta_{m+n,0}
\end{equation}
The triple $L_{-1}, L_0, L_1$ generated the global conformal transformations. A holomorphic field $\phi(z)$ is \emph{quasi-primary} with \emph{scaling dimension} $\Delta\in \mathbb{N}$ if
\begin{equation}
	[L_n, \phi(z)]=\Delta (n+1)z^n\phi(z)+ z^{n+1}\partial\phi(z)
\end{equation}
for $n=-1,0,1$, and \emph{Virasoro primary} if the commutation relations hold for all $n\geq -1$. The space of all \emph{holomorphic quasi-primary fields with integer scaling dimension} in a CFT is denoted by $\mathcal{W}$. Both identity and $L$ belong to $\mathcal{W}$. In general, the conserved current associated to any continuous symmetry of the system lies in $\mathcal{W}$. We call $\mathcal{W}$ the field space of the \emph{chiral (symmetry) algebra}. Throughout this section, let $I$ be an index set and $\{W^i\}_{i\in I}$  a basis for $\mathcal{W}$.

\paragraph{OPE Ansatz of Quasi-primary Fields} Here we will review the general OPE ansatz of quasi-primary fields \cite{nahm}. Given any triple of holomorphic quasi-primary fields $\phi^i, \phi^j, \phi^k$ (not necessarily with integer dimensions), we define $\Delta(ijk)=\Delta_i+\Delta_j-\Delta_k$. Furthermore, $d_{ij}$ and $C_{ijk}$ are defined through the following correlation functions:
\begin{equation}
	\avg{\phi^i(z_i)\phi^j(z_j)}\equiv\frac{d_{ij}\delta_{\Delta_i, \Delta_j}}{z_{ij}^{2\Delta_i}}, \qquad
	\avg{\phi^i(z_i)\phi^j(z_j)\phi^k(z_k)}\equiv \frac{C_{ijk}}{z_{ij}^{\Delta(ijk)}
		z_{jk}^{\Delta(kji)}
		z_{ik}^{\Delta(kij)}
	}
\end{equation}
where $z_{ij}\equiv z_i-z_j$. With these notations, the OPE of two quasi-primary fields has the general form
\begin{equation}
	\phi^i(z)\phi^j(w)=\sum_{k}C_{ij}^k (z-w)^{-\Delta(ijk)}
	\sum_{n\geq 0}\frac{(z-w)^n}{n!}\frac{(\Delta(kij))_n}{(2\Delta_k)_n}\partial^n \phi^k(w)
	\label{OPE_Nahm}
\end{equation}
where $C_{ijk}:=C_{ij}^ld_{lk}$ (Einstein summation is used), the $k$ sum is over all quasi-primaries in the CFT, and $(x)_a=x(x+1)\cdots(x+a-1)$ is the Pochhammer symbol.

\paragraph{Chiral Algebra} Let $W$ stand for a field in $\mathcal{W}$ with scaling dimension $\Delta$. The mode expansion of $W$ is as follows: $W(z)=\sum_n z^{-n-\Delta}W_n$. Let $W^i$ be a basis for $\mathcal{W}$ and $\Delta_i$ the scaling dimension of $W^i$. The OPEs \eqref{OPE_Nahm} endow the modes of $\mathcal{W}$ with the structure of a infinite dimensional Lie algebra:
\begin{align}
	[W^i_m, W^j_n]&=\sum_{k}
	\theta(\Delta^{k}_{ij})\,
	C_{ij}^k
	\,
	\mathfrak{P}_{\Delta_i,\Delta_j}^{\Delta_k}(m,n)
	\,
	W^k_{m+n}
	+
	d_{ij}\delta_{m,-n}\binom{m+\Delta_i-1}{2\Delta_i-1}
	\\
	\mathfrak{P}_{\Delta_i,\Delta_j}^{\Delta_k}(m,n)&=
	\sum_{r+s=\Delta(ijk)-1}
	\frac{(-1)^{s}}{r!s!}\frac{
		(\Delta(kji)-2)_r
		(m+1-\Delta_i)_r
		(\Delta(kij)-2)_s
		(n+1-\Delta_j)_s}{(2\Delta_k)_{r+s}}
\end{align}
where $k$ runs over the all the basis fields $W^k$ in $\mathcal{W}$. Here, $\theta(x)$ is the step function. The field space, together with this Lie structure, is called the \emph{chiral algebra}. Assuming the field space is fully known, once the structure constants $d_{ij}$ and $C_{ij}^k$ are determined, the chiral algebra is fixed.

\paragraph{Triangular Decomposition of $\mathcal{W}$}  Let $\mathcal{W}^{\pm}=\mathrm{span}\{W_n^i\mid i\in I,  \pm n>0\}$ and $\mathcal{W}^{0}=\mathrm{span}\{W_0^i\mid i\in I\}$. We have a vector space decomposition $\mathcal{W}=\mathcal{W}^-\oplus \mathcal{W}^0\oplus \mathcal{W}^+$. Let $\mathcal{W}_\circ$ denote the maximal abelian subalgebra of $\mathcal{W}$. Then $\mathcal{W}_\circ\subseteq \mathcal{W}^0$ and $L_0\in \mathcal{W}_\circ$. Moreover, there exists subspaces $\mathcal{W}_{\pm}\subseteq \mathcal{W}^0\oplus \mathcal{W}^{\pm}$ such that $\mathcal{W}=\mathcal{W}_-\oplus \mathcal{W}_\circ \oplus \mathcal{W}_+$. This is a triangular decomposition of $\mathcal{W}$, i.e., $[\mathcal{W}_{\pm}, \mathcal{W}_{\circ}\oplus \mathcal{W}_{\pm}]\subseteq \mathcal{W}_{\pm}$.

\paragraph{State-Field Correspondence} Let $\ket{0}$ stand for the vacuum state. If $W\in \mathcal{W}$ has scaling dimension $\Delta$, then for all $n>-\Delta$ we have $W_n\ket{0}=0$. The state corresponding to $W(z)$ is defined as $\ket{W}=\lim_{z\to 0}W(z)\ket{0}=W_{-\Delta}\ket{0}$, while the field $W(z)$ corresponding to state $\ket{W}$ is the vertex-operator $W(z)=Y(\ket{W};z)$. The \emph{state-field correspondence} allows us to think of $\mathcal{W}$ as quasi-primary descendants of identity. Generalizing this idea, for any holomorphic field $\phi(z)$ in the CFT (not necessarily with integer scaling dimension), one defines the asymptotic state $\ket{\phi}$ as $\ket{\phi}=\lim_{z\to 0}\phi(z)\ket{0}$, where $\ket{0}$ is the vacuum state. The field corresponding to the state $\ket{\phi}$ is defined as a vertex operator $\phi(z)=Y(\ket{\phi},z)$.

\paragraph{Highest Weight Representations} A representation $\mathcal{H}_\phi$ of the Lie algebra $\mathcal{W}$ is called \emph{highest weight} if there exists unique vector $\ket{\phi}\in M$ and a functional $\mathfrak{h}_\phi: \mathcal{W}_\circ\to \mathbb{C}$ such that
\begin{enumerate}
\item $\mathcal{W}_+\ket{\phi}=0$, i.e. for any operator $O\in \mathcal{W}_+$ we have $O\ket{\phi}=0$.
\item The abelian subalgebra $\mathcal{W}_\circ$ acts diagonally on $\ket{\phi}$ via $\mathfrak{h}_\phi$. Concretely, let $W\in \mathcal{W}$ be such that $W_0\in \mathcal{W}_\circ$. Then $W_0\ket{\phi}=h_\phi^W\ket{\phi}$. The value $h_\phi^W$ is called the \emph{$W$-weight} of $\phi$
\end{enumerate}
One may construct a basis for $\mathcal{H}_{\phi}$ by repeated application of operators in $\mathcal{W}_-$ to $\ket {\phi}$. We say the field $\phi(z)$ corresponding to $\ket{\phi}$ is $\mathcal{W}$-primary. All other fields corresponding to vectors in $\mathcal{H}_{\phi}$ are a $\mathcal{W}$-descendant of $\phi(z)$. Any $\mathcal{W}$-primary field $\phi(z)$ is Virasoro primary with scaling dimension being the $L$-weight of $\phi$.

\paragraph{qNOPs \& Simple Fields}
For purposes of fractional quantum Hall effect (rational CFTs), we are only interested in chiral algebras $\mathcal{W}$ that are \emph{finitely generated}.
To explain what this means, we need to define the \emph{quasi-primary normal order product (qNOP)} \cite{nahm}. Going back to the OPE \eqref{OPE_Nahm}, we define $\mathcal{N}(\partial^s W^i, W^j)$ via
\begin{equation}
	\begin{aligned}
		W^i(z)W^j(w)&=\sum_{s=-(\Delta_i+\Delta_j)}^{-1}
		(z-w)^{s}
		\sum_k\delta_{-s, \Delta(ijk)}
		C_{ij}^k
		\sum_{n\geq 0}\frac{(z-w)^n}{n!}\frac{(2\Delta_i+s)_n}{(2\Delta_i+2\Delta_j + 2s)_n}\partial^n W^k(w)
		\\
		&+\sum_{s\geq 0}
		\frac{(z-w)^s}{s!}
		\sum_{n\geq 0} \frac{(z-w)^{n}}{n!} \frac{(2\Delta_i+s)_n}{(2\Delta_i+2\Delta_j + 2s)_n}\partial^n \mathcal{N}(\partial^s W^i, W^j)(w)
	\end{aligned}
\end{equation}
In general, $\mathcal{N}(\partial^s W^i, W^j)=(\partial^s W^i, W^j)+\cdots$, where $(\partial^s W^i, W^j)$ is the usual normal order product (NOP). The qNOP is a modification of NOP so that $\mathcal{N}(\partial^s W^i, W^j)$ is quasi-primary. A field of the form $\mathcal{N}(A,B)\in \mathcal{W}$ will be called a \emph{qNOP}. A field $W\in \mathcal{W}$ is called \emph{simple} if $d_{WX}=0$ for any qNOP field $X$. The only simple field that is not Virasoro primary is the energy-momentum tensor $L$. It can be shown that the commutation relations of the chiral algebra are fixed once the commutation between the simple fields are fixed. We say $\mathcal{W}$ is finitely generated if it has a finite number of simple fields.

\paragraph{Chiral Algebra of Parafermionic Theories} Let $\mathcal{W}$ be the chiral algebra of a parafermionic current algebra. In other words, $\mathcal{W}$ consists of mutually local bosonic conserved currents that describe the extended symmetry. This paragraph aims to point out some simple facts about $\mathcal{W}$. Firstly, recall that part of the definition of the parafermionic theories is the OPEs of the following form
\begin{equation}
	\psi_a(z)\psi_a^+(w)(z-w)^{2h_a} = 1 + \frac{2h_a}{c}(z-w)^2 L(w) + O(z-w)^3
\end{equation}
The fields in the chiral algebra $\mathcal{W}$ are the quasi-primary fields appearing in the above OPEs (that are not null). Looking closely at the OPE, there is no conserved current at dimension one, and the only quasi-primary at dimension two is the energy-momentum tensor. 

\paragraph{Chiral Algebra of $\mathbb{Z}_k^{(r)}\times U(1)$} Let $\phi$ be the free boson of a $U(1)$ theory and $\mathcal{W}_{\text{full}}$ be the chiral algebra of $\mathbb{Z}_k^{(r)}\times U(1)$. With $\mathcal{W}$ being the chiral algebra of the parafermionic part, we have $\mathcal{W}\subset \mathcal{W}_{\text{full}}$ a quotient of $\mathcal{W}_{\text{full}}$.  In addition, $\mathcal{W}_{\text{full}}$ contains (at least) five Virasoro primary chiral fields: $j(z)=i\partial \phi(z)$ with dimension one, the `electron' $V_e$ and `anti-electron' $V_e^+$ have scaling dimension $r$
\begin{equation}
V_e(z) =\psi(z) :\exp \left(i \sqrt{2r/k}\;\phi(z)\right):, \quad
V^+_e(z) =\psi^+(z) :\exp \left(-i \sqrt{2r/k}\;\phi(z)\right):
\end{equation}
and $\Gamma^{\pm}(z)=:\exp [\pm i \sqrt{2kr}\;\phi(z)]:$ with dimension $kr$. The three fields $j, \Gamma^{\pm}$ are the generators of the chiral algebra of the $U(1)$ theory. Let us also mention that, when $r=1$, the fields $j, V_e, V_e^+$ generate $\widehat{\mathfrak{su}}(2)_k$. In fact, Zamolodchikov-Fateev \cite{Zamolodchikov_Fateev_Parafermion} have shown that $\mathbb{Z}_k^{(1)}=\widehat{\mathfrak{su}}(2)_k/U(1)$.

\section{Multi-Parafermion OPE Ansatz}
\label{appendix_OPE}
In this appendix we will develop an ansatz for the full expansion of the $a$-field ($a=pk+q$, $p\geq 0$, $0<q\leq k$)
\begin{equation}
	\Phi_a(z_1, \cdots, z_a)=\frac{1}{g_k^pg_q}\psi(z_1)\cdots\psi(z_a)\prod_{i<j}(z_i-z_j)^{\gamma}
\end{equation}
Our method will be a direct generalization of the ansatz Nahm et. al. \cite{nahm} introduced for the OPE of quasi-primaries: Given two quasi-primary fields $\phi_i, \phi_j$ with scaling dimensions $d_i, d_j$, the OPE $\phi_i\phi_j$ is of the following form
\begin{equation}
	\phi_i(z)\phi_j(w)=\sum_{k}C_{ij}^k(z-w)^{d_k-d_i-d_j}\sum_{n\geq 0}\frac{(z-w)^{n}}{n!}\frac{(d_k+d_i-d_j)_n}{(2d)_n}\partial^n \phi_k(w)
\end{equation}
with $k$ running over all quasi-primary fields in the theory. The determination of the OPE is then reduced to finding the structure constants $C_{ij}^k$.

Let $\mathcal{Q}(q;s)$ be the space of quasi-primary $\mathcal{W}$-descendants of $\psi_q$, with scaling dimension $h_q+s$. Suppose $\{\psi_{q;s;i}\}$ is a basis for $\mathcal{Q}(q;s)$. Note that the charge conjugation symmetry $\psi_q\leftrightarrow \psi_q^+=\psi_{k-q}$ provides an isomorphism 
$\mathcal{Q}(q;s)\simeq \mathcal{Q}(k-q;s)$. In the case $k=q$, by convention, we define $\mathcal{Q}(0;s)= \mathcal{Q}(k;s)$ and understand ``charge conjugation'' as the identity map. We define $\psi^+_{q;s;i}\in \mathcal{Q}(k-q;s)$ as the isomorphic image of $\psi_{q;s;i}\in \mathcal{Q}(q;s)$. With these conventions, consider the ansatz
\begin{equation}
	\Phi_a(z_1, \cdots, z_a) = \sum_{s\geq 0}\sum_i
	\sum_{n\geq 0}
	\chi^{(a)}_{n;s;i}(z_1, \cdots, z_{a})\partial^n\psi_{q;s;i}(\widehat{z})
\end{equation}
where, as usual $\widehat{z}=(z_1+\cdots+z_a)/a$, and  $\chi^{(a)}_{n;s;i}(z_1, \cdots, z_{a})$ is a translational invariant symmetric polynomial, homogeneous of degree $2r\mathcal{S}_k(a)+s+n$. Since $\dim\mathcal{Q}(q,s=0)=1$, in that case we omit the indexes $s,i$ and write $\chi_n^{(a)}$. Now define
\begin{subequations}
	\begin{align}
		F^{(a)}_{s;i}(w;z_1, \cdots, z_a)&=\frac{1}{g_k^pg_q}\avg{\psi_{q;s;i}^+(x)\psi(z_1)\cdots\psi(z_a)}\prod_{i<j}(z_i-z_j)^{\gamma}\prod_i (w-z_i)^{(k-q)\gamma +s}\\
		Y^{(a)}_{s;i}(z_1, \cdots, z_a)&=
		\frac{1}{g_k^pg_q}\avg{\psi_{q;s;i}^+(\infty)\psi(z_1)\cdots\psi(z_a)}\prod_{i<j}(z_i-z_j)^{\gamma}
	\end{align}
\end{subequations}
Both $F$ and $Y$-functions are homogeneous translational invariant polynomials in their arguments, and are symmetric in the $z$-variables. For simplicity of notation, let us drop $s,i,a$ for the moment. Using the $f(z)=-1/z$ M\"{o}bius transformation on $F$, we find
\begin{equation}
	F(w;z_1, \cdots, z_a)=w^{d}\prod_{j=1}^a z_j^{2pr+s}
	F\left(-\frac{1}{w};-\frac{1}{z_1}, \cdots, -\frac{1}{z_a}\right)
\end{equation}
with $d=2p(k-q)r+(a-2)s$. In other words, $\deg_w F=d$ and $\deg_{z_i}F=2pr+s$. Moreover, using the scaling dimensions, we find that $F$ is homogeneous of degree
\begin{equation}
	\deg F =\frac{d+a(2pr+s)}{2}
	=
	2r\mathcal{S}_k(a)+s+d
\end{equation}
Now note that $Y^{(a)}=\lim_{w\to \infty}w^{-d}F^{(a)}$. Thus, $\deg Y^{(a)}=\deg F - d = 2r\mathcal{S}_k(a)+s$. We would like to find $F$ from $Y$. To this end, expand $F$ in powers of $w$
\[
F(w,\vv{z})=\sum_{m=0}^{d}\frac{w^{d-m}}{m!}Y_m(\vv{z})
\]
for some polynomials $Y_m$. Define the differential operators $\nabla := -w^2\partial_w + dw+\ell_+$, and $\ell_+:=-D_++(2pr+s)(z_1+\cdots+z_a)$, and $D_+=\sum_i z_i^2\partial_i$. Using conformal Ward identities, it is straightforward to show $\nabla F=0$. This then results in the recursive relations
\begin{equation}
	Y_{m+1}=-\ell_+Y_m, \qquad \ell_+Y_d=0
\end{equation}
Combined with the fact that $Y_0\equiv Y$, we find that
\begin{equation}
	\begin{aligned}
		F(w;z_1, \cdots, z_a)&=w^d\exp\left[-\frac{\ell_+}{w}\right]Y(\vv{z})=w^d\prod_{j=1}^a \left(1-\frac{z_j}{w}\right)^{2pr+s}\exp\left[\frac{D_+}{w}\right]Y(z_1, \cdots, z_a)\\
		&=w^d\prod_{j=1}^a\left(1-\frac{z_j}{w}\right)^{2pr+s}Y\left(\frac{z_1}{1-\frac{z_1}{w}},\cdots, \frac{z_a}{1-\frac{z_a}{w}}\right)\\
		&=\prod_{j=1}^a (w-z_j)^{2pr+s}Y\left(\frac{1}{w-z_1}, \cdots, \frac{1}{w-z_a}\right)
	\end{aligned}
\end{equation}
To see the validity of the above identities, let $P(x)$ be a one-variable polynomial of degree $m$. Then, it is straightforward to see that
\begin{equation}
	\exp(\tau x^2d/dx)P(x) = P\left(\frac{x}{1-\tau x}\right), \qquad 
	\exp(\tau [mx-x^2d/dx])P(x) = (1+\tau x)^m P\left(\frac{x}{1+\tau x}\right)
\end{equation}
The identities we have used are $N$-variable generalizations of these relations.

In order to relate $F,Y$ functions to the polynomials in the ansatz of $\Phi_a$, we need one more transformation; namely
\begin{equation}
	F_{s;i}(w; \vv{z})=F_{s;i}(w-\widehat{z}; \vv{z}-\widehat{z}) = (w-\widehat{z})^d\prod_{j=1}^a \left(1-\frac{z_j-\widehat{z}}{w-\widehat{z}}\right)^{2pr+s}
	\sum_{l\geq 0}
	\frac{1}{l!}\frac{(D_+^l Y_{s;i})(\vv{z}-\widehat{z})}{(w-\widehat{z})^l}
\end{equation}
Now, by using the ansatz in the original definition of $F_{s;i}$, we have
\begin{align*}
	F^{(a)}_{s;i}(w;z_1, \cdots, z_a)&=\frac{1}{g_k^pg_q}\avg{\psi_{q;s;i}^+(x)\psi(z_1)\cdots\psi(z_a)}\prod_{i<j}(z_i-z_j)^{\gamma}\prod_i (w-z_i)^{(k-q)\gamma +s}\\
	&=
	\prod_i (w-z_i)^{(k-q)\gamma +s}
	\sum_{t\geq 0}\sum_j
	\sum_{n\geq 0}
	\chi^{(a)}_{n;t;j}(\vv{z})\partial^n
	\avg{\psi_{q;s;i}^+(w)\psi_{q;t;j}(\widehat{z})}\\
	&=
	(w-\widehat{z})^{d}
	\prod_i \left(1-\frac{z_i-\widehat{z}}{w-\widehat{z}}\right)^{(k-q)\gamma +s}
	\sum_j d^{[q;s]}_{ij}
	\sum_{n\geq 0}
	\chi^{(a)}_{n;t;j}(\vv{z})
	\frac{(2h_q+2s)_n}{(w-\widehat{z})^{n}}
\end{align*}
where we have defined $d^{[q;s]}_{ij}=\avg{\psi_{q;s;i}^+(\infty)\psi_{q;t;j}(0)}$. To finish the job, note that
\[
\prod_{i}\left(1-\tau x_i\right)^{-d}= \sum_{m\geq 0}\tau^m \sum_{|\lambda|=m}\frac{(d)_{\lambda}}{\lambda!}m_{\lambda}(x_i)
\]
where, the second sum is over all partitions of $m$. For a partition $\lambda=(\lambda_1, \lambda_2, \cdots)$, we have defined $\lambda!=\prod_{j}\lambda_j!$ and $(d)_{\lambda}=\prod_j (d)_{\lambda_j}$. Also, $m_{\lambda}$ are the symmetric monomials. Putting everything together, we find
\begin{equation}
	\sum_j d^{[q;s]}_{ij}
	\chi^{(a)}_{n;s;j}(\vv{z})
	=
	\frac{1}{(2h_q+2s)_n}
	\sum_{l+|\lambda|=n}\frac{(2r-a\gamma)_\lambda}{\lambda!l!}m_{\lambda}(\vv{z}-\widehat{z})(D^l_+Y^{(a)}_{s;i})(\vv{z}-\widehat{z})
\end{equation}
In other words, determining $Y^{(a)}_{s;i}\in \mathcal{T}_a^{2r\mathcal{S}_k(a)+s}$ fixes the entire OPE. The polynomials $Y^{(a)}_{s;i}$ play a similar role the three point correlations play in the case of two-point OPEs. Determination of $Y^{(a)}_{s;i}$ itself often boils down to relating its specializations to other data of the CFT.

\section{The Field Corresponding to the State $\ket{\partial;W,a}$}
\label{appendix_null}
Let $W\in \mathcal{W}$ have spin $s$, $\ket{a}=\lim_{z\to 0}\psi_a(z)\ket{0}$ and define $\ket{\partial; W,a}=N(2W_{-1}L_0 - sL_{-1}W_0)\ket{a}$, with $N$ is an appropriate normalization. Let us symbolically denote the field by $\partial_W\psi_a$. Note that since $W$ is quasi-primary $[L_1, W_n]:=(s-n-1)W_{n+1}$ and 
$[L_0, W_n]:=-nW_{n}$. Thus,
in particular, $[L_1, W_{-1}]:=sW_{0}$ and $[L_0, W_0]=0$. This can be used to show that $L_1\ket{\partial;W,a}=0$. Furthermore, $L_2\ket{\partial;W,a}$ lies in $\mathcal{W}_+\ket{a}$ which is zero since $\ket{a}$ is $\mathcal{W}$-primary. This means $\partial_W\psi_a$ is Virasoro primary. We want to show $\partial_W\psi_a$ is (at least) ``effectively'' null; i.e. $\partial_W\psi_a$ decouples from the parafermions. Now, either $\ket{\partial;W,a}=0$ (in which case we are done), or $N$ can be chosen so that
\[
\braket{\partial; W,-a}{\partial; W,a}=1
\]
This normalization amounts to $\avg{\partial_W\psi_a^+(\infty)\,\partial_W\psi_a(0)}\equiv d^W_{-a,a}=1$. Let $a+b\leq k$ and consider the block
\[
\rho_{a,b}(x)=\avg{\psi_b^+(\infty)\psi_a^+(1)\psi_a(x)\psi_b(0)}(1-x)^{2h_a}x^{ab\gamma}
\]
To compute this we, first make the observation that
\[
\avg{\psi_a(x)\psi_b(y)\partial_W\psi_{a+b}^+(z)}=\frac{C_{a,b}^{\partial_W \psi_{a+b}}}{(x-y)^{\gamma(a,b)-1}(x-z)^{\gamma(-a,a+b)+1}(y-z)^{\gamma(-b,a+b)+1}}
\]
Taking the complex conjugate of the above relation, with $\psi_a^\dagger(x)=\bar{x}^{-2h_a}\psi_a(1/\bar{x})$, we find
\[
C_{a,b}^{\partial_W\psi_{a+b}} C_{-b,-a}^{\partial_W\psi^+_{a+b}} = 
C_{a,b}^{\partial_W\psi_{a+b}} C_{b,a}^{\partial_W\psi_{a+b}} = |C_{a,b}^{\partial_W\psi_{a+b}}|^2
\]
Using this, and computing $\rho_{a,b}(x)$ to first order in $x$, it can be shown that
\[
\rho_{a,b}(x) = C_{a,b}^2 - \frac{2rab}{a+b}C_{a,b}^2 x + x\sum_{k} |C_{a,b}^{\partial_{W_k}\psi_{a+b}}|^2+O(x^2)
\]
where $k$ goes over a basis of $\mathcal{W}$.
\begin{lem}
$\rho_{a,b}(x)=C_{a,b}^2(1-2rab/(a+b)x)+O(x^2)$. Consequently, $C_{a,b}^{\partial_{W_k}\psi_{a+b}}=0$. The latter means that none of $\partial_W\psi_{a}$ can appear in the OPEs between the parafermions $\psi_a$.
\end{lem}
\begin{proof}
Recall that
\[
\chi^{(k+b)}(z_1, \cdots, z_{k+b})=\frac{1}{g_k g_b}\avg{\psi_b^+(\infty)\psi(z_1)\cdots \psi(z_{k+b})}\prod_{i<j}(z_i-z_j)^{\gamma}
\]
is a semi-invariant of degree $S_{k+b}=2rb$. As per usual, define
\[
P_n(z_1, \cdots, z_{k+b}) = \sum_{i=1}^{k+b}\left(z_i - \widehat{z}\right)^n
\]
with $\widehat{z}=(z_1+\cdots+z_{k+b})/(k+b)$. The polynomials $P_\lambda$ with $\lambda=(\lambda_1, \cdots, \lambda_{\ell})$ with $\lambda_1\leq k+b$ and $\lambda_\ell>1$ and $|\lambda|=2kb$ make a basis for $\mathcal{T}_{k+b}^{2rb}$. Therefore, we have
\[
\chi^{(k+b)}=\sum_{\lambda} \alpha_\lambda P_\lambda
\]
for some $\alpha_\lambda$. At the same time, $\rho_{a,b}(x) = \chi^{(k+b)}(1^{\times k-a}, x^{\times a}, 0^{\times b})$. We will use the properties of $P_n$ to compute $\rho_{a,b}(x)$ up to first order in $x$. Firstly,
\begin{align*}
	P_n(1^{\times k-a}, x^{\times a}, 0^{\times b}) &= (k-a)\left(1-\frac{k-a+ax}{k+b}\right)^n
	+
	a\left(x-\frac{k-a+ax}{k+b}\right)^n+
	b\left(-\frac{k-a+ax}{k+b}\right)^n\\
	&=P_n(1^{\times k-a}, 0^{\times a+b})\left(1-\frac{na}{a+b}x\right)+O(x^2)
\end{align*}
Therefore,
\[
\begin{aligned}
	P_\lambda
	(1^{\times k-a}, x^{\times a}, 0^{\times b}) &= 
	P_\lambda
	(1^{\times (k-a)}, 0^{\times a+b})
	\prod_{j=1}^{\ell}\left(1-\frac{\lambda_ja}{a+b}x\right)+O(x^2)\\
	&=
	P_\lambda
	(1^{\times k-a}, 0^{\times a+b})
	\left(1-\frac{2rab}{a+b}x\right)+O(x^2)
\end{aligned}
\]
At the same time,
\[
C_{a,b}^2=
\frac{g_{a+b}}{g_ag_b}\avg{\psi_b^{+}(\infty)\psi_a^+(1)\psi_{a+b}(0)}=
\rho_{a,b}(1^{\times k-a}, 0^{\times a+b})
=\sum_{\lambda}\alpha_\lambda 
P_\lambda
(1^{\times k-a}, 0^{\times a+b})
\]
Consequently,
\[
\rho_{a,b}(x) = 
C_{a,b}^2\left(1 - \frac{2rab}{a+b} x\right)+O(x^2)
\]
\end{proof}
\noindent
If we treat the $\mathbb{Z}_k^{(r)}$ algebra as a subalgebra of a representation of the chiral algebra $\mathcal{W}$, then even though the fields $\partial_W \psi_a$ might not be null for some $W\in \mathcal{W}$, they do not couple to the $\mathbb{Z}_k^{(r)}$. So for all intents and purposes we can treat them as null fields.

\section{Finitely Generated Graded Algebras}
\label{appendix_graded}
\paragraph{Definition} A commutative ring $R$ (with unity) is called an $\mathbb{N}$-\emph{graded ring}  if $R$ is a direct sum $R=\bigoplus_{m\geq 0} R^m$ such that $R^iR^j\subseteq R^{i+j}$. We say $R$ is a \emph{graded $\mathbb{C}$-algebra} if furthermore $R^0=\mathbb{C}$. An element $x\in R^m$ is called \emph{homogeneous of degree $m$}. We say $R$ is finitely generated, if there are finitely many homogeneous elements $T_1, T_2, \cdots, T_N$ in $R$ such that for any element $F\in R$ there exists a polynomial $f\in \mathbb{C}[x_1, \cdots, x_N]$ (not necessarily unique) such that $F = f(T_1, T_2, \cdots, T_N)$. Being finitely generated, in particular, means $\dim_{\mathbb{C}}R^m<\infty$ (dimension as a vector space) for all $m$.

\paragraph{Example} The canonical example of a finitely generated graded $\mathbb{C}$-algebra is $\mathbb{C}[x_1, \cdots, x_N]$. The standard grading is obtained by choosing $\deg x_1=\deg x_2=\cdots=\deg x_N=1$. However, one can choose a different grading for $R$. Let $\mathfrak{d}=(d_1, d_2, \cdots, d_N)$ be a sequence of positive integers (we also assume $d_i\leq d_{i+1}$). We call $\mathfrak{d}$ the degree profile of the grading of $\mathbb{C}[x_1, \cdots, x_N]$ induced by taking $\deg x_i = d_i$. We use the notation $S_N^\mathfrak{d}$ for the algebra $\mathbb{C}[x_1, \cdots, x_g]$ graded via the degree profile $\mathfrak{d}$. The standard grading has $\mathfrak{d}=(1,1,\cdots,1)$. An example with a different $\mathfrak{d}$ would be the ring of symmetric polynomials in $n$ variables: $\Lambda_N=\mathbb{C}[e_1, e_2, \cdots, e_N]$, where $e_i$ is the $i$th elementary symmetric polynomial. Obviously it is more natural to choose $\deg e_i = i$, making $\mathfrak{d}=(1,2,3,\cdots, N)$, and $\Lambda_N = S_N^{\mathfrak{d}}$. 

\paragraph{Convention for Generators} From now on, $R$ is taken to be a finitely generated algebra. We symbolically write the generators of $R$ as $X_1, X_2, \cdots, X_D, Y_1, \cdots Y_g$. It will become clear shortly why we have separated the generators into $X,Y$ sets. The meaning of $D$ and $g$ will be explained in due course as well. We also make the convention that $g=0$ will mean that $R$ has no $Y$-generators. We will use the notation $d_i=\deg X_i$, and  $\mathfrak{d}=(d_1, \cdots, d_D)$. Furthermore, let $\deg Y_j = d'_j$.

\paragraph{Poincaré Series, Krull Dimension \& Multiplicity}
Since $R=\bigoplus_{m}R^m$ is finitely generated, each $R^m$ is a finite dimensional vector space. This allows us to define the formal power series:
\begin{equation}
	P_R(t)=\sum_{n\geq 0} [\dim_{\mathbb{C}}R^m] t^m
\end{equation}
This is called the \emph{Poincaré series} of $R$. A standard result in theory of finitely generated graded algebras (Hilbert-Serre theorem; see for example \cite[Thm. 11.1]{atiyah}) states that there exists an integral polynomial $q_R\in \mathbb{Z}[t]$ so that:
\begin{equation}
	P_R(t)= \frac{p_R(t)}{\prod_{i=1}^D (1-t^{d_i})}, \qquad
	p_R(t)=\frac{q_R(t)}{\prod_{i=1}^g (1-t^{d'_i})}
\end{equation}
We have chosen $D$ so that the rational function $p_R(t)$ is regular at $t=1$. In other words, $D$ is the degree of the pole of $P_R(t)$ at $t=1$. This is called the \emph{Krull dimension} of $R$ (important to note that there are a multitude of equivalent definitions for Krull dimension). The \emph{multiplicity} of $R$ is defined as $\mu = p_R(1)$.

\paragraph{Homogeneous System of Parameters (hsop)}
With $D$ the Krull dimension of $R$, a \emph{homogeneous system of parameters (hsop)} for $R$ is a set of $D$ homogeneous elements $X_1, \cdots, X_D\in R$, such that the quotient $\overline{R}=R/(X_1, X_2, \cdots, X_D)$ is a finite-dimensional $\mathbb{C}$-vector space. It is a standard result that an hsop always exists (if $R$ is finitely generated), and $X_1, X_2, \cdots, X_D$ are \emph{algebraically independent}\footnote{Algebraic independence means that: if $f\in \mathbb{C}[x_1, \cdots, x_D]$ is such that $f(X_1, \cdots, X_D)=0$, then $f=0$.}. Let $l$ be defined through $\dim_{\mathbb{C}}\overline{R}=l+1$. The algebra $R$ is a finite $S_D^{\mathfrak{d}}\equiv \mathbb{C}[X_1, \cdots, X_D]$-module. In fact, there exists a finite number of homogeneous elements $Z_0\equiv 1,Z_1, \cdots, Z_l\in R$ such that any $F\in R$ can be written in the form
\begin{equation}
	F = \sum_{j=0}^l Z_j f_j(X_1, \cdots, X_l)
	\label{eq_decom}
\end{equation}
where $f_j\in \mathbb{C}[x_1, \cdots,x_D]$. In general, $R$ is \emph{not} a free $S_D^{\mathfrak{d}}$-module; i.e. no matter which $Z_i$'s we choose, the decomposition \eqref{eq_decom} is not necessarily unique.

\paragraph{Cohen-Macaulay Property} 
Let $e_i$ be the degree of $Z_i$ (with $e_0=\deg 1 =0$). A finitely generated graded algebra $R$ is \emph{Cohen-Macaulay} if $R$ is a \emph{free} finite $S_D^{\mathfrak{d}}$-module; i.e. the decomposition \eqref{eq_decom} is now unique. More concretely, for any homogeneous $F$ we have
\begin{equation}
	F = \sum_{b=0}^{l}\:
	\sum_{\substack{
			\nu_1, \cdots, \nu_D\geq 0\\
			\nu_1d_1+\cdots +\nu_D d_D = \deg F - e_b}}
	C_{\nu}^{(b)} Z_b \prod_{i=1}^D X_i^{\nu_i}
\end{equation}
where $C_\nu^{(b)}\in \mathbb{C}$. Another consequence of Cohen-Macaulay condition is that Poincaré series is now
\begin{equation}
	P_R(t)=\frac{1+t^{e_1}+\cdots+t^{e_l}}{\prod_{i=1}^D(1-t^{d_i})}
\end{equation}
meaning $p_R(t)=1+t^{e_1}+\cdots+t^{e_l}$. The multiplicity is then $\mu=l+1=\dim \overline{R}$. At this point, the only topic remaining to be discussed is the relation between $Z_b$'s with $b>0$ and $Y_j$'s. The generators $Y_1,Y_2, \cdots, Y_g$ are always an elements of $\{Z_1, \cdots, Z_{l}\}$. Those $Z_b$ not equal to some $Y_j$ (if they exist) can always be reduced to $\prod_{j=1}^g Y_j^{a_j}$ for some integral sequence $a=(a_1, \cdots, a_g)$ satisfying $\sum_j a_j d'_j=e_b$. In other words, $g\leq l=\mu-1$.

\paragraph{About the Special Case $R=\mathcal{B}_n$}
For the special case $R=\mathcal{B}_n$ (we use $D(n)$, $\mu(n)$ for Krull dimension and multiplicity of $\mathcal{B}_n$), Springer has found a closed formula for the Poincaré series \cite{springer1, springer2}:
\begin{equation}
	P_n(t)=\sum_{0\leq j<\frac{n}{2}}(-1)^j \phi_{n-2j}\left(\frac{(1-t^2)t^{j(j+1)}}{(j,t^2)!(n-j, t^2)!}\right)
\end{equation}
where $(d,t)!\equiv (1-t)(1-t^2)\cdots (1-t^d)$ and the operator $\phi_d$ transforms a rational function $f$ in $t$ to a rational function $\phi_df$ according to
\begin{equation}
	(\phi_d f)(t^d)=\frac{1}{d}\sum_{j=1}^d f(e^{2\pi i j/d}t)
\end{equation}
From this relation, it can be shown that $D(1)=1$ and $D(n)=n-2$ for $n\geq 3$.  For $n=2,3,4$ we have $\mu(n)=1$, meaning $\mathcal{B}_n$ is polynomial algebras. When $n>4$, then $\mathcal{B}_n$ is no longer freely generated. Let $d_1, \cdots, d_{n-2}$ be the degrees of the invariants in the hsop. For $n\geq 5$ we have $\mu(n)>1$ and the multiplicity is given by the formula
\begin{equation}
	\frac{\mu(n)}{\prod_{i=1}^{n-2}d_{n-2}} = -\frac{2^{\frac{1}{2}[1+(-1)^n]}}{4n!}\sum_{0\leq e<n/2}(-1)^e \binom{n}{e}\left(\frac{n}{2}-e\right)^{n-3}
\end{equation}

\section{Invariants of Binary Forms in Terms of Roots}
\label{appendix_direct}
In section \ref{sec_iso} we found an isomorphism $\mathcal{B}_n^m\simeq \mathcal{U}_m^n$. In this appendix we explore yet another isomorphism: $\mathcal{B}_n^m\simeq \mathcal{U}_n^m$. First of all, let us define the inverted coefficients 
\begin{equation}
	\vv{a}^\tau =(a^{\tau}_0,a^\tau_1, \cdots, a^\tau_{n-1},a^\tau_n)=(a_n, a_{n-1}, \cdots, a_1, a_n)
\end{equation}
This is the transformation of $\vv{a}$ under $\tau\in \mathbf{GL}_2$ given by $\tau(x,y)^t=(y,x)^t$. In terms of the inverted coefficients, the binary form can be written as
\begin{equation}
	Q_n(X,Y)=\sum_{j=0}^n \binom{n}{j}a^\tau_j X^{n-j}Y^j
\end{equation}
Suppose $a^\tau_0=a_n\neq 0$. Then, using the fundamental theorem of algebra we can write
\begin{equation}
	Q_n(X,Y)=\sum_{j=0}^n \binom{n}{j}a^\tau_j X^{n-j}Y^j=a^\tau_0\prod_{i=1}^n (X-Yz_i)
\end{equation}
where $z_i$ are the roots of the monic polynomial $P(x)=\sum_{j=0}^{n}\binom{n}{j}\frac{a^\tau_j}{a^\tau_0} x^{n-j}$. Let us find how the roots transform under $g\in \mathbf{GL}_2$. We define
\begin{equation}
	g =\begin{pmatrix}
		\alpha& \beta\\
		\gamma& \delta
	\end{pmatrix},\qquad f_g(z) = \frac{\alpha z+\beta}{\gamma z + \delta}
\end{equation}
One can check that under $g\in \mathbf{GL_2}$, we have $z_i\mapsto f_g(z_i)$. Also, note that $f_{\tau}(z)=1/z$.

Let us now relate the roots to coefficients explicitly. The elementary symmetric polynomials $e_j(z_1, \cdots, z_n)$ can be defined through the generating function $\prod_{i=1}^n (1+tz_i)=\sum_{j=0}^n e_j(z_1, \cdots, z_n)t^n$. Therefore, we find
\begin{equation}
	\label{eq_replacement}
	\binom{n}{j}\frac{a^\tau_j}{a^\tau_0}=(-1)^{j} e_{j}(z_1, \cdots, z_n)
\end{equation}
Suppose $I(a_0, a_1, \cdots, a_n)$ is an invariant of binary $n$-ic of degree $m$ and weight $w=nm/2$. We would like to know what $I$ would look like in terms of the roots. First of all, we have $I(\vv{a}^\tau)=(-1)^wI(\vv{a})$. Moreover, using homogeneity and Eq. \eqref{eq_replacement}, we have
\begin{equation}
	\label{eq_temp2}
	I(\vv{a}) = (-1)^{w}[a_0^\tau]^m I(\vv{a}^\tau/a_0^\tau)=
	(-1)^{w}[a_0^\tau]^m P_I(z_1, \cdots, z_n)
\end{equation}
where
\[
P_I(z_1, \cdots, z_n) =
I\left(\cdots, \frac{(-1)^j}{\binom{n}{j}}e_{j}(z_1, \cdots, z_{n}),\cdots\right)
\]
Clearly $P_I$ is a symmetric polynomial of local degree $m$.

Next, let $g\in \mathbf{GL}_2$. We would like to know how $P_I$ transforms under $f_g$. We start by finding the transformation of $a_0^\tau$. After a bit of computation, one finds:
\begin{equation}
	a^\tau_0\mapsto
	a_0^\tau \prod_{i=1}^n \frac{(\gamma z_i +\delta)}{\det g} = a_0^\tau
	(\det g)^{-n/2} \prod_i  \left(\frac{df_g}{dx}\right)^{-1/2}_{x=z_i}
\end{equation}
Combined with Eq. \eqref{eq_temp2}, and $I$ being an invariant of degree $m$, we find that
\[
P_I(f_g(z_1), \cdots, f_g(z_n))=
\prod_i\left(\frac{df_g}{dx}\right)^{m/2}_{x=z_i}
P_I(z_1, \cdots, z_n)
\]
In other words, $P_I$ is an $(n,m)$ uniform state. The mapping $I\mapsto P_I$ is an isomorphism $\mathcal{B}_n^m\simeq \mathcal{U}_n^m$. It is worth mentioning that: the composition of $\mathcal{B}_n^m\simeq \mathcal{U}_n^m$ and the isomorphism 
$\mathcal{B}^n_m\simeq \mathcal{U}_n^m$ discussed in section \ref{sec_binary}, leads to $\mathcal{B}_n^m\simeq \mathcal{B}_m^n$. This is called the \emph{Hermite reciprocity} theorem.

\begin{rmk}
	To find the algebra of invariants of binary quadratics and cubics, all we need to do is to find all uniform states in two and three variables. Starting with quadratics, the only symmetric translational invariant polynomials are
	\[
	P_{m}(z_1, z_2) = (z_1-z_2)^{2m}
	\]
	The polynomial $P_m$ corresponds to the invariant $\Delta^m$, with $\Delta$ being the discriminant. Moving on to the cubic invariants, let $P(z_1, z_2, z_3)$ be a uniform state. Due to conformal invariance (exactly the same way three point correlators are found in CFT), we find that $P$ is necessarily of the form
	\[
	P_m(z_1, z_2, z_3)=(z_1-z_2)^{2m} (z_1-z_3)^{2m}(z_2-z_3)^{2m}
	\]
	for some $m\in \mathbb{N}$. The polynomial $P_m$ corresponding to the invariant $\Delta_3^m$, with $\Delta_3$ being the cubic discriminant.
	
\end{rmk}

\subsection{Principal $\mathbb{Z}_2^{(r)}$-WFs: Paired States}
Previously, the isomorphism $\mathcal{B}_{2r}^{2k}\simeq \mathcal{U}_{2k}^{2r}$, together with our knowledge of $\mathcal{B}_{2r}$, allowed us to construct an ansatz for $\mathbb{Z}_k^{(r)}$-WFs with a fixed $r$. Similarly, we can use the isomorphism $\mathcal{B}^{2r}_{2k}\simeq \mathcal{U}_{2k}^{2r}$, combined with a known basis for $\mathcal{B}_{2k}$, to obtain an ansatz for $\mathbb{Z}_k^{(r)}$-algebra; however, this time we are fixing $k$ but varying $r$. The latter isomorphism $\mathscr{E}:\mathcal{U}_{2k}^{2r}\xrightarrow{\sim} \mathcal{B}_{2k}^{2r}$ can be obtained explicitly as follows:

\begin{enumerate}
\item Let $P(z_1, \cdots, z_{2k})$ be a $(2k,2r)$ uniform state.
\item Find the ``symmetric reduction'' of $P$; i.e. the decomposition of $P$ in the basis of elementary symmetric polynomials:
\[
P(z_1, \cdots, z_{2k}) = \sum_{\lambda} \beta_\lambda e_{\lambda}(z_1, \cdots, z_{2k})
\]
where $\lambda=(\lambda_1, \lambda_2, \cdots, \lambda_{2r})$ are partitions of $2kr$, $e_\lambda=e_{\lambda_1}e_{\lambda_2}\cdots$ and $\beta_\lambda\in \mathbb{C}$.

\item Replace $e_j$ with 	$(-1)^j\binom{2k}{j}a_j/a_0$; i.e.
\[
P'(\vv{a})=
\sum_{\lambda}\binom{2k}{j} \beta_\lambda 
\prod_{i=1}^{2r}\frac{a_{\lambda_j}}{a_{0}}
\]
\item Finally, $\mathscr{E}[P](\vv{a})\equiv a_0^{2r} P'(\vv{a})$ which is a homogeneous polynomial, and in fact an invariant of binary $2k$-ic with degree $2r$. Note that $\mathscr{E}[P_1P_2]=\mathscr{E}[P_1]\mathscr{E}[P_2]$.
\end{enumerate}
In the case of $2k=4$, we know that $\mathcal{B}_4$ is generated by two invariants $I_2, I_3$ of degrees $2,3$. Explicitly, we have
\begin{equation}
I_2 = (\mathscr{E}\circ \Psi)\left[\vcenter{
\hbox{
\includegraphics[scale=.75]{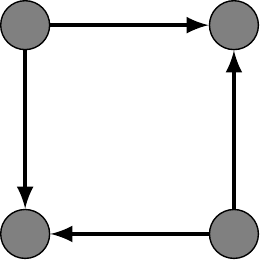}
}
}\right], \qquad
I_3 = (\mathscr{E}\circ \Psi)\left[\vcenter{
	\hbox{
		\includegraphics[scale=.75]{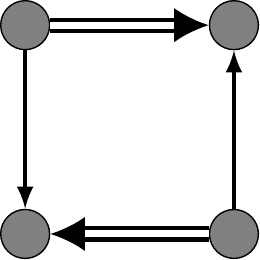}
	}
}\right]
\end{equation}
Thus, the most general invariant of degree $2r$ is of the form
\[
I= \sum_{m=0}^{\lfloor r/3\rfloor}
\alpha_m I_2^{r-3m} I_3^{2m}, \qquad \alpha_m\in \mathbb{C}
\]
Utilizing the isomorphism, our ansatz for the $\mathbb{Z}_2^{(r)}$-WFs is
\begin{equation}
	\label{Eq_Z2r}
\Psi = 
\sum_{m=0}^{\lfloor r/3\rfloor}
\alpha_m \Psi\left[\vcenter{
	\hbox{
		\includegraphics[scale=.75]{shard_pf}
	}
}\right]^{r-3m} \Psi\left[\vcenter{
\hbox{
	\includegraphics[scale=.75]{shard_gf}
}
}\right]^{2m}
\end{equation}
The normalization of $\Psi$ leads to $\sum_{m}\alpha_m=1$. Moreover, by computing $\xi(x)=\lim_{w\to \infty}w^{-2r}\Psi(w, 1,x,0)$, we find the coefficient of $x^2$ leads to
\begin{equation}
\sum_{m=0}^{\lfloor r/3\rfloor} m\alpha_m= \frac{4r}{27}\left(1-\frac{r}{2c}\right)
\end{equation}
In general, we suspect that (other than the normalization constraint) the coefficient $\alpha_m$ are free parameters of the CFT and equation \eqref{Eq_Z2r} is the final solution. In other words, we do not need to determine $\alpha_m$; rather we need to interpret their role as CFT data.

\bibliographystyle{elsarticle-harv}
\bibliography{foo} 

\end{document}